\documentclass[acmsmall,screen,nonacm]{acmart}

\AtBeginDocument{%
  \providecommand\BibTeX{{%
    \normalfont B\kern-0.5em{\scshape i\kern-0.25em b}\kern-0.8em\TeX}}}

\setcopyright{acmcopyright}
\copyrightyear{2024}
\acmYear{2024}
\acmDOI{XXXXXXX.XXXXXXX}

\usepackage{geometry}
\usepackage{fancyhdr}
\usepackage{color}
\usepackage{amsmath}
\usepackage{algorithm}
\usepackage{subcaption}
\usepackage{float}
\usepackage{mathpartir}
\usepackage{stmaryrd}
\usepackage{mathtools}
\usepackage{xcolor}
\usepackage{soul}
\usepackage{syntax}
\usepackage{pifont}
\usepackage{proof}
\usepackage[noend]{algpseudocode}
\usepackage{comment}
\usepackage{listings}
\usepackage{wrapfig}
\usepackage[inline]{enumitem}

\usepackage{pifont}% http://ctan.org/pkg/pifont

\usepackage{etoolbox}
\AtBeginEnvironment{definition}{\vspace{-0.5em}}
\AtEndEnvironment{definition}{\vspace{-0.5em}}
\AtBeginEnvironment{theorem}{\vspace{-0.5em}}
\AtEndEnvironment{theorem}{\vspace{-0.5em}}
\AtBeginEnvironment{minted}{\vspace{-0.5em}}
\AtEndEnvironment{minted}{\vspace{-0.5em}}
\AtBeginEnvironment{align*}{\vspace{-0.5em}}
\AtEndEnvironment{align*}{\vspace{-0.5em}}
\usepackage{titlesec}
\titlespacing*{\section}{0pt}{0.3em}{0.3em}   % Adjusts spacing for \section
\titlespacing*{\subsection}{0pt}{0.2em}{0.2em} % Adjusts spacing for \subsection
\titlespacing*{\subsubsection}{0pt}{0.1em}{0.1em} 
\definecolor{mGreen}{rgb}{0,0.6,0}
\definecolor{mGray}{rgb}{0.5,0.5,0.5}
\definecolor{mPurple}{rgb}{0.58,0,0.82}
\definecolor{backgroundColour}{rgb}{0.95,0.95,0.92}

\lstdefinestyle{CStyle}{
    backgroundcolor=\color{backgroundColour},   
    commentstyle=\color{mGreen},
    keywordstyle=\color{tean},
    numberstyle=\tiny\color{mGray},
    stringstyle=\color{mPurple},
    basicstyle=\footnotesize,
    breakatwhitespace=false,         
    breaklines=true,                 
    captionpos=b,                    
    keepspaces=true,                 
    numbers=left,                    
    numbersep=5pt,                  
    showspaces=false,                
    showstringspaces=false,
    showtabs=false,                  
    tabsize=2,
    language=C
}

\newcommand{\probP}{\text{I\kern-0.15em P}}
\newcommand{\toolname}{\textsc{DLIA}$^2$}
\newcommand{\la}{$\textit{LinearArbitrary}$}
\newcommand{\icedt}{$\textit{ICE-DT}$}

\newcommand{\gspacer}{$\textit{GSpacer}$}

\usepackage[
raggedright
]{sidecap}   
\sidecaptionvpos{figure}{t} 

\newcommand{\dstate}{\mathcal{D}_{\textit{state}}}

\begin{document}

\title{\scalebox{0.92}{Probabilistic Guarantees for Practical LIA Loop Invariant Automation}}

\author{Ashish Kumar}
\affiliation{%
  \institution{The Pennsylvania State University}
  \streetaddress{201 Old Main}
  \city{State College}
  \state{Pennsylvania}
  \country{USA}
  \postcode{16802}}

\author{Jilaun Zhang}
\affiliation{%
  \institution{The Pennsylvania State University}
  \streetaddress{201 Old Main}
  \city{State College}
  \state{Pennsylvania}
  \country{USA}
  \postcode{16802}}

\author{Saeid Tizpaz-Niari}
\affiliation{%
  \institution{UT El Paso}
  \streetaddress{ }
  \city{El Paso}
  \state{Texas}
  \country{USA}
  \postcode{ }}

\author{Gang Tan}
\affiliation{%
  \institution{The Pennsylvania State University}
  \streetaddress{201 Old Main}
  \city{State College}
  \state{Pennsylvania}
  \country{USA}
  \postcode{16802}
}

\renewcommand{\shortauthors}{Kumar et al.}

\begin{abstract}
Despite the crucial need for formal safety and security verification of programs, discovering loop invariants remains a significant challenge. Static analysis is a primary technique for inferring loop invariants but often relies on substantial assumptions about underlying theories. Data-driven methods supported by dynamic analysis and machine learning algorithms have shown impressive performance in inferring loop invariants for some challenging programs. However, state-of-the-art data-driven techniques do not offer theoretical guarantees for finding loop invariants.

We present a novel technique that leverages the simulated annealing (SA) search algorithm combined with SMT solvers and computational geometry to provide probabilistic guarantees for inferring loop invariants using data-driven methods. Our approach enhances the SA search with real analysis to define the search space and employs parallelism to increase the probability of success. To ensure the convergence of our algorithm, we adapt $\epsilon$-nets, a key concept from computational geometry. Our tool, \texttt{\toolname}, implements these algorithms and demonstrates competitive performance against state-of-the-art techniques. We also identify a subclass of programs, on which we outperform the current state-of-the-art tool \gspacer{}.
\end{abstract}

\begin{comment}
% GT: unnessary for a submission
\begin{CCSXML}
<ccs2012>
<concept>
<concept_id>10002978.10002986.10002990</concept_id>
<concept_desc>Security and privacy~Logic and verification</concept_desc>
<concept_significance>500</concept_significance>
</concept>
<concept>
<concept_id>10003752.10010061.10010063</concept_id>
<concept_desc>Theory of computation~Computational geometry</concept_desc>
<concept_significance>500</concept_significance>
</concept>
<concept>
<concept_id>10003752.10010061.10010065</concept_id>
<concept_desc>Theory of computation~Random walks and Markov chains</concept_desc>
<concept_significance>500</concept_significance>
</concept>
</ccs2012>
\end{CCSXML}

\ccsdesc[500]{Security and privacy~Logic and verification}
\ccsdesc[500]{Theory of computation~Computational geometry}
\ccsdesc[500]{Theory of computation~Random walks and Markov chains}
\end{comment}

% \keywords{Loop Invariants, Simulated Annealing, CounterExample Guided Synthesis, Z3 verification, $\epsilon$-nets.}

\maketitle

\section{Introduction.}
\label{sec:intro}
The formal verification of software for safety and security requirements such as memory safety, termination guarantees, and information leaks is crucial to safeguarding critical infrastructures against intentional and unintentional risks.
Since existing unverified software is too large to be fully verified manually, researchers have turned to building tools that automate the verification of safety and security properties. One prevalent method in automated formal
verification is Hoare Logic~\cite{hoare1969axiomatic}: it reduces the problem of verifying the whole program with respect to a given property to verifying each program statement separately with a given pre- and post-condition.

Loops, however, pose significant challenges to Hoare-logic style verifiers since the verification has to supply a special type of proposition called \emph{loop invariants}. 
Intuitively, a loop invariant is a proposition that is \texttt{true} at the beginning of the loop and remains \texttt{true} for each iteration of the loop body. 
It is well known that finding loop invariants is undecidable~\cite{blass2001inadequacy}. However, existing approaches have had success by limiting the invariant search space. 
Previous techniques for finding loop invariants can be classified into static or data-driven, based on whether they use runtime program data. Static analysis techniques find invariants by reasoning over first-order theories. However, such techniques usually work for very restrictive first-order theories. 
Data-driven techniques are known to perform well experimentally on a broader class of programs when compared to their static counterparts. Still, they
do not provide completeness guarantees in finding loop invariants, even for linear loops (i.e., loops where all assignments and guards are linear expressions of program variables). 
% \saeid{we need citations for these works. do we have a list?} \ak{Added some references, I will read more papers on Invariant Automation, to get a list.}

We present the first data-driven technique with probabilistic guarantees for finding linear loop invariants of programs with integer variables.  We use a guess-and-check algorithm (~\cite{garg2016learning, garg2014ice, krishna2015learning, nguyen2017counterexample, zhu2018data}) to compute loop invariants. At the high level, a guess-and-check algorithm iteratively produces guess invariants and uses an SMT solver to verify these guesses. Typically, there is also a reinforcement component: if the SMT solver rejects a guess, it gives counterexamples that guide the guess algorithm to `refine' their guess until it converges. 

% \stsays{there is a jump on CHC cluase; can we gently, with a minimal text, introduce CHC clauses in the previous paragraph?}

In the guess phase of our algorithm (the Simulated Annealing Loop), we sample a dataset of positive, negative, and implication-pair points, similar to~\cite{garg2014ice}, from a Constrained Horn Clause (CHC) clause system. We use $\epsilon$-nets (\cite{haussler1986epsilon}) for sampling: our positive, negative, and implication pair points form an $\epsilon$-net of the respective sampling regions. This allows us to give termination guarantees on the check phase of our algorithm. We design a novel distance-based cost function, which is zero if and only if a candidate guess invariant satisfies the approximate CHC clause system.
%i.e., the CHC clause system when quantified only over our sampled data set instead of the full state space.
Intuitively, this cost function measures how `far' our current guess invariant is from satisfying the approximate CHC clause system. We use a multi-threaded simulated annealing search (\cite{mitra1986convergence}) to find a guess invariant satisfying the approximate CHC system. Our simulated annealing approach is syntactically motivated: we use coefficient and constant changes of Hamming distance $1$. 

In the check phase of our algorithm (the Datapoint Sampling Loop), we verify our approximate invariant on the CHC clause system using an SMT solver like Z3~\cite{de2008z3}. If it is not a correct invariant, we add the counterexamples generated by the SMT solver to our data points. Additionally, we also `tighten' the $\epsilon$-net periodically. Our guess-and-check approach demonstrates that guess-and-check algorithms can come with probabilistic completeness guarantees on finding invariants for a large class of loop programs.

Although linear loop invariants theoretically form a small subclass of loop invariants, in practice, they form the bulk of the required loop invariants for loop programs. We discuss the scope and limitations of our work and report good experimental results on the benchmarks from the literature, compared to the state-of-the-art loop invariant synthesis tools.  We summarize the contributions of our paper as follows:
\begin{itemize}
\item We improve the convergence guarantee provided by \cite{garg2016learning} from $O(|\mathcal{D}_{state}|^{\frac{1}{n}})$ to $O(\log_2 |\mathcal{D}_{state}|)$. Furthermore, unlike \cite{garg2016learning}, our convergence guarantees
can be extended to all CEGIS approaches.
% Our proof techniques can be extended to provide theoretical termination guarantees for all CEGIS-based algorithms - this is accomplished by using concepts from computational geometry.
\item We introduce novel counterexample generation techniques (\textit{dispersed cex}) in SMT-solvers and a generalization of implication pairs (\textit{iterated implication pairs}) that guarantee local progress of our invariant search algorithm.
\item We construct a novel tight approximated polytope-set distance function to serve as the cost of our invariant search and develop the simulated annealing algorithm to allow for parallelism.
\item We show the effectiveness of our tool, \toolname{}, on real benchmarks as well as in comparison with the state-of-the-art tools including \gspacer{} and \la{}. We identify a subclass of programs where we outperform \gspacer{}.
% \item We introduce techniques for prompting SMT-solvers to generate interesting counterexamples, which we call `dispersed counterexamples', to input CHC clauses.
% \item We define a new type of datapoint, the 'Iterated Implication Pair,' which generalizes the concept of loop unrolling.
% \item We identify parameters of our datapoint sampling algorithm which affect the success probability of our search space algorithm. Typically, these algorithms are handled orthogonal to each other in generic CEGIS-based techniques.
% \item We provide an efficiently computable approximate polytope-set distance function for invariant search, with upper bounds on the approximation factor.
% \item We leverage parallelism in our invariant search algorithm to (i) handle parameterized invariant spaces and (ii) increase the probability of success.

\end{itemize}

\section{Related Work.} \label{sec:relatedWorks}
We categorize and describe existing works into static and dynamic approaches.
%There have been several techniques in loop invariant inference via  We briefly describe them here.

% One primary technique is abstract interpretation~\cite{cousot1976static}, e.g., Interval Domain Analysis~\cite{bradley2007calculus}, Karr's Analysis~\cite{bradley2007calculus}, Octagonal Abstract Domain Analysis~\cite{mine2006octagon}, Combined Heap Analysis~\cite{qin2013loop} and Polynomial Equality Domain Analysis~\cite{rodriguez2004abstract} operate by fixing an invariant domain from some decidable theory and then repeatedly strengthening the invariant in order to find a loop invariant. 
% In many cases, they have to use a `widening operator' to guarantee the termination of their algorithm.
% Col\'{o}n et al.~\cite{colon2003linear} proposed a static technique for generating linear invariants wherein they use Farkas' Lemma to convert linear constraints into nonlinear constraints, and then solve the nonlinear constraints.
\vspace{0.5 em}
\noindent \textbf{Static Techniques}. 
One primary technique is abstract interpretation~\cite{cousot1976static,bradley2007calculus,mine2006octagon,qin2013loop,rodriguez2004abstract,sankaranarayanan2004non} that operate by fixing an invariant domain from some decidable theory and then repeatedly strengthening the invariant to find a loop invariant. 
% For example, Sankaranarayanan et al.~\cite{sankaranarayanan2004non} describes a method to generate polynomial invariants using Gr\H{o}bner basis.
These techniques, however, often struggle with finding sufficiently strong loop invariants, or they cannot infer invariants that use the disjunction operator.  
Interpolation techniques via Counterexample Guided Abstract Refinement (CEGAR)  such as in \cite{mcmillan2013solving, heizmann2010nested, rummer2013disjunctive, sharma2012interpolants} is another major technique, but the method is limited to theories that allow interpolants; so this method cannot be use to find invariants for theories like LIA, which does not admit Craig Interpolants. %\ak{Those techniques find interpolants and keep refining the interpolants based on counterexample till the interpolant becomes inductive (an inductive interpolant is a loop invariant!). But not all theories allow Craig Interpolants (so those methods dont apply only), specifically LIA theory (our theory) doesn't allow Craig Interpolation. But LRA theory allows Craig Interpolants.} 
Another technique that finds loop invariants using a CEGAR loop is predicate abstraction such as in~\cite{dillig2013inductive, janota2007assertion, furia2010inferring}, but these methods do not solve for invariants with disjunctions.
Crucially, GSPACER \cite{vediramana2023global}, an extension of SPACER \cite{komuravelli2016smt}, enhances loop invariant discovery through global guidance by local reasoning and interpolation, making it the current state-of-the-art tool for invariant automation. Finally, work by~\citet{karimov2022s} used linear algebra to infer loop invariants for a small subclass of decidable linear loop programs.

% . Interpolation techniques are often not expressive enough to capture all invariants and are often limited to a limited class of first-order theories for which computing an interpolant is possible. 
% Recurrence solvers find loop invariants by reducing the CHC equation to a recurrence equation such as in \cite{kafle2021regular}. Typically such recurrence techniques struggle with complicated CHC equations like those that incorporate conditionals within the loop body. 
% Finally, \citet{karimov2022s} use linear algebra (specifically the Jordan decomposition of the transition matrix) for linear loops, to find loop invariants for a small subclass of decidable linear loop programs.

\vspace{0.5 em}
\noindent \textbf{Dynamic Techniques}. These techniques often sample data points through Counter-Example Guided Inductive Synthesis
(CEGIS)~\cite{solar2006combinatorial}. In doing so, they (i) sample some positive and negative data points, (ii) learn a classifier for the labeled dataset, and (iii) check if the classifier forms a
valid invariant using an SMT solver. If the SMT solver found a counterexample, they add it to the dataset and go to step (ii). 
The existing works differ from each other in how they learn a
classifier. \citet{li2017automatic} use support vector machines (SVMs) to learn invariants as a disjunction of
conjunctions where each conjunctive cube corresponds to a different path execution within the loop body. 
% \cite{krishna2015learning} use decision trees to learn the classifier; \cite{garg2016learning} extends this technique by using implication pairs as a new type of data point along with positive and negative data points in their decision-tree classifier generation problem.
\citet{si2020code2inv} follow a similar idea, but use neural networks as their classifier in a CEGIS loop. Other works~\cite{sharma2013data,nguyen2017counterexample,nguyen2021using,garg2014ice,padhi2016data,bhatia2020oasis,sharma2013verification} use 
a simple search (e.g., via arbitrary linear inequality) to learn polynomial equalities and octagonal inequalities as classifiers.
% These techniques differ from each other in the
% methodology mainly on how they learn a
% classifier.
% The work~\citet{li2017automatic} uses support vector machines (SVMs)
% to learn conjunctive invariants, and use a path-sensitive
% classification to learn invariants that are the disjunction of
% conjunctions (by learning a conjunctive invariant for each path within
% the loop body).
% Additionally, they describe a sampling idea by which
% they sample data points near the classification boundary---they call
% this `selective sampling'.
% Other studies have used non-ML
% methods to learn a classifier (see~\cite{nguyen2017counterexample,
%   nguyen2021using});
% they use a simple search to learn polynomial
% equalities and octagonal inequalities as classifiers, Sharma et al.
% \cite{sharma2013data} show that any invariant from their template of
% conjunctive polynomials of bounded degree must lie in the nullspace
% for the data matrix, and their classifier guesses correspond to the
% basis of the null space of the data matrix.
Extending decision tree classifiers~\cite{krishna2015learning,garg2016learning},
\citet{zhu2018data} use a
combination of SVM and decision trees to learn classifiers---they
learn coefficients of the hyperplanes using SVM and the constant
values of the hyperplanes using decision trees. 
% Garg et al.~\citet{garg2014ice}
% also use a simple search to find classifiers preferring simpler
% classifiers. 
\citet{si2018learning} use a neural network (GNN +
autoregressive model) to predict classifiers, inside a CEGIS Loop with
an SMT solver, which they refer to as reinforcement
learning. 
% Padhi et al.~\citet{padhi2016data} use a simple
% search to learn arbitrary linear inequality predicates, starting from
% predicates with smaller coefficients and constant
% values. 
% Bhatia et al.~\citet{bhatia2020oasis} build upon \citet{padhi2016data}'s work by using the same classifier learning strategy, but also does a variable reduction by identifying only the relevant variables, and learns a classifier over these variables
% only. 
% Sharma et al.~\citet{sharma2013verification} learn classifiers over the octagonal equality domain. 
All the above CEGIS systems require that the
corresponding first-order theory of the SMT solver be decidable.
Neider et al.~\cite{neider2020learning}, however, explore two instances of theory where the formula validity is not known to be decidable, namely invariants with quantifiers, and programs with heap memory. \citet{sharma2016invariant} - the only other randomized algorithm apart from ours - use Metropolis-Hastings as their search algorithm to learn classifiers in a CEGIS loop; however, they only test their tool on toy examples and do not have theoretical guarantees of success. 

\section{Preliminaries.}
\label{sec:background}

\subsection{CHC Clause System.}
\label{subsec:chc}
A Constrained Horn Clause (CHC) system
(e.g., see~\cite{bjorner2015horn}) formally expresses the requirements for a loop invariant. Suppose we have a program with a single loop with precondition $P$, postcondition $Q$, loop guard $B$, and the loop body transition relation $T$ ($T$ relates a state before the execution of the loop body to a state after the execution), then we construct the following CHC system:
\begin{align*}
(\forall \vec{s}) \: P(\vec{s}) &\rightarrow I(\vec{s}) &(C1) \\
(\forall \vec{s},\vec{s'}) \:  I(\vec{s}) \land B(\vec{s}) \land T(\vec{s},\vec{s'}) &\rightarrow I(\vec{s'}) &(C2) \\
(\forall \vec{s}) \: I(\vec{s})  &\rightarrow Q(\vec{s}) &(C3) 
\end{align*}
where $\vec{s},\vec{s'}$ are vectors denoting program states and $I$ is a loop invariant.  C1, C2, and C3 are called the fact, inductive, and query clauses, respectively. The fact clause (C1) states that a precondition state must satisfy the loop invariant, and the inductive clause (C2) states that if the loop runs for a single iteration on any state satisfying the invariant and the loop guard, the state after one loop iteration must also satisfy the invariant, and the query clause (C3) states that any state satisfying the invariant must also satisfy the postcondition\footnote{Note that we do not have the negation of the loop guard in $(C3)$ because the CHC clause $(\forall \vec{s}) \: I(\vec{s}) \land \neg B(\vec{s})  \rightarrow Q(\vec{s})$ is equivalent to $(\forall \vec{s}) \: I(\vec{s})  \rightarrow  Q(\vec{s}) \lor B(\vec{s})$. Therefore, we can weaken the postcondition to account for the loop guard.}.  
The antecedent and consequent of a clause are called the \emph{body} and \emph{head} of the clause, respectively. It is well known that finding loop invariants for the above CHC system is undecidable\footnote{Some authors use the term 'invariant' to refer to classifiers that distinguish between good and bad states, often relaxing the requirement of inductiveness. In our paper, however, we will consistently use the term 'invariant' to mean 'inductive invariant', specifically one that satisfies the inductive clause. }~\cite{blass2001inadequacy}. Practically, however, we are often interested in a finite state space, which we shall denote by $\mathcal{D}_{\textit{state}}$. {For example, if the program has $n$ signed integer variables then $\mathcal{D}_{\textit{state}} = \mathbb{Z}^n \cap \footnotesize{[- \text{INT}\_\text{MAX} - 1, \text{INT}\_\text{MAX}]}^n$, where $\footnotesize{\text{INT}\_\text{MAX}}$ is the largest integer that the system can store.} This leads to the following bounded CHC clause system:
\begin{alignat*}{2}
(\forall \vec{s} \in \dstate) &\: P(\vec{s}) &&\rightarrow I(\vec{s})  \\
(\forall \vec{s},\vec{s'} \in \dstate) &\:  I(\vec{s}) \land B(\vec{s}) \land T(\vec{s},\vec{s'}) &&\rightarrow I(\vec{s'})  \\
(\forall \vec{s} \in \dstate) &\: I(\vec{s})  &&\rightarrow Q(\vec{s}) 
\end{alignat*}
As the invariant space for the above bounded CHC clause system is $\mathcal{P}(\dstate)$ and hence finite, finding a loop invariant for the above system is decidable. However, a naive brute force search takes $O(2^{|\dstate|})$ time, which is well beyond the limits of modern computers even for a modest-sized $\dstate$.

\subsection{Ground Positive, Ground Negative, and Implication Pair States.}
We define the \emph{ground positive states} $+_g$, \emph{ground negative states} $-_g$, and \emph{ground implication pair states} $\rightarrow_g$ \cite{garg2014ice, garg2016learning} below:
\begin{align*}
    +_g &:= \{ \vec{s} \in \mathcal{D}_{\text{state}} \: : \: P(\vec{s}) \} \\
    -_g &:= \{ \vec{s} \in \mathcal{D}_{\text{state}} \: : \: \neg Q(\vec{s}) \} \\
    \rightarrow_g &:= \{ (\vec{s}, \vec{s'}) \in \mathcal{D}_{\text{state}}^2 \: : \: B(\vec{s})  \land T(\vec{s}, \vec{s'}) \}
\end{align*}
Intuitively, a ground-positive state refers to the loop entry for which the precondition $P$ holds. A ground-negative state refers to the loop exit for which the postcondition $Q$ does \emph{not} hold. A ground implication pair state $(\vec{s}, \vec{s'})$ means that if the loop starts with state $\vec{s}$ satisfying the loop guard, then after one iteration the state of the program is $\vec{s'}$. From a different angle, if $\vec{s}$ is a state that is reachable from a state in $+_g$, then $\vec{s'}$ must also be a reachable state; on the other hand, if $\vec{s}$ is not reachable, there is no obligation that $\vec{s'}$ must be reachable.
For implication pairs, we introduce the following terminology: If $(\vec{s}, \vec{s'})$ is an implication pair, the states $\vec{s}$ and $\vec{s'}$ are referred to as the head and tail of the implication pair, respectively. Similarly, if $\rightarrow$ is a set of implication pairs, then $HEAD(\rightarrow)$ and $TAIL(\rightarrow)$ denote the sets of all heads and tails of the implication pairs, respectively.

\subsection{Linear Integer Arithmetic Theory.}
Linear Integer Arithmetic (LIA) theory, denoted $T_{\mathbb{Z}}$, is syntactic sugar for the decidable Presburger Arithmetic theory $T_{\mathbb{N}}$ (see Appendix \ref{App: LIA theory}). $T_{\mathbb{Z}}$ is defined as a first order theory with signature $\Sigma = \{ ... , -2,-1,0,1, ... ,-2\cdot, -1\cdot, 2\cdot, ... , + , - ,  = , < , \le, > , \ge \}$ s.t. every formula $F$ in $T_{\mathbb{Z}}$ is equisatisfiable to some formula in $T_{\mathbb{N}}$ (see Appendix \ref{App: LIA theory}).  Standard SMT solvers (eg. Z3 \cite{de2008z3}) check the validity and satisfiability of formulas in LIA theory using a mix of decision procedures like DPLL for boolean abstraction and the generalized simplex algorithm followed by Branch and Bound for LIA literals \cite{dutertre2006fast, barrett2013decision}.

\subsection{$\epsilon$-Nets for Arbitrary Range Spaces.} \label{sec: enet}
Our method for finding loop invariants utilizes sampling of randomized $\epsilon$-nets w.r.t. ellipsoids. Such samples are defined and constructed using concepts of range spaces and Vapnik Chervonenkis (VC) dimension (\cite{matouvsek1993epsilon}) of range spaces. 
Intuitively, an $\epsilon$-net for a set $S$ w.r.t. some geometrical shape such as ellipsoids is a collection of points such that any `sufficiently' large ellipsoidal subset of $S$ must contain some point of the $\epsilon$-net. Hence if the data points we sample form an $\epsilon$-net w.r.t. ellipsoids, then intuitively all counterexamples to our guessed invariant must be contained in very small ellipsoids, and if we keep tightening the maximum size of these ellipsoids with consecutive iterations, we get a convergence guarantee for our algorithm. We next define range spaces, VC dimension, $\epsilon$-nets, and refer to a theorem to construct an $\epsilon$-net w.r.t. any range space $\mathcal{R}$ by uniform sampling.

\begin{definition}
    A pair $\Sigma = (X,\mathcal{R})$ is called a range space if $X$ is a non-empty set and $\mathcal{R}$ is a collection of subsets of $X$. Members of $\mathcal{R}$ are called ranges. $\Sigma = (X,\mathcal{R})$ is said to be finite if $X$ is finite. 
\end{definition}

\begin{definition}
    Let $(X,\mathcal{R})$ be a range space. A finite set $A \subseteq X$ is said to be shattered by $\mathcal{R}$ if $\{ A \cap r \: | \: r \in \mathcal{R}\} = 2^{A}$. Then the Vapnik Chervonenkis (VC) dimension of $(X,\mathcal{R})$ is defined as the size of the largest (in terms of cardinality) subset of $X$ that is shattered by $\mathcal{R}$.
\end{definition}

We give the VC dimension for polytopes in $\mathbb{R}^n$ and $n$-dimensional ellipsoids in Appendix~\ref{App:enetvc} with illustrations in Appendix~\ref{App:enetvc}. We next define $\epsilon$ nets for a range space.

\begin{definition}
    Given a finite range space $(X,\mathcal{R})$ and a parameter $0 < \epsilon < 1$, a set $N \subseteq X$ is called an $\epsilon$-net for $X$ w.r.t. $\mathcal{R}$ if for every range $r \in \mathcal{R}$ satisfying $|r| \ge \epsilon |X|$, we must have $N \cap r \neq \varnothing$. 
\end{definition}

Intuitively, an $\epsilon$-net $N$ for $X$ w.r.t. $\mathcal{R}$ is a set which `approximates' $X$ w.r.t. $\mathcal{R}$ such that any range $r \in \mathcal{R}$ either has relative frequency (w.r.t. $X$) less than $\epsilon$ or must intersect $N$. 
The following theorem gives a randomized construction of an $\epsilon$ net for any finite range space of bounded VC dimension:

\begin{theorem} \label{Thm: HAW} \cite{haussler1986epsilon}
    If $(X, \mathcal{R})$ is a finite range space of VC dimension $vc$, $\epsilon > 0$ and $m \ge \frac{8}{\epsilon}$, then a set $N \subseteq X$ obtained from $m$ uniform random independent draws from $X$ is an $\epsilon$-net for $X$ w.r.t. $\mathcal{R}$ with probability strictly greater than $1 - 2 \phi_{vc}(2m)2^{-\frac{\epsilon m}{2}}$ where 
\begin{align*}
    \phi_d(n) &:= \begin{cases}
        \sum_{i = 0}^d \binom{n}{i}, & \mbox{if}\ d \le n \\
        2^n, &\mbox{if}\ d > n
    \end{cases}
\end{align*}
\end{theorem}

The above theorem says that any randomly sampled subset of $X$ of `sufficient size' forms an $\epsilon$-net for $X$ w.r.t. any range space with some success probability. However constructing a random sample from the lattice points within a polytope is a complicated task; {a vector $\vec{x} \in \mathbb{R}^n$ for some $n \in \mathbb{N}$ is called a lattice point if $\vec{x} \in \mathbb{Z}^n$; i.e. it only has integer coordinates.} To solve the issue, we use a combination of rejection sampling and sampling from the solution spaces of linear Diophantine equations \cite{gilbert1990linear} to efficiently construct such random samples (more details in Appendix \ref{App: Uniform Sampling poly lIA}).

% For bounded polytopes, we can produce random samples by using rejection sampling (described in Appendix \ref{App: Rejection Sampling}). In short, rejection sampling repeatedly produces uniformly distributed samples from the axis-aligned orthotope containing the polytope until such a sample lies within the polytope. The (time) efficiency of rejection sampling to produce a sample is directly proportional to the relative frequency of the bounded polytope w.r.t. its axis-aligned orthotopic superset. Hence, when sampling from polytopes with `low' relative frequency of polytope w.r.t it's axis aligned orthotope, rejection sampling proves to be impractical, and a different algorithm is required. Luckily, in practice all such polytopes we encounter are non-axis aligned \textit{affine spaces}, and we can use an alternate technique (described in Appendix \ref{App: Sampling from Affine Spaces.}) based on sampling from the solution space of the corresponding linear diophantine equation as described in ~\cite{gilbert1990linear}. This method first transforms the matrix equation representing the affine space into a another matrix equation where the coeffecient matrix is in row-reduced echelon form by using only unimodular row transformations. It is then able to produce a random sample for the new matrix equation using a combination of back-substitution and rejection sampling - which it then to construct a random sample for the original matrix equation. The complete algorithm to construct uniform samples is given in Appendix \ref{App: Uniform Sampling poly lIA}.

\subsection{Ellipsoids in $\mathbb{R}^n$.} \label{Sec: PrereqEllipsoid}
The $\epsilon$-nets we construct for our algorithm use \( n \)-dimensional ellipsoids as ranges. As such, we need specific results on \( n \)-dimensional ellipsoids for our convergence guarantees. \\
Firstly, we note that the number of lattice points contained in an ellipsoid is proportional to the volume of the ellipsoid (see Appendix \ref{App:IOLJE} for a formal theorem). Furthermore, for any \( n \)-dimensional polytope, there exist two unique ellipsoids---one inscribed by and one circumscribing the polytope---such that the outer ellipsoid is a dilation of the inner ellipsoid. These ellipsoids are referred to as the inner and outer Lowner-John ellipsoids of the polytope, respectively (see Appendix \ref{App:IOLJE} for the formal theorem).

\subsection{Simulated Annealing.} 
Given a search space \(\mathcal{X}\), a symmetric neighborhood relation \(\mathcal{N}\) over \(\mathcal{X}\), and a cost function \(c : \mathcal{X} \rightarrow [0,\infty)\), simulated annealing (SA) \cite{bertsimas1993simulated, romeo1991theoretical} searches for a  candidate \(s^* \in \mathcal{X}\) such that \(c(s^*) = 0\), which denotes a global minimum of \(\mathcal{X}\) with respect to \(c\) (see Algorithm \ref{alg: SimA}).

% \begin{algorithm} 
% \caption{\texttt{SIMULATED\_ANNEALING}}\label{alg: SimA}
% \begin{algorithmic}[1] 
% \State \textbf{Input: } $(\mathcal{X}, \mathcal{N}, c , s_0)$
% \State \textbf{Hyperparameters: } $t_{max}$
% \State $(s,c) \longleftarrow (s_0, c(s_0))$
% \State $T_0 \longleftarrow \texttt{INITIAL\_TEMPERATURE}(\mathcal{X}, \mathcal{N}, c)$
% \For{$t = 1$ to $t_{max}$}
% \If {$c = 0$} \Return $s$
% \EndIf
% \State $T \longleftarrow \frac{T_0}{\ln 1 + t}$
% \State $s' \longleftarrow \texttt{UNIFORM\_SAMPLE\_LIST}(\mathcal{N}(s))$
% \State $c ' \longleftarrow c(s')$
% \State $a \longleftarrow e^{-\frac{(c' - c)^+}{T}}$
% \If{$\texttt{UNIFORM\_SAMPLE\_INTERVAL}([0,1]) \le a$}
% \State $(s,c) \longleftarrow (s',c')$.
% \EndIf
% \EndFor
% \State \Return \texttt{FAIL}
% \end{algorithmic}
% \end{algorithm}

\noindent\begin{minipage}{0.50\textwidth}
    \begin{algorithm}[H]
    \caption{\texttt{SIMULATED\_ANNEALING}}\label{alg: SimA}
    \begin{algorithmic}[1] 
    \State \textbf{Input: } $(\mathcal{X}, \mathcal{N}, c , s_0)$
    \State \textbf{Hyperparameters: } $t_{max}$
    \State $(s,c) \longleftarrow (s_0, c(s_0))$
    \State $T_0 \longleftarrow \texttt{INITIAL\_TEMPERATURE}(\mathcal{X}, \mathcal{N}, c)$
    \For{$t = 1$ to $t_{max}$}
    \If {$c = 0$} \Return $s$
    \EndIf
    \State $T \longleftarrow \frac{T_0}{\ln 1 + t}$
    \State $s' \longleftarrow \texttt{UNIFORM\_SAMPLE\_LIST}(\mathcal{N}(s))$
    \State $c ' \longleftarrow c(s')$
    \State $a \longleftarrow e^{-\frac{(c' - c)^+}{T}}$
    \If{$\texttt{UNIFORM\_SAMPLE\_INTERVAL}([0,1]) \le a$}
    \State $(s,c) \longleftarrow (s',c')$.
    \EndIf
    \EndFor
    \State \Return \texttt{FAIL}
    \end{algorithmic}
    \end{algorithm}
    \vspace{-10pt} % Adjust vertical spacing if needed
\end{minipage}
\noindent\begin{minipage}{0.46\textwidth}
SA first sets the current candidate to \(s_0\), and computes its cost \(c\). Then, it sets the initial temperature using the \texttt{INITIAL\_TEMPERATURE} procedure. For up to \(t_{max}\) iterations, where \(t_{max}\) is a hyperparameter, it repeatedly does the following: it checks if the cost of the current candidate is zero. If so, it returns the current candidate; otherwise, it computes the temperature for the current iteration $t$ using a logarithmic cooling schedule, i.e., \(T \longleftarrow \frac{T_0}{\ln(1 + t)}\). Next, it uniformly samples a neighbor of \(s\) w.r.t the neighborhood relation by calling the \texttt{UNIFORM\_SAMPLE\_LIST} procedure. It then computes the cost of this neighbor, \(c'\), and chooses to make this transition with probability \(e^{-\frac{(c' - c)^+}{T}}\), known as the acceptance ratio. If it is unable to find an optimal solution within \(t_{max}\) iterations, it returns \texttt{FAIL}.
\end{minipage} \\
With the logarithmic cooling schedule, simulated annealing converges to a global minimum in time \( t \), with a probability at least \( 1 - \frac{\alpha}{t^\beta} \), irrespective of the starting candidate where \(\alpha\) and \(\beta\) are determined by the parameters of the invariant space. Theorem \ref{thm: SAguarantee} in Appendix \ref{App:SA} formalizes the convergence guarantees of simulated annealing for a logarithmic schedule. \\
We finally mention how to set $T_0$. Theorem \ref{thm: SAguarantee} advises making \(T_0 \geq rL\), where $r$ represents the radius of the invariant space and $L$ represents the maximum cost jump in a single transition (see Theorem \ref{thm: SAguarantee}). While setting \(T_0 \ge rL\) guarantees that our simulated annealing search will terminate by Theorem \ref{thm: SAguarantee}, in practice, this does not always yield the best results. 

We leverage the study in \cite{ben2004computing} to compute an experimentally optimal value for \(T_0\) (refer to Appendix \ref{App:SA} for the algorithm).

\section{Problem Statement: Standard LIA Single Loop CHC Systems.}
\label{Sec: Problem Statement}
The main goal of this study is to find a loop invariant that satisfies a specific CHC system of the form stated in Section~\ref{subsec:chc}, which is we call \emph{a standard LIA single loop CHC system}. Our tool further takes inputs $c$ and $d$ and aims to find an invariant from LIA theory that can be written as the disjunction of $d$ cubes, where each cube is the conjunction of at most $c$ LIA predicates of the form $\vec{w}_{ij} \cdot \vec{x} \le b_{ij}$, where $\vec{w_{ij}} \in \mathbb{Z}^n$, $b_{ij} \in \mathbb{Z}$ are variables forall $i$ and $j$, and $\vec{x}$ represents the vector of $n$ program variables\footnote{To ensure uniformity, we exclusively employ inequality LIA predicates in our representation. An equality LIA predicate of the form \(LHS = RHS\) is modeled using the conjunction of 2 inequality predicates as \((LHS \le RHS) \land (LHS \ge RHS)\).}. We call such invariants $d$-$c$ LIA invariants, which we now formally define.
\begin{definition}
A $d$-$c$ DNF is defined as a formula expressed as the disjunction of $d$ cubes, where each cube is the conjunction of $c$ predicates. A $d$-$c$ LIA DNF is then a $d$-$c$ DNF with the additional restriction that the predicates are derived from the LIA theory. Using this terminology, a $d$-$c$ LIA invariant is just an invariant written in the form of a $d$-$c$ LIA DNF.    
Such $d$-$c$ LIA invariants (and in general, $d$-$c$ LIA DNF's) can be viewed geometrically in $\mathbb{R}^n$ as\footnote{$n$ is the number of program variables.} the lattice points contained in the union of at most $d$ convex compact polytopes each with at most $c$ faces (we call the set of all lattice points inside a polytope as 'polytope-lattice points').
\end{definition}
%\\[1\baselineskip]

% \vspace{0.25 em}
\noindent \textit{Assumptions.} The input {standard LIA Single Loop CHC system} is characterized by the following requirements:
\begin{enumerate*}
    \item The system considers loop programs with a single loop. Programs with multiple loops require separate loop invariants, and solving the corresponding CHC system would necessitate an extension of the methods and techniques presented here.
    
    \item The system permits only integer variables in the programs.
    
    \item The propositions $P$, $B$, and $Q$ in these loop programs are expressible as $d$-$c$ LIA DNFs, for some values of $d$ and $c$.
    
    \item The transition relation $T$ is a piece-wise linear integer relation over the loop guard $B$, where the piece-wise linear integer relation is defined in Definition \ref{Def: linear relation}.
\end{enumerate*}
\begin{definition} \label{Def: linear relation}
The LIA formula $T$ defines a piece-wise linear integer relation over $B$ if there exists a finite partition $\{ B_i \}$ of $B$ s.t. $B_i$ are LIA formulas for all $i$,
and $T$ when restricted on $B_i$, defined as $T_{| B_i} := \{ (x,y) \: | \: T(x,y) \land x \in B_i \}$, is a linear integer relation over $B_i$.
Specifically, $T_{| B_i}$ is a linear integer relation over $B_i$ if there exists a finite set of functions on $B_i$, denoted $t_{ij} : B_i \rightarrow \mathbb{R}^n$, where for all $j$, $t_{ij}$ is a linear function with integer coefficients and $ T_{| B_i}(\vec{s}, \vec{s'}) \equiv \bigvee_j (\vec{s'} = t_{ij}(\vec{s}))$. Then $T(\vec{s}, \vec{s'}) \equiv \bigwedge_{i} (B_i(\vec{s})  \rightarrow T_{| B_i} (\vec{s}, \vec{s'})) $. Additionally, if the partition size is $d$ and the maximum number of functions per partition block is $r$, then $T$ is called a $d$-piecewise linear integer $r$-relation over $B$. 
\end{definition}
In short, a piecewise linear integer relation over \( B \) generalizes linear transformations over $B$ induced by straight line code, to allow for conditionals with deterministic and nondeterministic guards (see Appendix \ref{App: MotivateT} for a motivating example). Converting a transition relation to a piecewise-linear integer relation may lead to an increase in the size of the transition relation, depending on the number of conditional branches within the loop body. However, this blow-up is not an inherent limitation of our theory; rather, it arises as a result of representing the transition relation as a series of linear transformations over disjoint spaces. Extending this work to handle linear transformations over possibly overlapping spaces, which would avoid such blowups, is left for future research. 

Although seemingly constrained, this CHC system is highly expressive, enabling us to encode the majority of benchmarks—including those with randomly initialized variables, nondeterministic loop guards, conditionals with both deterministic and nondeterministic branches, and constructs such as \texttt{break}/\texttt{continue} and \texttt{goto} statements—into our input format (see Appendix~\ref{App: Representing Practical Programs in the Standard CHC Clause system} for a more comprehensive discussion). Additionally, by focusing on CHC systems with these constraints, we can leverage mathematical theories, such as computational geometry, to identify invariants—a capability not explored by previous tools. Furthermore, our experiments in Section \ref{Sec: ExpComp} reveal that existing state-of-the-art tools still face significant challenges in generating invariants for standard LIA single-loop CHC systems.

\textit{Design Choice Motivation.} 
As previously mentioned, our goal is to find \( d \)-\( c \) LIA invariants for a standard LIA single-loop CHC system for some \( d \) and \( c \). We first observe that any invariant expressed by an LIA formula can be represented as a DNF, and thus can also be expressed as a \( d \)-\( c \) LIA invariant for some \( d \) and \( c \). Therefore, our requirement for a \( d \)-\( c \) LIA invariant is simply a restriction on the existence of an LIA invariant, which we now motivate. We note that single-loop programs relying on transition relation formulas from LIA theory generally require LIA formulas to represent their invariants, a phenomenon validated by our experiments on the benchmarks. However, discovering invariants using other non-linear theories or Linear Real Arithmetic (LRA) for more generalized CHC systems is an extension of our current work and is left for future study. \\
The next thing to motivate is that $d$ and $c$ are input by the user. Note that if a standard LIA Single Loop CHC system admits a $d$-$c$ LIA invariant, then it also admits a $d'$-$c'$ LIA invariant for any $d' \ge d$ and $c' \ge c$ (as you can pad the invariant by creating copies of any predicate or  $c$-cube). Hence there is an infinite range of values for $(d,c)$ which we can supply for which there exists a $d$-$c$ LIA invariant if the given benchmark program admits a LIA invariant in the first place (We also implement an ML algorithm to advise user on a good choice of `d' and `c'). Furthermore, this design allows for finding an invariant with a size bound, an option that may not be available for other algorithms.

\section{Methodology.}
\label{Sec:Algorithm}
We structure this section as follows: in Section~\ref{Sec:Overview}
we overview our approach, as implemented in a tool, called \toolname{}, at a high level on a toy example. In Section~\ref{SubSec:CEGIS Loop} and Section~\ref{SubSec:SA Loop}, we describe the details of the two major components of our algorithm i.e. the Datapoint Sampling Loop and Simulated Annealing Loop, respectively. Finally, in Section~\ref{Subsec: Putting it all together}, we put our whole algorithm together.

\subsection{Overview on a Toy Example.}
\label{Sec:Overview}
Consider an example of a simple \texttt{C} program in Fig~\ref{eg1}.
Its loop invariant is an LIA formula over the integer variables $x$ and $y$ denoted by $I(x,y)$, which satisfies the following standard LIA CHC clause system:
\begin{alignat*}{2}
    (\forall x,y \in \mathcal{D}_{\text{int}}) &\: (x = 1) \land (y = 1) &&\rightarrow I(x,y) \\
    (\forall x,y, x', y'\in \mathcal{D}_{\text{int}}) &\: I ( x,y ) \land (x' = x + y) \land (y' = x') &&\rightarrow I(x', y') \\
    (\forall x,y \in \mathcal{D}_{\text{int}}) &\:I(x,y) &&\rightarrow y \ge 1
\end{alignat*}
where $\mathcal{D}_{\text{int}} := \mathbb{Z} \cap \footnotesize{[- \text{INT}\_\text{MAX} - 1, \text{INT}\_\text{MAX}]}$.

% \begin{wrapfigure}{r}{0.52\textwidth}
% 	\begin{flushright}
% 	\vspace{-2.0 em}
%  \begin{scriptsize}
%     \begin{lstlisting}[caption=Overview Example, label={eg1}, language=C]
%     int main(void) {
%         int x = 1; int y = 1;
%         while (rand_bool()) { 
%             x = x + y; 
%             y = x;
%         }
%         assert (y >= 1);
%         return 0;
%     }
%     \end{lstlisting}
%    \end{scriptsize}
% 	\vspace{-1.0 em}
% \end{flushright}
% % \label{eg1}
% \end{wrapfigure}

\begin{minipage}{0.35\textwidth}
\begin{figure}[H]
\begin{verbatim}
int main(void) {
    int x = 1; int y = 1;
    while (rand_bool()) { 
        x = x + y; 
        y = x; 
    }
    assert (y >= 1);
    return 0;
}
\end{verbatim}
\vspace{-2.5 em}
\caption{An example \texttt{C} program. }
\label{eg1}
\end{figure}
\end{minipage}
\begin{minipage}{0.6\textwidth}
\toolname{} takes parameters $c$ and $d$ as inputs.  For this
example, assume $c = 2$ and $d = 1$, meaning that an invariant
should be the conjunction of two LIA predicates.  \toolname{} comprises of two loops as depicted in Fig \ref{fig:DSSA Connection}.  The inner loop, called the Simulated Annealing (SA) loop, takes a set of data samples, say $+_k$, $-_k$, and $\rightarrow_k$ as input (where $k$ is the $k$-th time the SA loop 
is invoked by the outer loop) and performs a search over the invariant space to find a correct invariant of  the following \textit{approximate CHC system}
\end{minipage}

\begin{alignat*}{2}
    (\forall (x,y) \in +_k) &\: (x = 1) \land (y = 1) &&\rightarrow I(x,y) \\
    (\forall ((x,y) , (x', y')) \in \: \rightarrow_k) &\: I ( x,y ) \land (x' = x + y) \land (y' = x') &&\rightarrow I(x', y') \\
    (\forall (x,y) \in -_k) &\:I(x,y) &&\rightarrow y \ge 1
\end{alignat*}
This new CHC system differs from the original in the set domains over which each CHC clause is quantified. We call an invariant to the above approximate CHC system an \emph{approximate invariant}.  The approximate invariant is then passed to the outer loop, called the DS (Datapoint Sampling) loop,
which uses an SMT solver (e.g., Z3) to check if the approximate
invariant satisfies the original bounded CHC
system. If yes, it finds a correct loop invariant. Otherwise, we
ask the SMT solver for a small set of counterexamples, which we use to
construct additional data samples $(+^\times_k,
-^\times_k, \rightarrow^\times_k)$. These samples are appended to the data points to form a new domain of quantification for the next SA loop iteration; i.e., we let 
$(+_{k+1}, -_{k+1},
\rightarrow_{k+1}) = (+_k \cup +^\times_k, -_k \cup -^\times_k,
\rightarrow_k \cup \rightarrow^\times_k)$. Additionally, at periodic iteration intervals, we
substitute our $\epsilon$-net with a `tighter' net, i.e. an
$\epsilon'$-net for $0 < \epsilon' < \epsilon$.  These factors (a
bounded number of counterexamples combined with decreasing $\epsilon$)
ensure faster convergence of the DS loop than a simple brute force approach for finding the invariant. For example, for the current example under consideration, we take $+_0, -_0,$ and $ \rightarrow_0$ to be (randomized) $\epsilon$-nets of $(x = 1) \land (y = 1)$, $\neg (y \ge 1)$, and $\{
((x,y) , (x + y,x + y)) \}$\footnote{In the current context, we call $\rightarrow \subseteq S \times S'$ an $\epsilon$-net of $S \times S'$ w.r.t ranges $\mathcal{R} \subseteq 2^{S}$ if $HEAD(\rightarrow)$ forms an $\epsilon$-net of $S$ w.r.t $\mathcal{R}$. Note that this definition deviates from the standard meaning of an $\epsilon$-net of $S \times S'$ w.r.t ranges $\mathcal{R}' \subseteq 2^{S \times S'}$.} w.r.t. ellipsoids respectively, for some chosen value of $\epsilon$.  

\begin{figure}[!tb]
    \centering
    \includegraphics[width=0.8\textwidth]{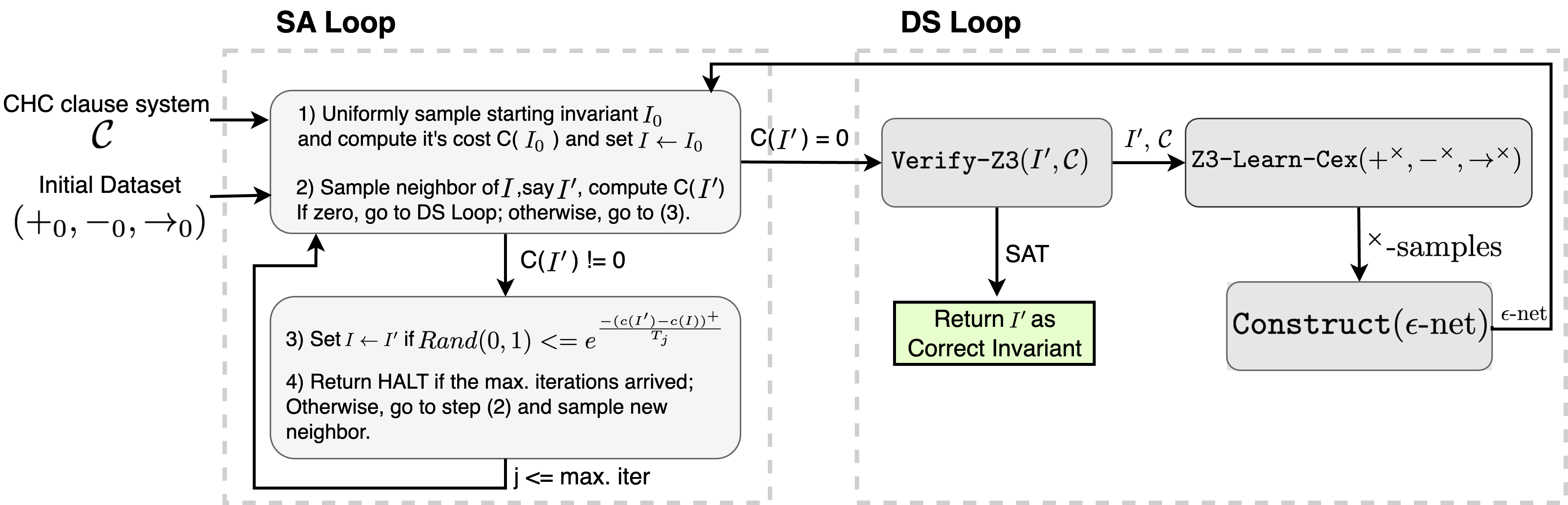}
    \caption{Diagram showing input-output behavior of DS and SA Loops of \toolname{}.  }
    \label{fig:DSSA Connection}
\end{figure}

In the SA loop, we search for an approximate invariant by running a simulated annealing search~\cite{mitra1986convergence} on the candidate invariant space: a candidate invariant is the union of any $d$ polytopes, each with at most $c$ faces.
% (for this example a candidate invariant is the intersection of two halfplanes).
In this search, our neighbors are defined by a change in the constant or coefficient value of some boundary line of some halfplane of the original candidate invariant and our cost function is defined as a linear-time computable approximation to the set-distance metric. Additionally, we use parallelism within the SA loop to increase the probability of finding an approximate invariant and to search across parameterized invariant spaces for different parameter values.

\subsection{Datapoint Sampling (DS) Loop.}\label{SubSec:CEGIS Loop}
In the previous section, we discussed approximate CHC systems, approximate invariants, and counterexamples. We now define these terms formally.
\begin{definition}
    Let $+_g$, $-_g$ and $\rightarrow_g$ be the ground positive, ground negative, and ground implication pair states respectively. If $+ \subseteq +_g$, $- \subseteq -_g$ and $\rightarrow \; \subseteq \;\rightarrow_g$, then we call the the triplet $(+, \rightarrow, -)$ a dataset and a $(+, \rightarrow, -)$-approximate CHC system is then defined as the following CHC system\footnote{Note that in the approximate CHC system, the quantifier sets have been restricted to $+$, $\rightarrow$ and $-$ for the fact, inductive and query clause respectively when compared to the standard LIA Single Loop CHC system.}:
 \begin{align*} 
(\forall \vec{s} \in  +) &\: P(\vec{s}) \rightarrow I(\vec{s}) \\
(\forall (\vec{s},\vec{s'}) \in \: \rightarrow) &\:  I(\vec{s}) \land B(\vec{s}) \land T(\vec{s},\vec{s'}) \rightarrow I(\vec{s'})  \\
(\forall \vec{s} \in -) &\: I(\vec{s})  \rightarrow Q(\vec{s}) 
\end{align*}
Any invariant (resp. $d$-$c$ LIA invariant) satisfying the above CHC system is called a $(+, \rightarrow, -)$-approximate invariant (resp. $(+, \rightarrow, -)$-approximate $d$-$c$ LIA invariant).
\end{definition}
\begin{definition}
    Let $I$ be an approximate invariant for some dataset.
    The plus counterexamples (abbreviated as 'cex' henceforth) $c_I^+$, implication counterexample pairs $c_I^{\rightarrow}$ and minus counterexamples $c_I^-$ for $I$ w.r.t the standard LIA Single Loop CHC system are defined as:
    \begin{align*}
        c_I^+ &= \{ \vec{t} \in \mathcal{D}_{state} \: : \: \neg \big( P(\vec{t}) \rightarrow I(\vec{t}) \big) \} \\
        c_I^\rightarrow &= \{ (\vec{t}, \vec{t'}) \in \mathcal{D}_{state}^2 \: : \: \neg \big( I(\vec{t}) \land B(\vec{t}) \land T(\vec{t}, \vec{t'}) \rightarrow I(\vec{t'}) \big) \} \\
        c_I^- &= \{ \vec{t} \in \mathcal{D}_{state} \: : \: \neg \big( I(\vec{t}) \rightarrow Q(\vec{t}) \big) \} 
    \end{align*}
    Finally, the cex-dataset for $I$ w.r.t the standard LIA Single Loop CHC system is defined as the triplet $(c_I^+, c_I^\rightarrow, c_I^-)$ and is denoted by $\mathcal{C}_I$. Following prior literature, we will call implication counterexample pairs \emph{ICE pairs}. Note that an implication pair is just a pair of head and tail states, whereas an ICE pair is an implication pair which is also a counterexample to some approximate invariant.
    
\end{definition}

%We now give an algorithmic description for the DS Loop before jumping into the details of each component. 

% \begin{algorithm}[!tb]
% \caption{\texttt{DATAPOINT\_SAMPLING\_LOOP}}\label{alg: DS Loop}
% \begin{algorithmic}[1]
% \State \textbf{Input: } $(P, B, T, Q)$
% \State \textbf{Hyperparameters: } $t_{max}$
% \State $\mathcal{S} \longleftarrow $ \texttt{INITIAL_DATASET}(P, B, T, Q)
% \State $I_0 \longleftarrow \texttt{INITIAL\_INVARIANT}(\mathcal{S})$
% \State $\mathcal{C} \longleftarrow \varnothing$
% \For{$t = 1$ to $t_{max}$}
% \State (\texttt{SUCCESS}, $I_{aprx}$) $\longleftarrow$ \texttt{PARALLEL\_SIMULATED\_ANNEALING}($\mathcal{S}, I_0$)
% \If{$\neg$ \texttt{SUCCESS}}
% \Return{\texttt{FAIL}}
% \EndIf
% \State (\texttt{CORRECT}, $\mathcal{C}_{I_{aprx}}$) $\longleftarrow$ \texttt{VERIFIER}($I_{aprx}$, P, B, T, Q)
% \If{\texttt{CORRECT}}
% \Return{$I_{aprx}$}
% \EndIf
% \State $\mathcal{C} \longleftarrow \mathcal{C} \sqcup \mathcal{C}_{I_{aprx}}$
% \If{ \texttt{REFINE\_DATASET\_CRITERION}($t$) }
% \State $\mathcal{S} \longleftarrow$ \texttt{REFINED\_DATASET}($P, B, T, Q$)
% \EndIf 
% \State $(\mathcal{S}, I_0) \longleftarrow (\mathcal{S} \sqcup \mathcal{C}, I_{aprx})$
% \EndFor
% \State \Return{\texttt{FAIL}}
% \end{algorithmic}
% \end{algorithm}

\subsection*{Overall Algorithm for the DS Loop.}
\begin{minipage}{0.45\textwidth}
The algorithm for the DS Loop is given in Algorithm \ref{alg: DS Loop}. The DS Loop takes as input the precondition $P$, the loop guard $B$, the transition relation $T$, and the postcondition $Q$ for the standard LIA Single Loop CHC system. It then computes an initial dataset $\mathcal{S}$ using its input by calling the \texttt{INITIAL\_DATASET} procedure; $\mathcal{S}$ is some triplet of plus points, implication pairs, and minus points. It then initializes an initial invariant by \texttt{INITIAL\_INVARIANT} and finally
initializes the cex-dataset $\mathcal{C}$ to be empty. we will give the exact details of these procedures later. The DS Loop then runs a loop for $t_{max}$ number of iterations. Within this loop, we first use parallel simulated annealing to compute an $\mathcal{S}$-approximate invariant for the approximate CHC system. \texttt{PARALLEL\_SIMULATED\_ANNEALING} is a procedure that takes as input a dataset $\mathcal{S}$ and an initial invariant $I_0$ and outputs a boolean variable \texttt{SUCCESS} (indicating if the simulated annealing
\end{minipage} 
\begin{minipage}{0.55\textwidth}
% \begin{wrapfigure}{r}{0.45\textwidth}
\begin{algorithm}[H]
\caption{\texttt{DATAPOINT\_SAMPLING\_LOOP}}\label{alg: DS Loop}
\begin{algorithmic}[1]
\State \textbf{Input: } $(P, B, T, Q)$
\State \textbf{Hyperparameters: } $t_{max}$
\State $\mathcal{S} \longleftarrow $ \texttt{INITIAL_DATASET}(P, B, T, Q)
\State $I_0 \longleftarrow \texttt{INITIAL\_INVARIANT}(\mathcal{S})$
\State $\mathcal{C} \longleftarrow \varnothing$
\For{$t = 1$ to $t_{max}$}
\State (\texttt{SUCCESS}, $I_{aprx}$) $\longleftarrow$  
\Statex \hspace*{1 cm} \texttt{PARALLEL\_SIMULATED\_ANNEALING}($\mathcal{S}, I_0$)
\If{$\neg$ \texttt{SUCCESS}}
\Return{\texttt{FAIL}}
\EndIf
\State (\texttt{CORRECT}, $\mathcal{C}_{I_{aprx}}$) $\longleftarrow$ \texttt{VERIFIER}($I_{aprx}$, P, B, T, Q)
\If{\texttt{CORRECT}}
\Return{$I_{aprx}$}
\EndIf
\State $\mathcal{C} \longleftarrow \mathcal{C} \sqcup \mathcal{C}_{I_{aprx}}$
\If{ \texttt{REFINE\_DATASET\_CRITERION}($t$) }
\State $\mathcal{S} \longleftarrow$ \texttt{REFINED\_DATASET}($P, B, T, Q$)
\EndIf 
\State $(\mathcal{S}, I_0) \longleftarrow (\mathcal{S} \sqcup \mathcal{C}, I_{aprx})$
\EndFor
\State \Return{\texttt{FAIL}}
\end{algorithmic}
\end{algorithm}
% \end{wrapfigure}
\end{minipage}
 search was successful or not) and an $\mathcal{S}$-approximate invariant $I_{aprx}$ if \texttt{SUCCESS} = \texttt{TRUE}. If \texttt{SUCCESS} $\neq$ \texttt{TRUE}, the DS Loop fails and halts. If not, the DS Loop feeds the \(\mathcal{S}\)-approximate invariant $I_{aprx}$ along with the CHC parameters \(P\), \(B\), \(T\), and \(Q\) to the procedure \texttt{VERIFIER}, which checks if \(I_{aprx}\) is a valid loop invariant using an SMT solver and if not, the \texttt{VERIFIER} outputs a cex-dataset $\mathcal{C}_{I_{aprx}}$ for \(I_{aprx}\) w.r.t the standard LIA Single Loop CHC system. If $I_{aprx}$ is the correct loop invariant, the DS Loop returns $I_{aprx}$ and halts. If not, it appends $\mathcal{C}_{I_{aprx}}$ to its current cex dataset $\mathcal{C}$.\footnote{If $\mathcal{S} = (+, \rightarrow, -)$ and  $\mathcal{S'} = (+', \rightarrow', -')$ are two datasets; and $\mathcal{C} = (c^+, c^{\rightarrow}, c^-)$ is a cex-dataset, then the expressions $\mathcal{S} \cup \mathcal{S'}$ and $\mathcal{S} \cup \mathcal{C}$ are understood as the triples $(+ \cup +', \rightarrow \cup \; \rightarrow', - \cup -')$ and $(+ \cup c^+, \rightarrow \cup \; c^{\rightarrow}, - \cup c^- )$ respectively.} Finally, at periodic intervals of the time counter, the DS Loop chooses to refine the dataset $\mathcal{S}$ by replacing it with new datapoints found by calling the \texttt{REFINED\_DATASET} procedure.  Finally, we append the cex-dataset $\mathcal{C}$ to the dataset $\mathcal{S}$, set $I_0$ to $I_{aprx}$, move on to the next iteration of the loop.
%If program control exits the loop, the DS Loop fails and halts.
We next give the details of the \texttt{INITIAL\_DATASET}, \texttt{REFINED\_DATASET} and \texttt{REFINE\_DATASET\_CRITERION} procedures.

\subsection*{Initial Dataset, Refined Dataset, and Dataset Refining Criterion.}
Before diving into how we initialize and refine the datasets, we first provide the rationale behind our approach. We wish to give some intuition on what a `good' dataset should be; ideally, a good dataset would be one which `approximates' the ground positive, ground implication pairs, and ground negative states nicely. We give an example of a poor and good approximation for $+_g$ in Fig \ref{fig:MotivateENets}. More rigorously, a dataset $\mathcal{S}$ is a `good approximation' of $(+_g, \rightarrow_g, -_g)$ if we can give strict upper bounds on the size of any cex-dataset\footnote{The size of a dataset is defined as the total count of its plus states, implication pairs and minus states.} for any $\mathcal{S}$-approximate invariant. To give such upper bounds, we need to first understand the geometric structure of our cex-dataset.
\begin{figure}
    \centering
    \includegraphics[width= 0.7\linewidth]{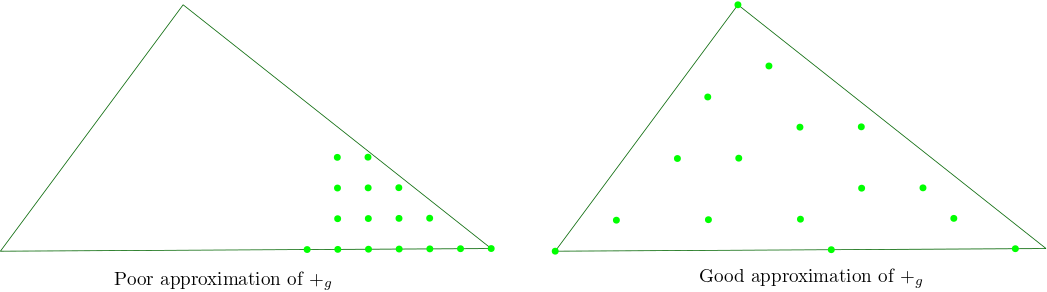}
    \caption{In the examples shown, the ground positive states are represented by all the lattice points within the dark green triangle, while the green dots indicate the sampled positive points. The left example demonstrates a poor approximation of $+_g$ as the sampling is concentrated near a single vertex of the triangle. In contrast, the right example provides a better approximation by distributing the sampled points more evenly across the triangle.}
    \label{fig:MotivateENets}
\end{figure}
\begin{lemma} \label{Lemma: CEX spaces}
    Let $\mathcal{S}$ be a dataset and $I$ be any $\mathcal{S}$-approximate invariant. Then the cex plus, ICE pairs, and cex minus states for $I$ w.r.t the standard LIA Single Loop CHC system satisfy:
\begin{align*}
    c_I^+ &\subseteq \{ \vec{t} \: : \: P(\vec{t}) \} \\
    c_I^\rightarrow &\subseteq \{ (\vec{t}, \vec{t'}) \: : \: B(\vec{t}) \land T(\vec{t}, \vec{t'}) \} \\
    c_I^- &\subseteq \{ \vec{t} \: : \: \neg Q(\vec{t}) \}
\end{align*}
Furthermore, given \(I\) is an \(\mathcal{S}\)-approximate \(d\)-\(c\) LIA invariant,
\begin{enumerate}
    \item If $P$ is a $d_P$-$c_P$ LIA DNF, then $c_I^+$ can be represented as a $(c^d d_P)$-$(d + c_P)$ LIA DNF.
    \item If $B$ is a $d_B$-$c_B$ LIA DNF and $T$ is a $d_T$-piecewise linear integer $r_T$-relation over $B$ where each partition block is a $d_{B_T}$-$c_{B_T}$ LIA DNF, then $HEAD(c_I^\rightarrow)$ can be represented as a $(c^ddd_Bd_Tr_Td_{B_T})$-$(d + c + c_B + c_{B_T})$ LIA DNF. Furthermore, $HEAD(c_I^\rightarrow) \subseteq \{ \vec{t} \: : \: B(\vec{t}) \}$.    
    \item If $Q$ is a $d_Q$-$c_Q$ LIA DNF, then $c_I^-$ can be represented as a $(c_Q^{d_Q}d)$-$(d_Q + c)$ LIA DNF.
\end{enumerate}
\end{lemma}

We defer the proof of Lemma \ref{Lemma: CEX spaces} to Appendix \ref{App: Proofs for Lemmas in Methodology}. Lemma \ref{Lemma: CEX spaces} tells us that for any $\mathcal{S}$-approximate $d$-$c$ LIA invariant $I$, $c_I^+$, $HEAD(c_I^\rightarrow)$, and $c_I^-$ are the union of some fixed number of polytope-lattice points in $P$, $B$, and $\neg Q$ respectively. This observation clearly suggests defining $+$, $HEAD(\rightarrow)$, and $-$ as $\epsilon$-nets of $P$, $B$, and $\neg Q$ with polytopes as ranges as such a choice would bound the size of any polytope-lattice points in the cex-dataset for any approximate invariant.\footnote{To clarify, by definition, an $\epsilon$-net $E$ for a space $X$ with respect to ranges $\mathcal{R}$ ensures that for every range $r \in \mathcal{R}$, $r \cap E = \varnothing \Rightarrow |r| < \epsilon |X|$. In our context, any counterexample space $r$ to an approximate invariant for the dataset $E$ is necessarily disjoint from $E$ by the definition of a `counterexample'. Consequently, by the definition of an $\epsilon$-net, this implies that $r$ has bounded size.} 
Using the number of these polytope-lattice points in the cex datasets, we can further bound the size of the cex-dataset. 

The above reasoning, while sound, presents a practical issue: The \texttt{PARALLEL\_SIMULATED\_ANNEALING} procedure (described in Section \ref{SubSec:SA Loop}) runs a sequence of iterations, and the time complexity of each iteration is directly proportional to the size of the dataset. Consequently, sampling very large datasets would make \texttt{PARALLEL\_SIMULATED\_ANNEALING} very slow. If we set our initial dataset to be \(\epsilon\)-nets with respect to polytope spaces, we end up with a very large initial dataset. This is because, according to Theorem \ref{Thm: HAW}, the size of an \(\epsilon\)-net is positively correlated with the VC dimension of the range space. Additionally, Lemma \ref{Lemma: VCdimPolytope} says that the VC dimension of polytopes is positively correlated with the number of faces they have. Since Lemma \ref{Lemma: CEX spaces} provides only a loose bound on the number of faces in a polytope of the cex-datasets, we can only construct $\epsilon$-nets for polytopes with very large sizes. We get around the aformentioned hurdle by making a crucial observation: Any $n$-dimensional polytope has uniquely inscribed and circumscribing Lowner-John ellipsoids, which are dilations of each other (refer Section \ref{Sec: PrereqEllipsoid}); and the VC dimension of such an $n$-dimensional ellipsoid is not large for modest $n$ (refer Lemma \ref{Lemma: VCdimEllipsoid}). So if we were to use ellipsoids as ranges (i.e. define $+$, $HEAD(\rightarrow)$ and $-$ as $\epsilon$-nets of $P$, $B$, and $\neg Q$ with ellipsoids as ranges), we can bound the size of all inner Lowner-John ellipsoids of all polytope-lattice points in the cex-dataset for any approximate invariant. As the outer Lowner-John ellipsoid is a dilation of the inner Lowner-John ellipsoid, this gives a bound on the size of the outer Lowner-John ellipsoid, which in turn bounds the size of the polytope-lattice points as the outer Lowner-John ellipsoid circumscribes the original polytope. Thus, we can construct a relatively small dataset using $\epsilon$-nets s.t. the size of any cex dataset is directly proportional to $\epsilon$. If we reduce the value of $\epsilon$ periodically in our DS Loop, we obtain a convergence guarantee for the DS Loop.

The final detail to address is the construction of \(\epsilon\)-nets for the required range spaces. By utilizing Theorem \ref{Thm: HAW}, we can construct a randomized \(\epsilon\)-net for \(X\) (given as an LIA DNF) for any range space with VC dimension \(vc\), with a probability of at least \(\delta\), provided we can efficiently sample from polytope-lattice points. We note that this sampling can be efficiently achieved using a combination of techniques—rejection sampling and sampling from solution spaces of Diophantine equations (further details in Appendix \ref{App: Uniform Sampling poly lIA}).\footnote{A common misconception is that we must count the total number of lattice points within a polytope (an expensive task) in order to uniformly sample from it.  It is not necessary; rejection sampling efficiently produces the required samples.} The details of the \texttt{RANDOMIZED\_$\epsilon$\_NET} procedure are deferred to Appendix \ref{App: suppAlg}. 
% \\[1\baselineskip]
%
We now describe how we initialize and refine our dataset, as well as the criteria used for refinement.
% (algorithmic details are deferred to Appendix \ref{App: suppAlg}).

\subsubsection*{Initial Dataset.} For some hyperparameters \(\epsilon_0\) and \(\delta_0\), we set \(+\), \(\text{HEAD}(\rightarrow)\), and \(-\) as randomized \(\epsilon_0\)-nets of \(P\), \(B\), and \(\neg Q\) respectively, w.r.t ellipsoids, each with a probability of at least \(\delta_0\). Finally, \(\rightarrow\) is constructed from \(\text{HEAD}(\rightarrow)\) using the transition relation.

\subsubsection*{Refined Dataset.} We refine the dataset by setting a lower value for $\epsilon$ (we divide the previous value by $2$), and then setting $+$, $HEAD(\rightarrow)$, and $-$ to be $\epsilon$-nets of $P$, $B$, and $\neg Q$ w.r.t ellipsoids as ranges, each with probability at least $\delta_0$. We compute $\rightarrow$ from $\text{HEAD}(\rightarrow)$ as before.

\subsubsection*{Refinement criterion.}
Theoretically, refining the dataset at each timestep would ensure faster convergence of the DS Loop. However, doing so increases the size of the $\epsilon$-net drastically within a few iterations, which slows down the \texttt{PARALLEL\_SIMULATED\_ANNEALING} procedure (as each iteration of this loop requires a single pass over the dataset $\mathcal{S}$). Therefore, a tradeoff between faster convergence of the DS Loop and the faster execution of \texttt{PARALLEL\_SIMULATED\_ANNEALING} must be made. To this end, we refine our dataset at intervals of $t_{\text{refine}}$, a hyperparameter. \\[1\baselineskip]
Based on the discussion above, we formalize the termination guarantee of the DS Loop and also provide a progress measure on the DS Loop by bounding the size of the cex-dataset at timestep \( t \) in Theorem \ref{Thm: DS Loop Conv Guar}. 

\begin{theorem}\label{Thm: DS Loop Conv Guar}
Let \( P \) be a \( d_P \)-\( c_P \) LIA DNF, \( B \) be a \( d_B \)-\( c_B \) LIA DNF, \( Q \) be a \( d_Q \)-\( c_Q \) LIA DNF, and \( T \) be a \( d_T \)-piecewise linear integer \( r_T \)-relation over \( B \) where each partition block is a \( d_{B_T} \)-\( c_{B_T} \) LIA DNF. If \( I(t) \) is any approximate \( d \)-\( c \) LIA invariant at the \( t \)-th iteration of the DS Loop, then with probability \( \delta_0^3 \), we can bound the total number of counterexamples to $I(t)$, denoted by $\mathcal{C}_{I(t)}$, as follows:
\[
|\mathcal{C}_{I(t)}| \le \frac{\epsilon_0}{2^{\lfloor \frac{t}{t_{\text{refine}}} \rfloor}} n^n \left( \frac{1 + \frac{C}{n^{2 - \frac{2}{n+1}}}}{1 - C} \right) \cdot \left( c^d d_P \Lambda(P) + c^d d d_B d_T r_T d_{B_T} \Lambda(B) + c_Q^{d_Q} d \Lambda(\neg Q) \right)
\]
where the constant \( C > 0 \) is defined in Theorem \ref{Thm: Ratio of Ellipses}, and we use the convention that \( \Lambda(X) \) denotes the number of lattice points in \( X \). 
Additionally, with probability at least \( \delta_0^3 \), the DS Loop terminates within
\[
t_{\text{refine}} \log_2 \left( \epsilon_0 n^n \left( \frac{1 + \frac{C}{n^{2 - \frac{2}{n+1}}}}{1 - C} \right) \max (\Lambda(P), \Lambda(B), \Lambda(\neg Q)) \right)
\]
iterations and any approximate invariant found at this final iteration must be the correct loop invariant.
\end{theorem}

We defer the proof of Theorem \ref{Thm: DS Loop Conv Guar} to Appendix \ref{App: Proofs for Lemmas in Methodology}. As a side note, it is possible to omit the \( n^n \left( \frac{1 + \frac{C}{n^{2 - \frac{2}{n+1}}}}{1 - C} \right) \) factor if we were to use polytopes as ranges instead of ellipsoids, thereby achieving a better theoretical guarantee. However, as discussed before, this would significantly slow down \texttt{PARALLEL\_SIMULATED\_ANNE ALING}, making this approach infeasible. We note that our algorithm provides termination guarantees independently of the specifics of the \texttt{VERIFIER} procedure. This independence is not a limitation: it ensures that our termination techniques can be generally applied to any CEGIS algorithm.

%We now provide the details of the \texttt{VERIFIER} procedure.

\subsubsection*{SMT-Verification. }
We first address a crucial question: if refining the dataset by sampling finer $\epsilon$-nets already provides convergence guarantees for the DS Loop, why do we need to sample additional points using the \texttt{VERIFIER} procedure at all? This is because  in \texttt{REFINE_DATASET_CRITERION} we  refine the dataset only at every $t_{\text{refine}}$ iteration of the DS Loop\footnote{As mentioned before, this interval-based refinement prevents the size of the $\epsilon$-net from increasing too drastically within a few iterations, which would otherwise slow down the \texttt{SIMULATED_ANNEALING} procedure.}. If $t \mod t_{\text{refine}} \neq -1$, then without adding cex for an approximate invariant $I$ to the dataset at timestep $t$, $I$ would remain an approximate invariant to the dataset at timestep $t+1$, leading to no progress.  In other words, refining the dataset can be seen as making "global progress"—a process that eventually compels the DS Loop to terminate—while finding counterexamples to an invariant at each iteration represents "incremental refinement"—a process that ensures continuous local progress in the DS Loop iterations.

\begin{algorithm}
\caption{\texttt{VERIFIER}}\label{alg: VERIFIER}
\begin{algorithmic}[1]
\State \textbf{Input: } $I$, $P$, $B$, $T$, $Q$
\State $(\texttt{IS\_FACT\_CLAUSE\_VALID}, c_I^+) \leftarrow \texttt{CHC\_VERIFIER}( \forall \vec{s} \:\; P(\vec{s}) \rightarrow I(\vec{s}) )$
\State $(\texttt{IS\_INDUCTIVE\_CLAUSE\_VALID}, c_I^\rightarrow) \leftarrow \texttt{CHC\_VERIFIER}( \forall (\vec{s}, \vec{s'}) \:\; I(\vec{s}) \land B(\vec{s}) \land T(\vec{s}, \vec{s'}) \rightarrow I(\vec{s'}) )$
\If{$\neg \texttt{IS\_INDUCTIVE\_CLAUSE\_VALID}$}
\State  $c_I^\rightarrow \leftarrow c_I^\rightarrow \cup \texttt{ITERATED\_IMPLICATION\_PAIRS}(c_I^\rightarrow)$ 
\EndIf
\State $(\texttt{IS\_QUERY\_CLAUSE\_VALID}, c_I^-) \leftarrow \texttt{CHC\_VERIFIER}( \forall \vec{s} \:\; I(\vec{s}) \rightarrow  Q(\vec{s}) )$
\State $\texttt{CORRECT} \leftarrow \texttt{IS\_FACT\_CLAUSE\_VALID} \land \texttt{IS\_INDUCTIVE\_CLAUSE\_VALID} \land \texttt{IS\_QUERY\_CLAUSE\_VALID}$
\State \Return{$(\texttt{CORRECT}, (c_I^+, c_I^\rightarrow, c_I^-) )$}
\end{algorithmic}
\end{algorithm}

The algorithm for the \texttt{VERIFIER} procedure is presented in Algorithm \ref{alg: VERIFIER}. This procedure takes as input an approximate invariant \(I\), the precondition \(P\), the loop guard \(B\), the transition relation \(T\), and the postcondition \(Q\). The procedure begins by checking if the fact clause is valid for the approximate invariant $I$ using the \texttt{CHC\_VERIFIER} procedure. In brief, this procedure checks if the input CHC clause is valid and if not, produces a special set of counterexamples called \textit{dispersed cex}. We defer the definition of dispersed cex and the details of the \texttt{CHC\_VERIFIER} procedure momentarily. The \texttt{VERIFIER} procedure then checks the validity of the inductive clause and the query clause in a similar manner, collecting dispersed counterexamples for either if they are not valid. Additionally, if the inductive clause is not satisfied, it computes an additional kind of counterexamples called iterated ICE pairs by calling the \texttt{ITERATED\_IMPLICATION\_PAIRS} procedure (explained momentarily). We now define \textit{Dispersed Cex} and \textit{Iterated Implication pairs}, and give details on how to construct them.
\begin{figure}
    \centering
    \includegraphics[scale=0.4]{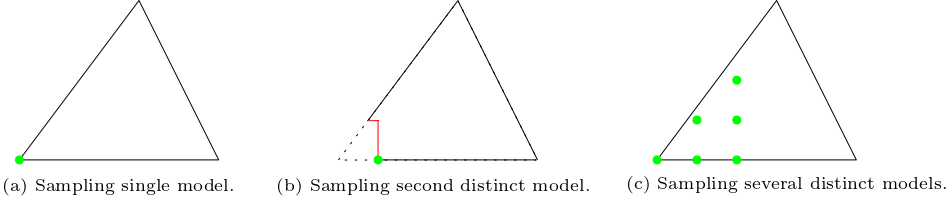}
    \caption{In the above figure, we aim to probe an SMT solver for distinct LIA models within the black triangle. Figure (a) illustrates the plot of the first model, which is a vertex of the triangle. To obtain a second distinct model, we add the constraint \(\big( (x,y) \neq (x_f, y_f) \big) \equiv \big( x < x_f \lor x > x_f \lor y < y_f \lor y > y_f \big), \) where \((x_f, y_f)\) are the coordinates of the first model. This constraint modifies the sampling space, with the new boundaries drawn in red. Sampling from this space results in a new vertex as a model, but unfortunately, it is very close to the first model, as depicted in Figure (b). Finally, Figure (c) shows the result of sampling many distinct counterexamples from the black triangle.}
    \label{fig:sampleDistinceCEX}
    \vspace{-1.0 em}
\end{figure}

\subsubsection*{Dispersed CEX. }
At a high level, given as input a CHC clause, we want to check if the input clause is valid, and if not, make `good' local progress by finding a set of `good' counterexamples to the CHC clause from the state space. However, if we repeatedly probe an SMT solver for multiple distinct counterexamples\footnote{To probe the SMT solver for distinct cex, you need to probe the SMT solver for models of the clause $\forall \vec{s} \: \neg (\texttt{CHC}(\vec{s})) \land \left(\bigwedge_{\vec{t} \in \texttt{CEX}} (\vec{s} \neq \vec{t})\right)$, where \texttt{CEX} represents the set of already probed cex. If you simply probe for models of $\forall \vec{s} \: \neg (\texttt{CHC}(\vec{s}))$, the SMT solver, being a deterministic machine, will keep returning the same cex.}, we encounter a problem: SMT solvers find models for formulas from LIA theory using simplex-based methods \cite{dutertre2006fast}. These methods find models that are vertices of the input polytope space, and tools implementing such methods use deterministic heuristic choices to determine which vertex to sample. Consequently, after disallowing a previously sampled model, new models found are the new vertices introduced by the addition of new predicates, which are very close to the previous model. If we sample many such distinct models by this method, we obtain a set of points all clustered together, not accurately reflecting the nature of the cex-space. This issue is demonstrated in Figure \ref{fig:sampleDistinceCEX}. To achieve a better sampling from the cex-space, we aim to sample cex that are distant from each other - we label such cex as \textit{dispersed cex}. We can sample dispersed cex by ensuring the new counterexample to sample is at least some fixed distance away from all previously sampled counterexamples. This is achieved by adding the predicate $\bigwedge_{\vec{t} \in \texttt{CEX}} ||\vec{s} - \vec{t}||_1 \ge d_0$, where \texttt{CEX} stores the already sampled cex models, $d_0 > 0$ is a hyperparameter, and $||\vec{u}||_1$ represents the L1-norm of $\vec{u}$.\footnote{If $x \in \mathbb{R}^n$, then the L1-norm of $x$ is given by $\sum_{i=1}^n |x_i|$. Modern SMT solvers provide support for representing the modulus operator as an LIA formula by using techniques described in \cite{fischetti2018deep}.} We use the L1-norm as our distance metric because it is one of the few norms that can be represented as an LIA formula. We defer the details of the \texttt{CHC\_VERIFIER} procedure for Appendix \ref{App: suppAlg}.

\subsubsection*{Iterated Implication Pairs. }
Before defining iterated ICE pairs, we first discuss their importance. Sampling dispersed counterexamples ensures good local progress for the fact and query clauses, but this is not the case for the inductive clauses. This distinction arises from the unique nature of implication pairs, which can be satisfied by an approximate invariant in one of two ways: either the head of the implication pair is not satisfied by the invariant, or the tail is satisfied by the invariant. Due to these two possibilities, many sampled implication pairs remain inactive during a single DS iteration, and local progress for satisfying the inductive clause is slow.

This issue is demonstrated in Figure \ref{fig:iipMotivate}. Both figures show successive iterations of the DS Loop. The state space is depicted by the black rectangle, with implication pairs marked by blue arrows and ICE pairs marked by purple arrows. The cyan rectangle represents the ground truth invariant, while the approximate invariants are shown in gray. In the left diagram, all six implication pairs are satisfied by the approximate invariant: the first two because both heads and tails are satisfied, and the remaining four because neither heads nor tails are satisfied. However, the ground truth invariant indicates that the last four pairs need both heads and tails included, making them inactive until the boundary of the current approximate invariant shifts to the right.

$\textit{ICE}_1$ is a counterexample for the current approximate invariant. In the next run, the new approximate invariant (right) satisfies the six implication pairs and $\textit{ICE}_1$, but introduces a new counterexample, $\textit{ICE}_2$, which was not previously a counterexample. This process keeps continuing: for each new DS iteration, the boundary of the approximate invariant moves exactly \textbf{one unit} to the right, requiring new ICE pairs. This cycle persists until the boundary includes the heads of $\rightarrow_3$ or $\rightarrow_5$, thereby activating these implication pairs; hence it is very slow progress in satisfying the inductive clause. We can resolve this issue by using iterated ICE pairs, which we now define.

\begin{definition}
For any natural number \( k \), a \( k \)-iterated implication pair is a pair \( (\vec{x}, \vec{y}) \) where there exist $k-1$ vectors \( \vec{z}_1, \vec{z}_2, \ldots, \vec{z}_{k-1} \in \mathcal{D}_{\text{state}} \) such that \( \bigwedge_{i=0}^{k-1} B(\vec{z}_i) \land T(\vec{z}_i, \vec{z}_{i+1}) \), where \( \vec{z}_0 = \vec{x} \) and \( \vec{z}_k = \vec{y} \). An iterated implication pair is then a $k$-iterated implication pair for some fixed $k \in \mathbb{N}$. An iterated ICE pair is just an iterated implication pair which is also a cex to some approximate invariant.
\end{definition}

\begin{figure}
    \centering
    \includegraphics[scale=0.39,width=0.8\textwidth]{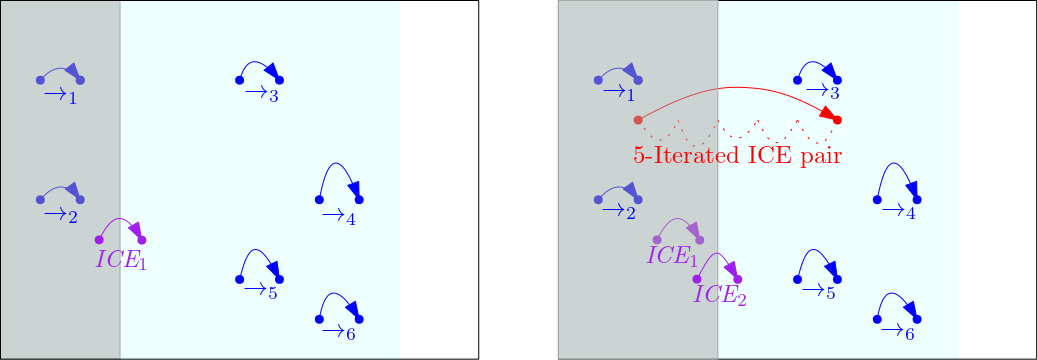}
    \caption{Iterations of DS Loop show the challenge of satisfying implication pairs and the benefit of  ICE pairs. %\gtan{We can cut down some slack space around the border to save some space.}
    }
    \vspace{-1.0 em}
    \label{fig:iipMotivate}
\end{figure}

An approximate invariant satisfies an iterated implication pair, in the same way, it satisfies an implication pair: either by not satisfying its head or by satisfying its tail. Upon closer examination, an iterated implication pair generalizes the concept of loop unrolling: a $k$-iterated implication pair effectively unrolls the loop $k$ times, starting from the head state. However, it is more general because the starting state does not need to be an actual reachable state from the loop. Iterated implication pairs solve the previously mentioned issue, as the approximate invariant must make significant progress to satisfy a previously unsatisfied iterated implication pair, thereby ensuring good local progress. For example in Fig \ref{fig:iipMotivate}, the iterated ICE pair (red arrow in the right diagram) solves the issue of slow progress in one step, as neither of the approximate invariants satisfy it, ensuring progress, with $\rightarrow_3$ or $\rightarrow_5$ becoming active after a single DS Loop iteration. We sample iterated implication pairs by unrolling the transition relation in the \texttt{ITERATED_IMPLICATION_PAIRS} procedure, whose details we defer to Appendix \ref{App: suppAlg}.

%Now that we have covered the Datapoint Sampling Algorithm, in the next section we provide details of the \texttt{PARALLEL\_SIMULATED\_ANNEALING} procedure. 

\subsection{Simulated Annealing (SA) Loop.}\label{SubSec:SA Loop}

We use parallel syntactic-based simulated annealing to find an $\mathcal{S}$-approximate invariant given a dataset $\mathcal{S}$ and initial invariant $I_0$. Before detailing how we modify simulated annealing to allow for parallelism, we first describe the basic parameters for our simulated annealing search: (i) state space, (ii) neighborhood relation, (iii) cost function, and (iv) initial invariant (starting point).

\subsubsection*{State Space.}

Our invariant space \(\mathcal{X}\) is parameterized by a parameter \(k\), which represents the maximum coefficient bound. If \(\rho(\mathcal{D}_{\text{state}})\) denotes the radius of \(\mathcal{D}_{\text{state}}\), then let \(k' := k \sqrt{n} \rho(\mathcal{D}_{\text{state}})\) represent the maximum constant bound. Given hyperparameters \(d\) and \(c\), our goal is to find \(d\)-\(c\) LIA invariants where the invariant space \(\mathcal{X}(k)\) is given by:
\[ 
\mathcal{X}(k) = \left\{ \Big( \forall \vec{x} \in \mathcal{D}_{\text{state}} \: \bigvee_{i=1}^d \bigwedge_{j=1}^c \vec{w}_{ij} \cdot \vec{x} \le b_{ij} \Big) \: : \:
    \forall i \forall j \: 0 < \|\vec{w}_{ij}\|_{\infty} \le k \land |b_{ij}| \le k'  \right\} 
\]
where \(|| \cdot ||_{\infty}\) represents the \(L_\infty \) norm\footnote{If \(\vec{x} \in \mathbb{R}^n\), then \( || \vec{x}||_{\infty} := \max_{1 \le i \le n} |x_i|\)}. Elements of the invariant space \(\mathcal{X}(k)\) are called candidate invariants. We now address the constraints we place on the coefficients and constants. We firstly note that our constraint on the constants of the predicates does not lose any generality beyond the restriction on the maximum coefficient value; we formally state and prove this statement in Lemma \ref{Lemma: SA constant bound} in Appendix \ref{App: Proofs for Lemmas in Methodology}. To address any loss of generality by the constraint we place on the maximum coefficient bound, we employ parallelism to do the SA search on spaces with different $k$ values.

\subsubsection*{Neighborhood Relation.}
For the neighborhood relation \(\mathcal{N}\), we define two candidate invariants \(I \equiv \Big( \forall \vec{x} \in \mathcal{D}_{\text{state}} \: \bigvee_{i=1}^d \bigwedge_{j=1}^c \vec{w}_{ij} \cdot \vec{x} \le b_{ij} \Big)\) and \(I' \equiv \Big( \forall \vec{x} \in \mathcal{D}_{\text{state}} \: \bigvee_{i=1}^d \bigwedge_{j=1}^c \vec{w'}_{ij} \cdot \vec{x} \le b'_{ij} \Big)\) as neighbors if they differ by exactly one predicate, with specific restrictions on how they differ: \(I\) and \(I'\) are neighbors if:
\begin{itemize}
    \item There exists an index pair \((i_0, j_0)\) such that either:
        \begin{itemize}
            \item The \(L_\infty\) norm of the difference in the coefficient vectors is \(\pm 1\), i.e., \(\|\vec{w}_{i_0 j_0} - \vec{w}'_{i_0 j_0}\|_{\infty} = 1\), or
            \item The constants differ by \(\pm 1\), i.e., \(b_{i_0 j_0} = b'_{i_0 j_0} \pm 1\).
        \end{itemize}
    \item For all \(i \neq i_0\) and \(j \neq j_0\), \(\vec{w}_{ij} = \vec{w}'_{ij}\) and \(b_{ij} = b'_{ij}\).
\end{itemize}

If the differing predicate is due to a difference in the coefficient vector, then a transition from \(I\) to \(I'\) is called a \textit{coefficient transition} whereas if the differing predicate is due to a difference in the constant, such a transition is called a \textit{constant transition}. It is easy to see that each candidate invariant in the search space \(\mathcal{X}(k)\) is reachable from another using the neighborhood relation.

\subsubsection*{Cost Function.}
Given a dataset $\mathcal{S}$, we aim to define a cost function on $\mathcal{X}(k)$ such that any $\mathcal{S}$-approximate invariant $I^*$ is an optimal solution, i.e., $c(I^*) = 0$. A possible approach is to define the cost as the ratio of the number of data points unsatisfied by a candidate invariant to the total number of data points sampled \cite{garg2016learning}. 
However, this method often fails to provide guidance on which transition to prefer. Figure \ref{fig:SAcostmotivate} illustrates this issue.

\begin{SCfigure}
    \centering
    \includegraphics[scale=0.31]{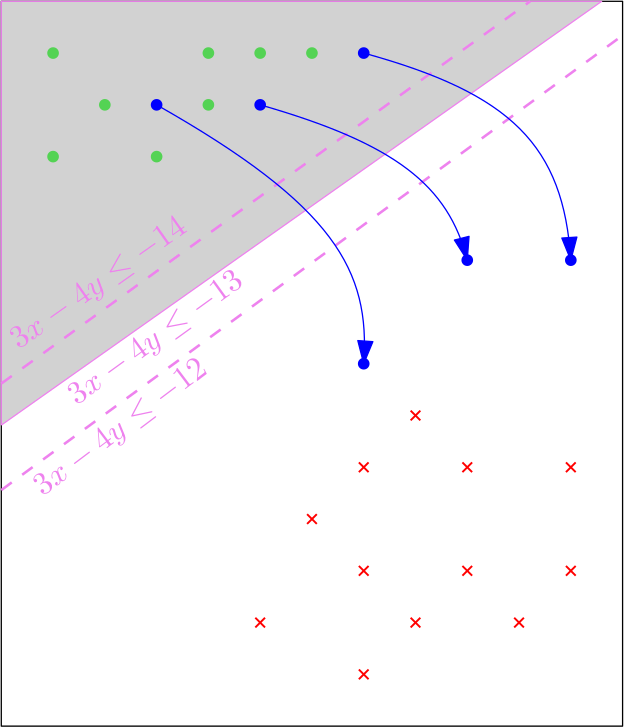}
    \caption{This plot highlights the limitation of using the ratio of unsatisfied points by a candidate invariant to the total sampled points as its cost. The black box outlines the state space \(\mathcal{D}_{state}\). Green dots represent positive points, red crosses denote negative points, and blue arrows indicate implication pairs. The current candidate invariant is depicted by the gray region, defined by a single predicate: \(3x - 4y \leq -13\). Two potential neighboring candidate invariants are shown with dashed lines: \(3x - 4y \leq -14\) and \(3x - 4y \leq -12\). Using the ratio of unsatisfied points as our cost function results in the same cost for the current candidate invariant and both of its neighbors. This occurs because all three satisfy all positive and negative points but fail to satisfy any implication pairs. With a constant cost function, the SA search lacks incremental feedback to determine whether transitioning to either neighbor is advantageous. However, the plot demonstrates that transitioning to \(3x - 4y \leq -12\) is preferable, as it brings us closer to achieving an approximate invariant.}
    \label{fig:SAcostmotivate}
\end{SCfigure}

A straightforward remedy is to use the average set distance of the candidate invariant from all data points instead of simply counting unsatisfied data points. Let $\delta : \mathcal{X}(k) \times \mathcal{D}_{state} \rightarrow \mathbb{R}_{\ge 0}$ denote a set distance function for candidate invariants (specific details to follow shortly). Then, the cost function $c(I)$ for a candidate invariant $I \in \mathcal{X}(k)$ can be defined as:
\[ c(I) = \frac{1}{3} \left( \frac{\sum_{\vec{p} \in +} \delta (I, \vec{p}) }{\sum_{\vec{p} \in +} 1} + \frac{\sum_{(\vec{h},\vec{t}) \in \rightarrow} \min (\delta (\neg I, \vec{h}), \delta (I, \vec{t})) }{\sum_{(\vec{h},\vec{t}) \in \rightarrow} 1} + \frac{\sum_{\vec{n} \in -} \delta (\neg I, \vec{n}) }{\sum_{\vec{n} \in -} 1} \right) \]
We provide an explanation for the above formula: Candidate invariants satisfy a plus-point by including it, and thus, the term $\delta(I, \vec{p})$ quantifies how much a plus-point $\vec{p} \in \mathcal{D}_{state}$ fails to be satisfied by $I$. Similarly, candidate invariants satisfy an implication pair by either excluding the head or including the tail. Therefore, $\min(\delta(\neg I, \vec{h}), \delta(I, \vec{t}))$ measures the extent to which an implication pair $(\vec{h}, \vec{t}) \in \mathcal{D}_{state}^2$ is unsatisfied by $I$.
Lastly, candidate invariants satisfy a minus point by excluding it. Hence, $\delta(\neg I, \vec{n})$ indicates how much a minus-point $\vec{n} \in \mathcal{D}_{state}$ is satisfied by $I$. 

We next give the definition for $\delta$.
Note that using the Euclidean set distance as \(\delta\) poses a problem, as computing this exact set distance would require an ILP solver, which is time-expensive. To address this, we introduce an approximation to the set distance function tailored specifically for our candidate invariants, which are disjunctions of polytope-lattice points in \(\mathbb{R}^n\). This approximation, defined solely for sets that conform to the shape of our candidate invariants, is called the \textit{approximate polytope set distance} \(\delta_{approx}\). For \(I \in \mathcal{X}(k)\) such that \(I \equiv \bigvee_{j=1}^d \bigwedge_{i=1}^c \vec{w_{ij}}\vec{x} \le b_{ij}\) and \(\vec{s} \in \mathcal{D}_{state}\), we define:

\[ \delta_{approx}(I , \vec{s}) := \min_{1 \le j \le d} \frac{1}{c} \cdot \sum_{i=1}^c \frac{(\vec{w_{ij}}^T \vec{s} - b_{ij})^+}{||\vec{w_{ij}}||} \]
$\delta_{approx}$ computes the average normalized distance of \(\vec{s}\) from the boundary hyperplanes defined by the cubes of \(I\), and takes the minimum value over all cubes (illustrated in Figure \ref{fig:costFuncPlot} in Appendix \ref{App: SuppFigures}). We state the properties of $\delta_{approx}$ in Lemma \ref{Lemma: Properties of delta}.

\begin{lemma} \label{Lemma: Properties of delta}
For all \(\vec{s} \in \mathcal{D}_{state}\) and all \(I \in \mathcal{X}(k)\), \(\delta_{approx}(I,\vec{s})\) satisfies the following properties:
\begin{itemize}
\item Non-negativity: \(\delta_{approx}(I,\vec{s}) \ge 0\). Furthermore, \(\delta_{approx}(I,\vec{s}) = 0 \Leftrightarrow \vec{s} \in I\).
\item Under approximation to actual set distance from \(\vec{s}\) to \(I\): If \(d_I(\vec{s}) := \inf \{ || \vec{s} -  \vec{t}|| \: : \:  \vec{t} \in I \}\) represents the actual set distance of \(\vec{s}\) from \(I\), then \( \delta_{approx}(I,\vec{s}) \le d_I(\vec{s})\).
\item Efficiently (linear-time) computable: \(\delta_{approx}(I,\vec{s})\) can be computed in \(O(ncd)\) time.
\end{itemize}    
\end{lemma}

We defer the proof of Lemma \ref{Lemma: Properties of delta} to Appendix \ref{App: Proofs for Lemmas in Methodology}. Our approximate polytope set distance \(\delta_{approx}\) is designed not only for use in simulated annealing applications but also as an efficiently computable approximation for polytope set distances given their H-representation\footnote{The H-representation of a polytope in \(\mathbb{R}^n\) defines the polytope as the intersection of a finite number of closed half-spaces, typically represented as \(\bigwedge_{i=1}^c \vec{w_i} \vec{x} \le b_i\). In contrast, the V-representation describes a polytope as the convex hull of a finite set of points, specifically the corner points of the polytope. Converting polytopes between H-representation and V-representation can involve significant computational expense. Our cost function leverages the fact that our polytopes are given in H-representation.}.

Setting \(\delta \longleftarrow \delta_{approx}\) makes sense, but it introduces a problem. Most coefficient transitions change the average set distance of the dataset to the candidate invariant so significantly that their transition probability becomes almost zero. Consequently, these transitions do not occur, leading to constant coefficients for all predicates during the SA search. We demonstrate this issue in Figure \ref{fig:costFunctionMotivate2} in Appendix \ref{App: SuppFigures}. To address this, we use a normalization function \(F: \mathbb{R}_{\ge 0} \rightarrow \mathbb{R}_{\ge 0}\) and set \(\delta\) to be the average normalized set distance of the candidate invariant from the dataset. The normalization function \(F\) is defined as:
\[
F_{\alpha,\beta}(x) = \begin{cases} 
\frac{x}{\beta} & , x \le \alpha \\
\frac{\alpha}{\beta} - 1 + \frac{2}{1 + e^{-\frac{2(x-\alpha)}{\beta}}} & , x > \alpha 
\end{cases}
\]
for some hyperparameters \(\alpha > 1\) and \(\beta \ge 1\). A plot of this function is given in Figure \ref{fig:normalizer} in Appendix \ref{App: SuppFigures} and some properties of $F_{\alpha, \beta}$ are given in Lemma \ref{Lemma: normProp} in Appendix \ref{App: Proofs for Lemmas in Methodology}. Finally, we set \(\delta \leftarrow F_{\alpha, \beta} \circ \delta_{approx}\) for some hyperparameters \(\alpha\) and \(\beta\).

\subsubsection*{Initial Invariant. }
We initialize the candidate invariant for the \texttt{PARALLEL\_SIMULATED\_ANNEALING} procedure call in the first iteration of the DS Algorithm. Subsequent iterations use the approximate invariant from the previous iteration for initialization. 

For the first DS Loop iteration, we initialize the invariant using the \texttt{INITIAL\_INVARIANT} procedure, which employs the small-constant heuristic \cite{pldi/2016/PadhiSM}. This heuristic prioritizes candidate invariants with small absolute constant values (see Appendix~\ref{App: suppAlg} for more details). We sample many $d$-$c$ LIA DNFs from which  the least cost is then chosen as the initial invariant.

% which are likely to approximate true invariants. Details of this procedure are deferred to , where many such $d$-$c$ LIA DNFs are sampled by uniformly selecting coefficient vectors and setting their constants to $0$. The sampled LIA DNF with\\[1\baselineskip]
%
% We now give the details of how we modify the simulated annealing search to allow for parallelism.

\subsubsection*{Parallel Simulated Annealing Search.}
The advantages of using parallel SA search are multifold: (1) Since our invariant space $\mathcal{X}$ is parameterized by the maximum coefficient bound $k$, we run multiple threads on search spaces $\mathcal{X}(k)$ for different values of $k$; (2) Running multiple threads on $\mathcal{X}(k)$ for the same value of $k$ also helps, as it increases the probability of finding the approximate invariant. We give the algorithm for \texttt{PARALLEL\_SIMULATED\_ANNEALING} in Algorithm \ref{alg:Parallel SA} in Appendix \ref{App: suppAlg}. The algorithm follows the basic simulated annealing procedure but includes additional logic to periodically verify global progress, ensuring efficient parallel exploration of the search space within a maximum iteration limit. We now present the probabilistic guarantee of success for our parallel SA search in Theorem \ref{thm: OurSAguarantee}.
% , based on Theorem \ref{thm: SAguarantee}.  

\begin{theorem}  \label{thm: OurSAguarantee}
Let \(\mathcal{S} = (+, \rightarrow, -)\) be the current dataset. Define:
\[
\kappa_\infty(\mathcal{S}) := \frac{1}{3} \left( \frac{\sum_{\vec{p} \in +} ||\vec{p}||_\infty}{|+|}  + \frac{\sum_{(\vec{h}, \vec{t}) \in \rightarrow} \max( ||\vec{h}||_\infty , ||\vec{t}||_\infty) }{|\rightarrow|} + \frac{\sum_{\vec{n} \in -} ||\vec{n}||_\infty}{|-|} \right)
\]
as the averaged \(L_\infty\) measure of the dataset \(\mathcal{S}\), and
\[
\lambda(\rightarrow) := \frac{\sum_{(\vec{h}, \vec{t}) \in \rightarrow} ||\vec{h} - \vec{t}|| }{|\rightarrow|}
\]
as the average pairwise distance of \(\rightarrow\).
 Using the terminology from Theorem \ref{thm: SAguarantee}, we have the following bounds on the parameters of our invariant space \(\mathcal{X}(k)\), for any value of $k \in kList$:
\begin{align*}
    cdk&\sqrt{n}\rho(\mathcal{D}_{state}) + cdkn \le r \le cdk\sqrt{n}\rho(\mathcal{D}_{state}) + cdkn + cd, \\
    w &= \frac{1}{cd(2n+2)},~~~~~
    L \le \frac{1}{\beta ck\sqrt{n}} \left(1 + \kappa_\infty(\mathcal{S})\right) + \frac{1}{3\beta} \lambda(\rightarrow),~~~~~
    \delta \ge \frac{1}{3k\beta\sqrt{n}|\mathcal{S}|}
\end{align*}
For every $k \in kList$, if $T_0 \ge \left( \frac{d\rho(\mathcal{D}_{state}) + d\sqrt{n}}{\beta} + \frac{d}{\beta k \sqrt{n}} \right) \left(1 + \kappa_\infty(\mathcal{S})\right) + \frac{cdk\sqrt{n}\rho (\mathcal{D}_{state}) + cdkn + cd}{3\beta} \lambda(\rightarrow)$ and $t_{max} \ge k\sqrt{n}cd\rho(\mathcal{D}_{state}) + kncd + cd$, we find a $\mathcal{S}$-approximate invariant in $\mathcal{X}(k)$, if one exists, with probability at least $1 - \frac{A}{\left\lfloor \frac{t_{max}}{k\sqrt{n}cd\rho(\mathcal{D}_{state}) + kncd + cd} \right\rfloor^{exp}}$
where $exp = \min \left( \frac{1}{(cd(2n+2))^{kcd\sqrt{n}(\rho(\mathcal{D}_{state}) + \sqrt{n})}}, \frac{1}{|\mathcal{S}| \left(\frac{3}{c} \kappa_\infty(\mathcal{S}) + k\sqrt{n} \lambda(\rightarrow) + \frac{3}{c}\right)} \right). $    
\end{theorem}

We defer the proof of Theorem \ref{thm: OurSAguarantee} to Appendix \ref{App: Proofs for Lemmas in Methodology}. Theorem \ref{thm: OurSAguarantee} is significant because it links the success of the SA search to the parameters provided by the DS Algorithm. Specifically, it identifies the average \(L_\infty\) measure of the dataset \(\mathcal{S}\), denoted as \(\kappa_\infty(\mathcal{S})\), and the average pairwise distance of \(\rightarrow\), denoted as \(\lambda(\rightarrow)\), as key parameters influencing the success of the SA search for \(\mathcal{S}\)-approximate invariants. While Theorem \ref{thm: OurSAguarantee} provides a probabilistic guarantee of success for our simulated annealing search, the computed probability of success is infinitesimally small, and the required initial temperature and maximum iterations are excessively large. This is because the analysis of success for a simulated annealing search in \cite{mitra1986convergence} was conducted for general search spaces and cost functions, and does not consider the specific details of our search. In practice, we achieve higher success probabilities in our experiments and determine better \(T_0\) bounds using the research from \cite{ben2004computing} and \cite{potter20150} (refer Algorithm \ref{alg: InitTemp}). 

\subsection{Putting it all together.} \label{Subsec: Putting it all together}
Finally, we summarize our algorithm and give the overall success guarantee.
Our algorithm employs a CEGIS-based method to find loop invariants.
For the invariant search, we use parallel simulated annealing on a parameterized search space, leveraging parallelism to explore different instantiations of the search space with varying parameter values. We introduce a novel, efficiently computable approximation to the polytope-set distance metric, which serves as the cost function for our simulated annealing search. We provide probabilistic guarantees for the convergence of simulated annealing search. We use \(\epsilon\)-nets w.r.t ellipsoids to guarantee the convergence of verifying the candidate invariants. Besides, we utilize clever SMT counterexample sampling and iterated implication pairs to make incremental local progress in our verification. Finally, the overall guarantees of our algorithm are presented below (proof in Appendix \ref{App: Proofs for Lemmas in Methodology}).

\begin{theorem} \label{Thm:FullSuccessGuarantees}
Let \( P \) be a \( d_P \)-\( c_P \) LIA DNF, \( B \) be a \( d_B \)-\( c_B \) LIA DNF, \( Q \) be a \( d_Q \)-\( c_Q \) LIA DNF, and \( T \) be a \( d_T \)-piecewise linear integer \( r_T \)-relation over \( B \) where each partition block is a \( d_{B_T} \)-\( c_{B_T} \) LIA DNF. If there exists an invariant \(I\) for the standard LIA Single Loop CHC system such that \(I \in \mathcal{X}(k)\) for any $k \in kList$, then with probability at least \( \delta_0^3 \), \toolname{} necessarily terminates within
\[
T:=  t_{\text{refine}} \log_2 \left( \epsilon_0 n^n \left( \frac{1 + \frac{C}{n^{2 - \frac{2}{n+1}}}}{1 - C} \right) \max (\Lambda(P), \Lambda(B), \Lambda(\neg Q)) \right)
\]
iterations \textbf{and}, \\ \toolname{} finds a loop invariant with a probability of at least \(\delta_0^3 p_0^T \), where 
\( 
p_0 := 1 - \frac{A}{\left\lfloor \frac{t_{max}}{k\sqrt{n}cd\rho(\mathcal{D}_{state}) + kncd + cd} \right\rfloor^{\exp}} 
\)
and 
\(
\exp = \min \left( 
\frac{1}{(cd(2n+2))^{kcd\sqrt{n}(\rho(\mathcal{D}_{state}) + \sqrt{n})}}, 
\frac{1}{\left( \tau\left( \frac{\epsilon_0}{\left\lfloor \frac{T}{t_{\text{refine}}} \right\rfloor}, \delta \right) + 3T \cdot \text{cex}_{\text{max}} \right) 
\left(\frac{3}{c} \rho(\mathcal{D}_{\text{state}}) + k\sqrt{n} \rho(\mathcal{D}_{\text{state}}) + \frac{3}{c}\right)} 
\right)
\)
where \(\tau(\epsilon, \delta)\) represents the minimum number of points to sample to obtain randomized \(\epsilon\)-nets of polytopes w.r.t ellipsoids with a probability of at least \(\delta_0\) using Theorem \ref{Thm: HAW} and $\delta_0$ is the probability to sample an $\epsilon$-net.
\end{theorem}
Although the theoretical probabilistic guarantee of success for \toolname{} is infinitesimally small, this algorithm represents a significant advance for several reasons: (i) it introduces the first and only data-driven invariant automation algorithm with theoretical guarantees on algorithm termination and success of finding an invariant for arbitrarily large invariant spaces; (ii) while these theoretical guarantees, based on general state spaces, are conservative, \toolname{} consistently performs significantly better in practice across various benchmarks, demonstrating robustness beyond its theoretical limits; and (iii) the CEGIS termination guarantee provided by \toolname{} can be generalized to all CEGIS  algorithms.

\begin{comment}
    Mean convergence probability - There are 4 runs per benchmark - compute a probability of convergence for each benchmark (average of scores, 1 represents convergence, 0 means no convergence), then take average of these probabilities for each folder (also include benchmarks with probability mean as 0 i.e. which never converged) 
\end{comment}

\section{Evaluation and Comparative Analysis.}
\label{Sec:Tool and Experiments}
% \subsection{Experiments and Implementation.}
We have implemented \toolname{} in Python.
% (see Figure \ref{fig:Implementation} in Appendix~\ref{App: SuppFigures}).
Our front end converts an input C program to an intermediate representation (IR) defined as a special Python class, specifically for standard LIA single loop Constrained Horn Clause (CHC) systems. The front end is semi-automated, requiring user input to identify the formulae \( P \), \( B \), \( T \), and \( Q \) in the original C code for conversion to our IR.
Finally, \toolname{} identifies an invariant by using the IR representation of the standard LIA single loop CHC system as input.

% \begin{wrapfigure}{r}{0.45\textwidth}
%     \centering
%     \includegraphics[width=\linewidth]{Figures/SAvsDS.png}
%     \caption{The plot above shows the time taken by the simulated annealing loop vs. the time taken by the datapoint sampling loop for successful searches of \toolname{} on the benchmarks (best runs).}
%     \label{fig:SAvsDS}
%     \vspace{-1.0 em}    
% \end{wrapfigure}
We evaluate \toolname{} using SV-COMP 2023 benchmarks\footnote{We note that the benchmarks in SV-COMP 2023 are curated to favor algorithms that employ static techniques, such as interpolation and bounded model checking, to prove safety. As a result, the state-of-the-art tool \gspacer{} outperforms our approach on these benchmarks. However, in Section \ref{Sec: ExpComp}, we identify a subclass of standard LIA single-loop CHC systems where our approach significantly outperforms \gspacer{} across 11 constructed benchmarks within this subclass. To further demonstrate efficiency, we also evaluate all benchmarks from the four folders in SV-COMP that fit our format (i.e., 67 benchmarks), ceasing once a sufficiently diverse set of instances has been considered.} ~\cite{svbenchmark}. Table~\ref{tab:results} (left) summarizes the results of \toolname{} and the baseline tools on a subset of benchmarks from the \textit{loop\_lit}, \textit{loop\_new}, \textit{loop\_simple}, and \textit{loop\_zilu} that satisfy our problem constraints (i.e., we discard benchmarks containing multiple loops, arrays, or floating point variables). We run the benchmarks with a timeout of $4$ hours. 

% \begin{table}[t]
% \centering
% \caption{Comparative Results for \toolname{} vs \gspacer{}, \la{} and \icedt{} on SV-COMP 2023. The numbers show the number of benchmarks converged for each tool in each folder.}
% \begin{tabular}{lccccc}
% \toprule
% \textbf{Folder Name} & \toolname{} & \gspacer{} & \la{} & \icedt{} & \textit{Total} \\
% \midrule
% \textit{loop\_lit}    & 6  & 9  & 9  & 4  & 11\\
% \textit{loop\_new}    & 2  & 0  & 0  & 0  & 2\\
% \textit{loop\_zilu}   & 48 & 50 & 48 & 31 & 53 \\
% \textit{loop\_simple} & 1  & 1  & 1  & 1  & 1 \\
% \textit{Total}        & 57 & 60 & 58 & 36 & 67 \\
% \bottomrule
% \end{tabular}
% \label{tab:results}
% \end{table} 

% \begin{table}[t]
% \centering
% \caption{Mean convergence probability of \toolname{} on SV-COMP 2023 Benchmarks}
% \begin{tabular}{lcc}
% \toprule
% \textbf{Folder Name} & \textbf{Mean Convergence Probability}\\
% \midrule
% \textit{loop\_lit} & 0.175\\
% \textit{loop\_new} & 1\\
% \textit{loop\_zilu} & 0.30 \\
% \textit{loop\_simple} & 1\\
% \textit{Total} &0.31\\
% \bottomrule
% \end{tabular}
% \label{tab:results2}
% \end{table}

\begin{figure}[!tb]
    \begin{minipage}[ht]{.5\linewidth}
    \centering
        \resizebox{\textwidth}{!}{%
        \begin{tabular}{lccccc}
        \toprule
        \textbf{Folder Name} & \toolname{} & \gspacer{} & \la{} & \icedt{} & \textit{Total} \\
        \midrule
        \textit{loop\_lit}    & 6  & 9  & 9  & 4  & 11\\
        \textit{loop\_new}    & 2  & 0  & 0  & 0  & 2\\
        \textit{loop\_zilu}   & 48 & 50 & 48 & 31 & 53 \\
        \textit{loop\_simple} & 1  & 1  & 1  & 1  & 1 \\
        \textit{Total}        & 57 & 60 & 58 & 36 & 67 \\
        \bottomrule
        \end{tabular}
        }
    \end{minipage}
    \hspace{1.0 em}
    \begin{minipage}[ht]{.38\linewidth}
        \resizebox{\textwidth}{!}{%
        \begin{tabular}{lcc}
        \toprule
        \textbf{Folder Name} & \textbf{Mean Convergence Probability}\\
        \midrule
        \textit{loop\_lit} & 0.175\\
        \textit{loop\_new} & 1\\
        \textit{loop\_zilu} & 0.30 \\
        \textit{loop\_simple} & 1\\
        \textit{Total} &0.31\\
        \bottomrule
        \end{tabular}
        }
    \end{minipage}
    \caption{(left) Comparative Results for \toolname{} vs \gspacer{}, \la{} and \icedt{} on SV-COMP 2023. The numbers show the number of benchmarks converged for each tool in each folder. (Right) Mean convergence probability of \toolname{} on SV-COMP 2023 Benchmarks.}
    \label{tab:results}
    \label{tab:results2}
    \vspace{-1.0 em}
\end{figure}

Out of 67 benchmarks, \toolname{} successfully finds the ground truth invariant in 57 cases, yielding a success rate of 85.0\%. We note that as \toolname{} is a probabilistic tool, it can have different outcomes on the same benchmark for different runs. The results for \toolname{} present in Table \ref{tab:results} (left) show the best outcome for each benchmark on multiple runs. Table~\ref{tab:results2} (right) shows the mean convergence probability of \toolname{} on the benchmarks.

\subsection{Comparison with Other Tools.} \label{Sec: ExpComp}
We compare \toolname{} with \gspacer{}~\cite{vediramana2023global}, \la{}~\cite{zhu2018data}  and \icedt{} \cite{garg2016learning} on the benchmarks and discuss the results\footnote{We note that there are several invariant approaches now such as \cite{si2020code2inv}, \cite{sharma2016invariant} etc. however they are all inferior to \gspacer{} and \la{} in terms of evaluation.}.

% \subsubsection*{} 
\noindent \textbf{Comparision with \gspacer{}.}
\gspacer{} is a bounded model checking-based algorithm designed to find loop invariants - it maintains over and under-approximations of procedure summaries which it uses to prove safety of the program or construct a counterexample. \gspacer{} achieved a 89.6\% accuracy on the benchmarks as reported in Table \ref{tab:results} (left), and is the current state-of-the-art tool in proving safety. \\
During our investigation we found a set of standard LIA single-loop CHC systems with which \gspacer{} performed poorly. Specifically, for a benchmark with \( n \) variables and a transition relation \( T(\vec{x}, \vec{x}')\) defined by \(\vec{x}' = A \cdot \vec{x} + \vec{b} \), where \( A \) is an \( n \times n \) matrix and \( \vec{b} \) is an \( n \times 1 \) vector, \gspacer{} struggles when all entries of \( A \) are non-zero. Figure~\ref{fig:flawGspacer} (left) provides an example of a program where \gspacer{} encounters this issue. Table \ref{tab:Compresults} (right) shows experimental results comparing \gspacer{} and \toolname{} on our 11 constructed benchmarks. We conjecture that since \gspacer{} establishes program safety by iteratively refining reachability and summary maps with increasing bounds on the call stack, finding invariants for complex transitions is challenging as it is harder to use summary facts from bounded call-stack code executions to prove safety when the transitions within the loop are complex. 

\begin{figure}
    % \centering
    \begin{minipage}[H]{.4\linewidth}
    \centering
    \begin{verbatim}
    int main(void) {
        int x = 1; int y = 0;
        while ( rand_bool() ) {
            x = a*x + b*y; y = y + 1;
        }
        assert(x >= y); return 0;
    }
    \end{verbatim}
    \end{minipage}
    \hspace{1.0 em}
    \begin{minipage}[ht]{.4\linewidth}
        \centering
        \resizebox{\textwidth}{!}{%
            \begin{tabular}{lccccc}
            \toprule
            \textbf{No. of variables} & \gspacer{} & \toolname{} & \textit{Total} \\
            \midrule
            \textit{$n=2$ Benchmarks}    & 0  & 8  & 8\\
            \textit{$n=3$ Benchmarks}    & 0  & 2  & 2\\
            \textit{$n=4$ Benchmarks}    & 0  & 0 & 1\\
            \bottomrule
            \end{tabular}
        }
    \end{minipage}
    \caption{(\textbf{Left}) The C code above is parameterized by non-negative constants \( a  \) and \( b \). When \( a \neq 0 \land b \neq 0 \), the code demonstrates a limitation of \gspacer{}. \gspacer{} struggles with such transition relations involving non-zero values for both \( a \) and \( b \), timing out on all benchmarks with a 600-second limit, whereas \toolname{} successfully finds the invariants across all 8 constructed benchmarks, with an average runtime of 21.4 seconds.
    (\textbf{Right}) Comparative Results for \toolname{} vs \gspacer{} on Benchmarks with Complex Transition Relations. The numbers show the number of benchmarks converged for each tool in each category.}
    \label{fig:flawGspacer}
    \label{tab:Compresults}
    \vspace{-1.0 em}
\end{figure}
\noindent \textbf{Comparision with \la{}.}
\la{} uses a CEGIS-based algorithm to discover invariants, integrating SVMs and Decision Trees within their invariant search.  \la{} achieves an 86.5\% accuracy on the benchmarks (see Table \ref{tab:results}). It is important to mention that \la{} is built as an LLVM pass in the \texttt{SEAHORN} framework \cite{gurfinkel2015seahorn}, and many of the invariant benchmarks input to \la{} are solved by \textit{Seahorn} instead of \la{} alone -\textit{Seahorn} utilizes techniques such as abstract interpretation and property-directed reachability to prove safety. 
A significant limitation of \la{} is its high sensitivity to counterexamples generated by the SMT solver: \la{} relies on SVM to derive coefficients from the current dataset; changes in the position of counterexamples can drastically alter these learned coefficients. Viewed in this light, if the SMT solver is considered adversarial in generating counterexamples, it is trivially seen that \la{} would never converge to an invariant. Additionally, \la{} produces its final approximate invariants using Decision Trees, which often result in larger (by size) invariants compared to \toolname{}, making it less effective at learning concise invariants.

%\subsubsection*{}
\noindent \textbf{Comparison with \icedt{}.}
\icedt{}, the seminal CEGIS-based ICE algorithm for invariant automation, uses Decision Trees to find approximate invariants.  \icedt{} obtains a 53.7\% accuracy on the benchmarks  (refer Table \ref{tab:results}). 
A limitation of \icedt{} is its inability to autonomously learn non-rectangular features, meaning coefficient vectors that do not align with standard basis vectors are not learned unless explicitly provided by the user. Therefore, this algorithm necessitates certain coefficient vectors to be either user input or hardcoded. Furthermore, \icedt{} guarantees termination through CEGIS by incrementally increasing the maximum bound on the constants of the invariant. Consequently, learning an invariant with constants $\ge M$ requires $\Omega(M)$ time, meaning its convergence time is proportional to $O(|\mathcal{D}_{state}|^{\frac{1}{n}})$, assuming sufficient features for each maximum modulus constant search. In contrast, \toolname{}'s CEGIS termination bound is proportional to $O(\log_2 |B|)$, where $B \subseteq \mathcal{D}_{state}$ represents the loop guard.

% \subsubsection*{Comparison with \ci{}}
% We also include a comparison with \ci{}, as it is the only other tool to use a probabilistic approach for learning invariants. The high-level algorithm for \ci{} is quite similar to \toolname{}: it uses Metropolis-Hastings (a variant of the simulated annealing algorithm) within a CEGIS algorithm to discover invariants. However, without $\epsilon$-nets, dispersed counterexamples, implication pairs, iterated implication pairs, and a smooth cost function for the invariant search via Metropolis-Hastings, \ci{} performs poorly on real benchmarks. Their tool has only been evaluated on 8 toy benchmarks in their paper.
% \stsays{instead of the last sentence, say what is the success rate? if you did not experiment with the same benchmark, just remove this paragraph.}

\noindent \textbf{Comparison Conclusions.}
Overall the comparison shows that, as an alternative approach, \toolname{} achieves competitive performance, in addition to providing probabilistic guarantees. Programs such as those in Figure~\ref{fig:flawGspacer} show that \toolname{} has complementary strengths w.r.t the current state-of-the-art tool \gspacer{}; it would be interesting future work to combine different approaches.

\subsection{Limitations and Discussions.}
We identify the primary limitation of \toolname{} as the non-convergence of SA on certain benchmarks. Additionally for converged benchmarks, SA runs dominate \toolname{}'s computational workload (see Figure \ref{fig:SAvsDS} in Appendix \ref{App: timeDSSA}). This is a trade-off we accept for using probabilistic algorithms in the guess phase—while these algorithms allow us to explore large search spaces without getting trapped in local minima, their probabilistic nature also slows the global search, sometimes to the point of non-convergence. In this section, we examine this phenomenon in detail, specifically analyzing the factors causing SA's non-convergence on some benchmarks and exploring whether adjusting the search space strategy could improve performance.

\noindent \textbf{Selection Bias.}
% \subsubsection*{Selection Bias.} 
The biggest reason for the failure of SA is selection bias; i.e. we sample a neighbor which may lead us to an approximate invariant with a very low probability. In terms of Theorem \ref{thm: OurSAguarantee}, this quantity is represented as $w := \frac{1}{cd(2n+2)}$. Consequently, we obtain poorer performance results for benchmarks with more variables or larger invariant sizes.

% \subsubsection*{Affine Spaces.} 
\noindent \textbf{Affine Spaces.}
Because of our invariant template, we can only represent affine spaces as the conjunction of two LIA predicates (as we allow only inequality predicates) instead of a single equality predicate. This implies that benchmarks that have affine spaces for invariants require larger invariant spaces (represented by variable $r$ in Theorem \ref{thm: OurSAguarantee}), and hence the probability of success 
% in such larger invariant spaces 
is smaller.

% \begin{table}[t]
% \centering
% \resizebox{0.6\textwidth}{!}{%
% \begin{tabular}{lcccccc}
% \toprule
% \textbf{Folder Name} & Simulated Annealing & Gradient Descent & Genetic Programming & Total \\
% \midrule
% \textit{loop\_lit}    & 6  & 1  & 3 & 11\\
% \textit{loop\_new}    & 2  & 1  & 2 & 2\\
% \textit{loop\_zilu}   & 48  & 25 & 33 & 53\\
% \textit{loop\_simple} & 1  & 0  & 1  & 1\\
% \textit{Total}        & 57  & 27 & 39 & 67\\
% \bottomrule
% \end{tabular}
% }
% \caption{Comparison of SA with other search space algorithms.}
% \label{tab:sagdea}
% \vspace{-2.0 em}
% \end{table}

\begin{minipage}{0.52\textwidth}
\begin{table}[H]
    \centering
    \resizebox{\textwidth}{!}{%
        \begin{tabular}{lcccc}
            \toprule
            \textbf{Folder Name} & \begin{tabular}{@{}c@{}}\textbf{Simulated} \\ \textbf{Annealing}\end{tabular} & \begin{tabular}{@{}c@{}}\textbf{Gradient} \\ \textbf{Descent}\end{tabular} & \begin{tabular}{@{}c@{}}\textbf{Genetic} \\ \textbf{Programming}\end{tabular} & \textbf{Total} \\
            \midrule
            \textit{loop\_lit}    & 6  & 1  & 3 & 11\\
            \textit{loop\_new}    & 2  & 1  & 2 & 2\\
            \textit{loop\_zilu}   & 48 & 25 & 33 & 53\\
            \textit{loop\_simple} & 1  & 0  & 1  & 1\\
            \textit{Total}        & 57 & 27 & 39 & 67\\
            \bottomrule
        \end{tabular}
    }
\caption{Comparison of SA with other search space algorithms.}
\label{tab:sagdea}
\end{table}
\end{minipage}%
\hfill
\begin{minipage}{0.46\textwidth}
\noindent \textbf{Design Choice of Search Space Algorithm.}
Given the limitations of SA highlighted above, we ask the natural question: Will substituting it with alternatives like gradient descent or genetic programming~\cite{yu2010introduction} improve evaluation? Table \ref{tab:sagdea} answers this question in the negative - and confirms that SA is indeed the optimal search space algorithm for our problem.
\end{minipage}

\section{Conclusion and Future Works.}
\label{Sec:Conclusions and Future Works.}
Our paper presents a data-driven algorithm to find invariants for single-loop CHC systems with background theory LIA. Our algorithm leveraged the simulated annealing with SMT solvers as well as computational geometry to provide probabilistic guarantees for inferring loop invariants. 
One exciting future work is to find invariants for single-loop CHC systems with Linear Real Arithmetic (LRA) as the background theory, allowing floating point program variables.
The major challenge of this extension is that even a bounded state space would be infinite because there are infinitely many real numbers in any bounded region, unlike integers. A possible way to tackle this issue is to use a precision cutoff for both the state space and invariant space. 
% This approach introduces additional issues, such as finding counterexamples by the SMT solver within the given precision range. 
Additionally, syntactic-based searching might not be the best approach for such a space. 
An extension of this work would be to allow the program to have both integer and real variables. The major challenge in integrating work on finding LIA invariants and LRA invariants would be to decide the template of the new general invariant.

\bibliographystyle{plainnat}
\bibliography{sample-base}

\newpage
\appendix

\section{Supplementary Preliminaries.}

\subsection{$T_{\mathbb{N}}$ axioms and Equisatisfiability of $T_{\mathbb{N}}$ and  $T_{\mathbb{Z}}$ formulas.} \label{App: LIA theory}
$T_{\mathbb{N}}$ is a first order theory with signature $\Sigma = \{ 0 , 1 , + , = \}$ and its axioms are given below:
\begin{align*}
    (\forall x) \; \neg (x + 1 = 0)  &\hspace{1 cm}\text{(Zero)} \\
    (\forall x,y) \; (x + 1 = y + 1 \rightarrow x = y)  &\hspace{1 cm}\text{(Successor)} \\
    (\forall x) \; \neg (x + 0 = x)  &\hspace{1 cm}\text{(Plus Zero)} \\
    (\forall x,y) \; \neg (x + (y + 1) = (x+y) + 1)  &\hspace{1 cm}\text{(Plus Successor)} \\
    F[0] \land ((\forall x) \; F[x] \rightarrow F[x+1]) \rightarrow ((\forall x) \; F[x])  &\hspace{1 cm}\text{(Induction)} 
\end{align*}
The final axiom, Induction is actually a meta-axiom which holds true for all $T_{\mathbb{N}}$ formula's $F$. \\[1\baselineskip]
Next we show the equisatisfiability of $T_{\mathbb{N}}$ and  $T_{\mathbb{Z}}$ formulas by showing that every extra symbol in the signature of $T_{\mathbb{Z}}$ is just syntactic sugar.
\begin{itemize}
    \item As each variable in $T_{\mathbb{Z}}$ can take integer values in it's standard interpretation, but a variable in $T_{\mathbb{N}}$ can ony take non negative integra values, any variable $x$ in a $T_{\mathbb{Z}}$ formula is substituted by the expression $(x^+ - x^-)$ where $x^+, x^-$ are variables in $T_{\mathbb{N}}$ theory (they represent the positive and negative parts of $x$).
    \item The constant $2$ is just syntactic sugar for $(1 + 1)$, and $2\cdot F$ for any wff $F$ is just syntactic sugar for $(F + F)$. The constant predicates $2, 3, 4, .. $ and $2\cdot, 3\cdot, 4\cdot, .. $ are defined similarly.
    \item A formula with the binary predicate $-$ is equisatisfiable to some $T_{\mathbb{N}}$ formula by moving the $-$ term to the other side of the equality (or inequality) i.e. $x - y = z$ is equisatisfiable with $x = y + z$. 
    \item Now the constants $-1, -2, -3, .. $ and $-1\cdot, -2\cdot, -3\cdot .. $ are defined similarly using the previous two ideas.
    \item Now the inequality symbols can be substituted by $T_{\mathbb{N}}$ formulas by adding another variable i.e. $ (x > y)$ is equivalent to $(x = y + z) \land \neg(z = 0)$. Recall that variables in $T_{\mathbb{N}}$ only take non negative values in the standard interpretation. Similarly we can find equisatisfiable $T_{\mathbb{N}}$ formulas for $T_{\mathbb{Z}}$ formulas using $\{ <, \le , > , \ge \}$.
\end{itemize}

\subsection{VC dimension and $\epsilon$-nets.}\label{App:enetvc}
We give en example of the set of ellipses shattering the vertices of a regular pentagon in Figure \ref{fig:ellipsesenet}, which shows the VC dimension of ellipses in the plane is at least $5$.
\begin{SCfigure}
    \centering
    \includegraphics[width=0.5\textwidth]{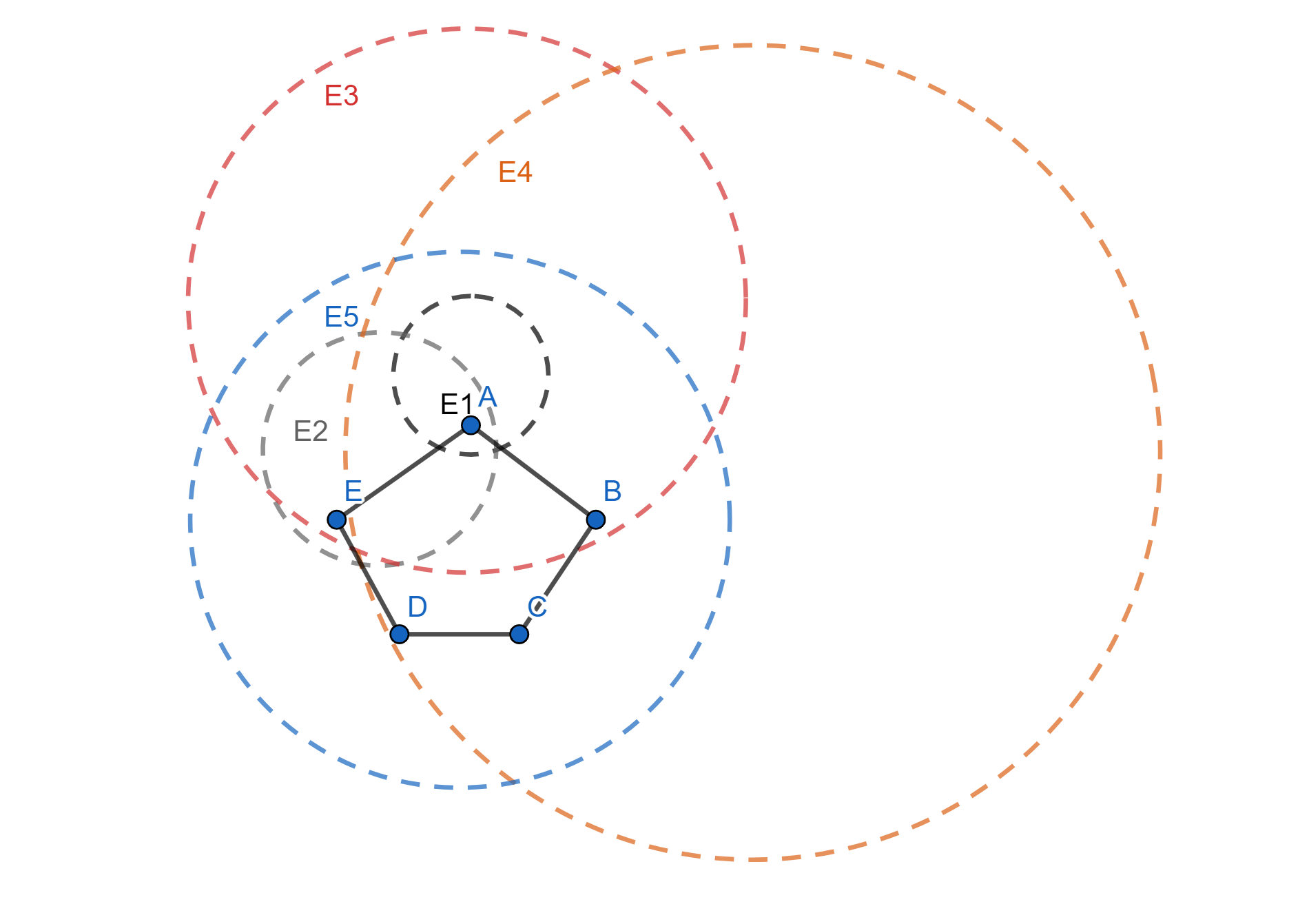}
    \caption{The picture in the right shows the vertices of a regular pentagon ABCDE being shattered by the range space ($\mathbb{R}^2$, $\mathcal{E}$) where $\mathcal{E}$ is the set of all ellipses in the plane. For $1 \le i \le 5$,  ellipse $Ei$ intersects ABCDE at exactly $i$ points; and by symmetry we can see that any subset of ABCDE will indeed be the intersection of some ellipse with ABCDE. Additionally, it can be shown that no 6 element subset of $\mathbb{R}^2$ can be shattered by the ellipses in the plane, which implies that the VC dimension of ellipses in the plane is $5$.}
    \label{fig:ellipsesenet}
\end{SCfigure}

We now give the VC dimension of a few range spaces.

\begin{lemma} \cite{blumer1989learnability} \label{Lemma: VCdimPolytope}
    the VC dimension of the class of $n$-dimensional polytopes with at most $c$ faces in $\mathbb{R}^n$ is less than or equal to $2c(n+1) \log(3c)$.
\end{lemma}

\begin{lemma} \cite{akama2011vc} \label{Lemma: VCdimEllipsoid}
     The VC dimension of the class of $n$-dimensional ellipsoids in $\mathbb{R}^n$ is $\frac{n^2 + 3n}{2}$.
\end{lemma}

\subsection{More on Ellipsoids in $\mathbb{R}^n$.} \label{App:IOLJE}
\paragraph{Lattice points contained in an ellipsoid}
Let \( X \subseteq \mathbb{R}^n \) be a Lebesgue measurable set. We use the notation \( \text{vol}(X) \) to represent the volume of \( X \) and \( \Lambda(X) \) to denote the number of lattice points contained in \( X \). For any \( n \)-dimensional ellipsoid \( \mathcal{E} \), let \( t \mathcal{E} := \{ t\vec{s} : \vec{s} \in \mathcal{E} \} \) denote the dilated \( n \)-dimensional ellipsoid. We have the following relation between \( \text{vol}(t\mathcal{E}) \) and \( \Lambda(t\mathcal{E}) \) from Landau (1915):

\begin{theorem} \cite{gotze2004lattice} \label{Thm: Ratio of Ellipses}
    Let \( n \ge 2 \). Then for any \( n \)-dimensional ellipsoid \( \mathcal{E} \) and \( t > 0 \), there exists a constant \( C > 0 \) such that
    $$  1 - \frac{C}{t^{2 - \frac{2}{n+1}}} \le \frac{\Lambda(t\mathcal{E})}{\text{vol}(t\mathcal{E})} \le 1 + \frac{C}{t^{2 - \frac{2}{n+1}}}. $$
\end{theorem}

Additionally, it is straightforward to observe that \( \text{vol}(t\mathcal{E}) = t^n \text{vol}(\mathcal{E}) \).

\paragraph{Inner and outer Lowner-John ellipsoids}
Any polytope in \( \mathbb{R}^n \) has a unique inner Lowner-John ellipsoid and a unique outer Lowner-John ellipsoid (see Figure \ref{fig:JLellipse} in for an example). Furthermore, the inner and outer Lowner-John ellipsoids share a common center and are dilations of each other w.r.t their common center. The following theorem formalizes the relationship between the inner and outer Lowner-John ellipsoids.

\begin{theorem} \cite{henk2012lowner} \label{Thm: JLEllipsoid}
    Let \( \mathcal{P} \) be a compact polytope in $\mathbb{R}^n$. Then, there exists an ellipsoid \( \mathcal{E} \) centered at the origin and a vector \( \vec{t} \) called the center, such that the inner and outer Lowner-John ellipsoids of \( \mathcal{P} \) are given by \( \vec{t} + \frac{1}{n}\mathcal{E} \) and \( \vec{t} + \mathcal{E} \), respectively. Finally, we have the following set-relations:
    $$\vec{t} + \frac{1}{n}\mathcal{E} \subseteq \mathcal{P} \subseteq \vec{t} + \mathcal{E}.$$
\end{theorem}
\begin{figure}[!tbh]
    \centering
    \includegraphics[scale=0.30]{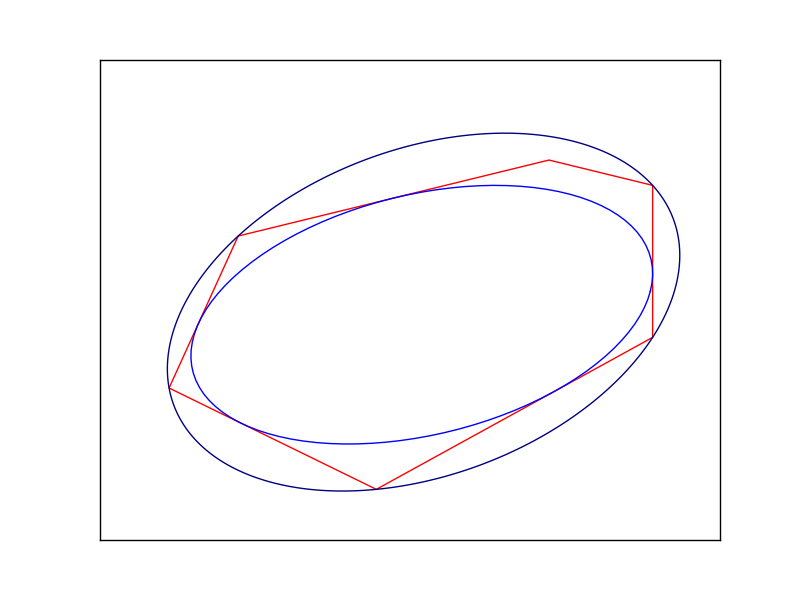}
    \caption{For the pentagon (red) in the plane, the figure shows the inner Lowner-John ellipse (blue) and the outer Lowner-John ellipse (dark blue). (Image from \cite{mosek}.)}
    \label{fig:JLellipse}
\end{figure}

\subsection{Simulated Annealing: Convergence Theorem and Initial Temperature Computation.} \label{App:SA}
We first give the main convergence guarantee of simulated annealing from \cite{mitra1986convergence}.

\begin{theorem} \cite{mitra1986convergence} \label{thm: SAguarantee}
Let \((X_t)\) be a sequence of random variables denoting a simulated annealing search on the search space \(\mathcal{X}\) with neighborhood function \(\mathcal{N}\) and cost function \(c\). Define \(\mathcal{X}_m\) as the set of all local maxima in \(\mathcal{X}\), \(\mathcal{X}^*\) as the set of all global minima in \(\mathcal{X}\), and \(d(s,s')\) as the length of the shortest path between \(s\) and \(s'\). We then define the following parameters:

\begin{itemize}
    \item \(r := \min_{s \in \mathcal{X} \setminus \mathcal{X}_m} \max_{s' \in \mathcal{X}} d(s,s')\) denotes the radius of the state space.
    \item \(w := \min_{s \in \mathcal{X}} \frac{1}{|\mathcal{N}(s)|}\) denotes the minimum probability of choosing a neighbor.
    \item \(L := \max_{s \in \mathcal{X}} \max_{s' \in \mathcal{N}(s)} |c(s) - c(s')|\) denotes the maximum cost change in a single transition.
    \item \(\delta := \min_{s \not\in \mathcal{X}^*} \left( c(s) - c(s^*) \right)\) denotes the minimum cost change in a transition to global minima.
\end{itemize}

Then, for any starting candidate \(s_0 \in \mathcal{X}\) and \(T_0 \ge rL\), we have for all natural numbers \(k\),

\[
\Pr(X_{kr} \in \mathcal{X}^*) \ge 1 - \frac{A}{k^{\min (a,b)}}
\]

where 

\[
a := \frac{w^r}{r^{\frac{rL}{T_0}}}, \quad b := \frac{\delta}{T_0}
\]

and \(A > 0\) is some constant.
\end{theorem}

We now describe the initial temperature computation from \cite{ben2004computing}.
In brief, this study attempts to compute a value for \(T_0\) such that the conditional expectation of the acceptance ratio for this initial temperature for positive transitions, sampled in conformance with the stationary distribution, is near the targeted acceptance ratio \(a_0\). For simulated annealing to work well, we want to set $a_0 \in [0.20,0.50]$ \cite{potter20150}. The study shows that using a random walk to obtain positive transitions, instead of sampling in conformance with the stationary distribution, works quite well too.
\begin{algorithm} 
\caption{\texttt{INITIAL\_TEMPERATURE \cite{ben2004computing}}} \label{alg: InitTemp}
\begin{algorithmic}[1]
    \State \textbf{Input: } $(\mathcal{X}, \mathcal{N}, c , s_0)$
    \State \textbf{Hyperparameters: } $a_0$, $\epsilon$, $t_{rw}$
    \State $ \delta \longleftarrow \texttt{BOUNDED\_RANDOM\_WALK}(\mathcal{X}, \mathcal{N}, c , s_0, t_{rw})$
    \State $ \delta^+ \longleftarrow \texttt{GET\_POSITIVE\_TRANSITIONS}(\delta)$
    \State $T \longleftarrow \frac{- \sum_{(c,c') \in \delta^+} (c' - c) }{\ln(a_0) |\delta^+|}  $
    \State $a \longleftarrow \displaystyle\sum_{ (c,c') \in \delta^+} \frac{e^{-\frac{c'}{T}}}{ e^{-\frac{c}{T}}}$
    \While{ $|a - a_0 | > \epsilon$}
    \State $T \longleftarrow  T \sqrt{ \frac{\ln(a)}{\ln(a_0)} }$
    \State $a \longleftarrow \displaystyle\sum_{ (c,c') \in \delta^+} \frac{e^{-\frac{c'}{T}}}{ e^{-\frac{c}{T}}}$
    \EndWhile
    \State \Return $T$
\end{algorithmic}
\end{algorithm}

The algorithm to compute \(T_0\) is given in Algorithm \ref{alg: InitTemp}. The \texttt{INITIAL\_TEMPERATURE} procedure takes as input the search space \(\mathcal{X}\), the neighborhood relation \(\mathcal{N}\), the cost function \(c\), and an initial candidate \(s_0\). It first computes a set of transitions \(\delta\) by performing a bounded random walk on \(\mathcal{X}\) for \(t_{rw}\) steps (refer Appendix \ref{App: suppAlg} for algorithms of these subprocedures) \footnote{In their paper, \cite{ben2004computing} also give a method for computing $t_{rw}$, but here we just treat it as a hyperparameter.}. It then filters out the positive transitions (i.e., transitions that lead to a higher cost) from \(\delta\) into \(\delta^+\). Next, it computes the initial temperature \(T\) using the formula:
\[T \longleftarrow \frac{- \sum_{(c,c') \in \delta^+} (c' - c)}{\ln(a_0) |\delta^+|}\] 
where \(|\delta^+|\) is the number of positive transitions. It then calculates the acceptance ratio \(a\) for the current temperature \(T\) using the formula:
\[
a \longleftarrow \sum_{(c,c') \in \delta^+} \frac{e^{-\frac{c'}{T}}}{e^{-\frac{c}{T}}}
\]
The algorithm iterates, adjusting \(T\) to minimize the difference between \(a\) and the targeted acceptance ratio \(a_0\). Specifically, it updates \(T\) using the scheme:
\[
T \longleftarrow T \sqrt{\frac{\ln(a)}{\ln(a_0)}}
\]
This adjustment continues until the acceptance ratio \(a\) is within an \(\epsilon\)-tolerance of \(a_0\). Finally, the algorithm returns the computed initial temperature \(T\). The study in \cite{ben2004computing} provides termination guarantees for this algorithm. 

\section{Encoding Real-World Loop Programs AS standard LIA SINGLE LOOP CHC SYSTEMS.} \label{App: Representing Practical Programs in the Standard CHC Clause system}
We demonstrate how a large class of practical loop programs can be represented using the standard LIA single loop CHC system via examples. Throughout this section, we use $\mathcal{D}_{\text{int}}$ to denote $\mathbb{Z} \cap \footnotesize{[- \text{INT}\_\text{MAX} - 1, \text{INT}\_\text{MAX}]}$ and $I(\vec{s})$ to denote the loop invariant over program states $\vec{s}$. We start with the most basic loop program.
\subsection{Basic Loop Programs.}
An encoding of a basic `while' loop C program is shown in Fig \ref{App: Basic Loop Program}.

\begin{figure}[H]
    \centering
\begin{minipage}{.3\textwidth}
\begin{verbatim}
int main(void) {
    int x = 1;
    int y = 3;
    while (x + y <= 14) {
        x = 2x + y;
        y = x + 4;
    }
    assert(x + y == 42);
    return 0;
}
\end{verbatim}
\end{minipage}
\begin{minipage}{.5\textwidth}
\begin{alignat*}{1}
    (\forall x,y \in \mathcal{D}_{\text{int}}) &\: (x = 1) \land (y = 3) \rightarrow I(x,y) \\
    (\forall x,y, x', y'\in \mathcal{D}_{\text{int}}) &\: I ( x,y ) \land (x + y \le 14) \land \\ &(x' = 2x + y)  \land (y' = 2x + y + 4) \rightarrow I(x', y') \\
    (\forall x,y \in \mathcal{D}_{\text{int}}) &\:I(x,y)  \rightarrow (x + y = 42) \lor (x + y \le 14)
\end{alignat*}
\end{minipage}
    \caption{Encoding a basic `while' loop program as a standard CHC system.}
    \label{App: Basic Loop Program}
\end{figure}

There are two key points to note here. Firstly, we note that our piecewise linear relation over $B$ with integer coefficients $T$ is encoded in the inductive clause as the conjunction of equalities where a single primed variable is on the left of the equality, and only unprimed variables are on the right\footnote{This is possible for piecewise linear relations over a set $B$ with integer coefficients.}---hence the relation $(y' = x' + 4)$ is encoded as $(y' = 2x + y + 4)$. Encoding $T$ in such a form requires some static analysis within the loop body (a loop-free fragment of code and hence decidable). Secondly, the query clause of the CHC system is written in $(\forall \vec{s})\: I(\vec{s}) \rightarrow Q'(\vec{s})$ form by observing that $(\forall \vec{s})\: (I(\vec{s}) \land \neg B(\vec{s}) \rightarrow Q(\vec{s})) \equiv (\forall \vec{s})\: (I(\vec{s}) \rightarrow Q(\vec{s}) \lor B(\vec{s}))$ and setting $Q'(\vec{s}) = Q(\vec{s}) \lor B(\vec{s})$. 

We note that the choice of the 'while' loop construct is arbitrary; we can similarly encode a 'for' loop or 'do-while' loop construct as depicted in Fig \ref{App: Basic Loop program 2}.

\begin{figure}[H]
\centering
\begin{minipage}{.5\textwidth}
\begin{verbatim}
int main(void) {
    int S = 0;
    for (int i = 0; i <= 100; i++) {
        S = S + i;
    }
    assert(S >= i);
    return 0;
}
\end{verbatim}
\end{minipage}%
\begin{minipage}{.5\textwidth}
\begin{alignat*}{1}
    (\forall i,S \in \mathcal{D}_{\text{int}}) &\: (i = 0) \land (S = 0) \rightarrow I(i,S) \\
    (\forall i,S, i', S'\in \mathcal{D}_{\text{int}}) &\: I ( i,S ) \land (i \le 100) \\ &\land (S' = S + i) \land (i' = i + 1)  \rightarrow I(i', S') \\
    (\forall i,S \in \mathcal{D}_{\text{int}}) &\:I(i,S)  \rightarrow (i \le 100) \lor (S \ge i)
\end{alignat*}
\end{minipage}
\begin{minipage}{.5\textwidth}
\begin{verbatim}   
int main(void) {
    int S = 0;
    int i = 0;
    do {
        S = S + i;
        i = i + 1;
    } while (i <= 100);
    assert(S >= i);
    return 0;
}
\end{verbatim}
\end{minipage}%
\begin{minipage}{.5\textwidth}
\begin{alignat*}{1}
    (\forall i,S \in \mathcal{D}_{\text{int}}) &\: (i = 1) \land (S = 0) \rightarrow I(i,S) \\
    (\forall i,S, i', S'\in \mathcal{D}_{\text{int}}) &\: I ( i,S ) \land (i \le 100) \\ &\land (S' = S + i) \land (i' = i + 1)  \rightarrow I(i', S') \\
    (\forall i,S \in \mathcal{D}_{\text{int}}) &\:I(i,S)  \rightarrow (i \le 100) \lor (S \ge i)
\end{alignat*}
\end{minipage}
\caption{Encoding a `for' loop program and a `do-while' loop program as a standard CHC system. Note that for a `do-while' loop program, the invariant is a proposition which holds when the loop guard is checked, and hence the body of the fact clause for such a CHC system is the set of states reachable from the precondition after exactly one forward loop unrolling; instead of simply the precondition states. }
\label{App: Basic Loop program 2}
\end{figure}

\subsection{Random Initialization of Variables.}

An random initialization statement as used in the below C program:
\begin{verbatim}
int main(void) {
    int x = rand_int();
    int y = rand_int();
    while (...) {
        ...
    }
    assert(...);
    return 0;
}
\end{verbatim}
could be represented by the precondition $\bot$. If a random initialization is followed by an assertion, then we get more sophisticated preconditions. For example, the initialization states in the below code can be represented by the precondition $(x + y \le 10)$.
\begin{verbatim}
int main(void) {
    int x = rand_int();
    int y = rand_int();
    assert(x + y <= 10)
    while (...) {
        ...
    }
    assert(...);
    return 0;
}
\end{verbatim}

\subsection{Conditionals Inside Loop Body.}
The presence of conditionals inside the loop body is why we allow T to be a piecewise linear relation over $B$ with integer coefficients instead of simply a linear function with integer coefficients.

\subsubsection{Conditionals with deterministic guards.}
An encoding of a loop C program with a single deterministic guard conditional is shown in Fig \ref{App: Conditional Deterministic}. A key point to note is that to encode such a loop we need to find an equivalent loop\footnote{Two loops $L_1$ and $L_2$ are called equivalent if $T(L_1) = T(L_2)$ where $T(L)$ denotes the transition relation of loop $L$.} where the loop body is a conditional statement\footnote{A conditional statement is a statement of the form: `if ( S )  \{ S \} \:\: (\: else if  ( S ) \{ S \} \: )* \:\:  else \{ S \}' where $S$ is a placeholder for any program statement.} and furthermore there are no conditionals within the body of any branch of the conditional statement---this can be done by static analysis on the loop body of the original loop. Fig \ref{App: Conditional Deterministic} shows how to encode the loop when the loop body contains a conditional substatement\footnote{If a statement $S$ is a composition given by $S1; S2$, then $S1$ and $S2$ are called substatements.} and Fig \ref{App: Conditional Deterministic 2} shows how to encode a loop when the loop body has nested conditional statements. Finally, we note that if the loop body contains deterministic guard conditionals only and the substatements in the loop body are LIA formulas, then $T$ is a piecewise linear function over the loop guard $B$ with integer coefficients\footnote{$T$ is a piecwise linear function over $B$ with integer coefficients if there exists a finite partition $\{ B_i \}$ of $B$ s.t. $T$ restricted on $B_i$ is a linear function on $B_i$ with integer coefficients forall $i$.}.

{ \footnotesize
\begin{figure}[H]
    \centering
\begin{minipage}{.3\textwidth}
\begin{verbatim}
int main(void) {
    int i = 1;
    int j = 0;
    while (i + j <= 100) {
        i = i + 3*j;
        if (i < 50) 
            j = j + 1;
        else 
            j = j + 2;
    }
    assert(i >= 2*j);
    return 0;
}
\end{verbatim}
\end{minipage}
\begin{minipage}{.5\textwidth}
\begin{alignat*}{1}
    (\forall i,j \in \mathcal{D}_{\text{int}}) &\: (i = 1) \land (j = 0) \rightarrow I(i,j) \\
    (\forall i,j, i', j'\in \mathcal{D}_{\text{int}}) &\: I ( i,j ) \land  \\ & ((i + j \le 100) \land (i + 3j < 50) \land (i' = i + 3j)  \land (j' = j + 1)) \land
    \\ & ((i + j \le 100) \land (i + 3j \ge 50) \land (i' = i + 3j)  \land (j' = j + 2)) \\ &\rightarrow I(i', j') \\
    (\forall i,j \in \mathcal{D}_{\text{int}}) &\:I(i,j)  \rightarrow (i + j > 100) \lor (i \ge 2j)
\end{alignat*}
\end{minipage}
    \caption{Encoding a loop with a deterministic-guard conditional substatement.}
    \label{App: Conditional Deterministic}
\end{figure}

\begin{figure}[H]
    \centering
\begin{minipage}{.3\textwidth}
\begin{verbatim}
int main(void) {
    int i = 1;
    int j = 0;
    while (i + j <= 100) {
        if (i < 50) {
            j = j + 1;
            if (i >= j)
                i = i + 3;
            else
                i = i + 4;
        }
        else 
            i = i + 5;
        j = j + 2;
   }
   assert(i >= j);
   return 0;
}
\end{verbatim}
\end{minipage}
\begin{minipage}{.5\textwidth}
\begin{alignat*}{1}
    (\forall i,j \in \mathcal{D}_{\text{int}}) &\: (i = 1) \land (j = 0) \rightarrow I(i,j) \\
    (\forall i,j, i', j'\in \mathcal{D}_{\text{int}}) &\: I ( i,j ) \land  \\ & ((i + j \le 100) \land (i < 50) \land ( i \ge j) \land (i' = i + 3)  \land (j' = j + 1)) \land \\ 
    & ((i + j \le 100) \land (i < 50) \land ( i < j) \land (i' = i + 4)  \land (j' = j + 1)) \land \\ 
    & ((i + j \le 100) \land (i \ge 50) \land (i' = i + 5)  \land (j' = j + 2)) \\ &\rightarrow I(i', j') \\
    (\forall i,j \in \mathcal{D}_{\text{int}}) &\:I(i,j)  \rightarrow (i + j > 100) \lor (i \ge j)
\end{alignat*}
\end{minipage}
    \caption{Encoding a loop with nested deterministic-guard conditional statements.}
    \label{App: Conditional Deterministic 2}
\end{figure}
}
\subsubsection{Conditionals with nondeterministic guards.}
An encoding of a loop C program with a single nondeterministic guard conditional is shown in Fig \ref{App: Conditional NonDeterministic}. Nondeterministic guard conditionals introduce logical OR in the inductive clause of the CHC system.Note that if the loop body contains nondeterministic guard conditionals only and the substatements in the loop body are LIA formulas, then $T$ is a linear relation over the loop guard $B$ with integer coefficients. Finally, we note that a real-world loop program may have both deterministic guard and nondeterministic guard conditionals, in which case our transition relation $T$ is a piece-wise linear relation over the loop guard $B$ with integer coefficients (assuming substatements in the loop body are LIA formulas).
\begin{figure}[H]
    \centering
\begin{minipage}{.4\textwidth}
\begin{verbatim}
int main(void) {
    int i = 1;
    int j = 0;
    while (i + j <= 100) {
        i = i + j;
        if (rand_bool())
            j = j + 1;
        else 
            j = j + 2;
    }
    assert(i >= 2*j);
    return 0;
}
\end{verbatim}
\end{minipage}
\begin{minipage}{.5\textwidth}
\begin{alignat*}{1}
    (\forall i,j \in \mathcal{D}_{\text{int}}) &\: (i = 1) \land (j = 0) \rightarrow I(i,j)  \\
    (\forall i,j, i', j'\in \mathcal{D}_{\text{int}}) &\: I ( i,j ) \land (i + j \le 100) \\ & ( ((i' = i + j)  \land (j' = j + 1)) \: \lor  \\
    & ((i' = i + j)  \land (j' = j + 2))) \rightarrow I(i', j') \\
    (\forall i,j \in \mathcal{D}_{\text{int}}) &\:I(i,j)  \rightarrow (i + j > 100) \lor (i \ge 2j)
\end{alignat*}
\end{minipage}
    \caption{Encoding a loop with nondeterministic-guard conditional.}
    \label{App: Conditional NonDeterministic}
\end{figure}

\subsubsection{Conditionals with partial nondeterministic guards.} A partial nondeterministic conditional guard is a statement of the form $B \land \text{rand\_bool()}$ or $B \lor \text{rand\_bool()}$ where $B$ is a deterministic conditional guard. We show an encoding of loops with such partial nondeterministic guard conditionals in Fig \ref{App: Conditional partial Nondeterministic}.

\begin{figure}[H]
\centering
\begin{minipage}{.5\textwidth}
\begin{verbatim}
int main(void) {
    int i = 0;
    while (i <= 100) {
        if ((i <= 25) && rand_bool())
            i = i + 1;
        else
            i = i + 10;
   }
   assert(i >= 102);
   return 0;
}
\end{verbatim}
\end{minipage}%
\begin{minipage}{.5\textwidth}
\begin{alignat*}{1}
    (\forall i \in \mathcal{D}_{\text{int}}) &\: (i = 0) \rightarrow I(i) \\
    (\forall i, i' \in \mathcal{D}_{\text{int}}) &\: I ( i) \land (i \le 100) \land  \\ & ((i \le 25) \land ((i' = i + 1) \lor (i' = i + 10))) \land \\ &
    ((i > 25) \land (i' = i + 10)) \rightarrow I(i') \\
    (\forall i \in \mathcal{D}_{\text{int}}) &\:I(i)  \rightarrow (i \le 100) \lor (i \ge 102)
\end{alignat*}
\end{minipage}
\begin{minipage}{.5\textwidth}
\begin{verbatim}
int main(void) {
    int i = 0;
    while (i <= 100) {
        if ( (i <= 25) || rand_bool() )
            i = i + 1;
        else
            i = i + 10;
   }
   assert(i <= 110);
   return 0;
}
\end{verbatim}
\end{minipage}%
\begin{minipage}{.5\textwidth}
\begin{alignat*}{1}
    (\forall i \in \mathcal{D}_{\text{int}}) &\: (i = 0) \rightarrow I(i) \\
    (\forall i, i' \in \mathcal{D}_{\text{int}}) &\: I ( i) \land (i \le 100) \land  \\ & ((i \le 25) \land (i' = i + 1)) \land \\ &
    ((i > 25) \land ( (i' = i + 1) \lor (i' = i + 10))) \\ &\rightarrow I(i') \\
    (\forall i \in \mathcal{D}_{\text{int}}) &\:I(i)  \rightarrow (i \le 100) \lor (i \le 110)
\end{alignat*}
\end{minipage}
\caption{Encoding loops with the two different partial nondeterministic conditional guards. }
\label{App: Conditional partial Nondeterministic}
\end{figure}

\subsection{Nondeterministic Loop Guards.}
An encoding of a loop with a nondeterministic guard is shown in Fig \ref{App: Nondeterministic Loop}.

\begin{figure}[H]
\begin{minipage}{.5\textwidth}
\begin{verbatim}
int main(void) {
    int x = 1;
    int y = 2;
    while (rand_bool()) {
        x = x + y;
        y = y + 1;
    }
    assert(x >= y || x == 1);
    return 0;
}
\end{verbatim}
\end{minipage}%
\begin{minipage}{.5\textwidth}
\begin{alignat*}{1}
    (\forall x,y \in \mathcal{D}_{\text{int}}) &\: (x = 1) \land (y = 2) \rightarrow I(x, y) \\
    (\forall x, y, x', y' \in \mathcal{D}_{\text{int}}) &\: I ( x,y) \land (x' = x + y) \\ &\land (y' = y + 1) \rightarrow I(x', y') \\
    (\forall x, y \in \mathcal{D}_{\text{int}}) &\:I(x,y)  \rightarrow (x \ge y) \lor (x = 1)
\end{alignat*}
\end{minipage}
\caption{Encoding a loop with a nondeterministic guard. }
\label{App: Nondeterministic Loop}
\end{figure}

\subsubsection{Loops with partial nondeterministic guards}
Encoding of loops with either type of a partial nondeterministic guards are shown in Fig \ref{App: Partial Nondeterministic Loop}.

\begin{figure}[H]
\begin{minipage}{.5\textwidth}
\begin{verbatim}
int main(void) {
    int x = 1;
    int y = 2;
    while (rand_bool() && (x <= 100)) {
        x = x + y;
        y = y + 1;
    }
    assert(x >= y || x == 1);
    return 0;
}
\end{verbatim}
\end{minipage}%
\begin{minipage}{.5\textwidth}
\begin{alignat*}{1}
    (\forall x,y \in \mathcal{D}_{\text{int}}) &\: (x = 1) \land (y = 2) \rightarrow I(x, y) \\
    (\forall x, y, x', y' \in \mathcal{D}_{\text{int}}) &\: I ( x,y) \land (x \le 100) \land (x' = x + y) \\ &\land (y' = y + 1) \rightarrow I(x', y') \\
    (\forall x, y \in \mathcal{D}_{\text{int}}) &\:I(x,y)  \rightarrow (x \ge y) \lor (x = 1)
\end{alignat*}
\end{minipage}
\begin{minipage}{.5\textwidth}
\begin{verbatim}
int main(void) {
    int x = 1;
    int y = 2;
    while (rand_bool() || (x <= 100)) {
        x = x + y;
        y = y + 1;
    }
    assert(x >= y || x == 1);
    return 0;
}
\end{verbatim}
\end{minipage}%
\begin{minipage}{.5\textwidth}
\begin{alignat*}{1}
    (\forall x,y \in \mathcal{D}_{\text{int}}) &\: (x = 1) \land (y = 2) \rightarrow I(x, y) \\
    (\forall x, y, x', y' \in \mathcal{D}_{\text{int}}) &\: I ( x,y) \land (x' = x + y) \\ &\land (y' = y + 1) \rightarrow I(x', y') \\
    (\forall x, y \in \mathcal{D}_{\text{int}}) &\:I(x,y)  \rightarrow (x \ge y) \lor (x = 1) \\ &\lor (x \le 100)
\end{alignat*}
\end{minipage}
\caption{Encoding loops with the two different partial nondeterministic guards. }
\label{App: Partial Nondeterministic Loop}
\end{figure}

\subsection{Break and Continue statements.}
The only interesting case in the use of either of these constructs is when they are used in the body of some branch of a conditional statement.\footnote{If the `break' or `continue' are used outside a conditional statement in the loop body, then the portion of the loop body following these statements is unreachable code.} The encoding of loops with either the `break' or `continue' statement in the body of some conditional branch usually requires some static analysis within the loop body which is demonstrated in Fig \ref{App: Break-Continue Loop}.

{ \footnotesize
\begin{figure}[H]
\begin{minipage}{.5\textwidth}
\begin{verbatim}
int main(void) {
    int x = 1;
    int y = 2;
    while (x <= 100) {
        x = x + y;
        if (x >= 50)
            break;
        y = y + 1;
    }
    assert(x >= y);
    return 0;
}
\end{verbatim}
\end{minipage}%
\begin{minipage}{.5\textwidth}
\begin{alignat*}{1}
    (\forall x,y \in \mathcal{D}_{\text{int}}) &\: (x = 1) \land (y = 2) \rightarrow I(x, y) \\
    (\forall x, y, x', y' \in \mathcal{D}_{\text{int}}) &\: I ( x,y) \land (x \le 100) \land \\ & ((x + y < 50) \land (x' = x + y) \land (y' = y + 1)) \\ &\rightarrow I(x', y') \\
    (\forall x, y \in \mathcal{D}_{\text{int}}) &\:I(x,y)  \rightarrow (x \ge y)  \lor \\ &( (x \le 100) \land ( (x > 100) \lor (x + y < 50)) )
\end{alignat*}
\end{minipage}
\begin{minipage}{.5\textwidth}
\begin{verbatim}
int main(void) {
    int x = 1;
    int y = 2;
    while (x <= 100) {
        x = x + y;
        if (x >= 50)
            continue;
        y = y + 1;
    }
    assert(x >= y);
    return 0;
}
\end{verbatim}
\end{minipage}%
\begin{minipage}{.5\textwidth}
\begin{alignat*}{1}
    (\forall x,y \in \mathcal{D}_{\text{int}}) &\: (x = 1) \land (y = 2) \rightarrow I(x, y) \\
    (\forall x, y, x', y' \in \mathcal{D}_{\text{int}}) &\: I ( x,y) \land (x \le 100) \land \\ & ((x + y < 50) \land (x' = x + y) \land (y' = y + 1)) \\ &((x + y \ge 50) \land (x' = x + y) \land (y' = y)) \\ &\rightarrow I(x', y') \\
    (\forall x, y \in \mathcal{D}_{\text{int}}) &\:I(x,y)  \rightarrow (x \ge y) \lor (x \le 100)
\end{alignat*}
\end{minipage}
\caption{Encoding loops with the `break' or `continue' statement. }
\label{App: Break-Continue Loop}
\end{figure}
}
\subsection{Goto Statements.}
When a `goto' statement causes the control flow to jump backwards in the control flow graph, we get a loop.\footnote{Not all `goto' statments model a loop: `goto' statements where the control flow jumps forward models a branch of a conditional statement.} We note that a goto statement can be encoded by the standard CHC clause only when it's body does not contain loop statements or other goto statements, and is not contained in a loop statement or another goto statement. We demonstrate encoding a program with a `goto' statement in Fig  \ref{App: Goto}.

\begin{figure}[H]
\begin{minipage}{.5\textwidth}
\begin{verbatim}
int main(void) {
    int i = 0;
    int S = 0;
    LOOP: 
    S = S + i;
    i = i + 1;
    if (i < 100)
        goto LOOP;
    assert(S >= i);
    return 0;
}
\end{verbatim}
\end{minipage}%
\begin{minipage}{.5\textwidth}
\begin{alignat*}{1}
    (\forall i,S \in \mathcal{D}_{\text{int}}) &\: (i = 1) \land (S = 0) \rightarrow I(i, S) \\
    (\forall i, S, i', S' \in \mathcal{D}_{\text{int}}) &\: I ( i,S) \land (i < 100) \land (S' = S + i) \\ &\land (i' = i + 1) \rightarrow I(i', S') \\
    (\forall i, S \in \mathcal{D}_{\text{int}}) &\:I(i,S)  \rightarrow (S \ge i) \lor (i < 100)
\end{alignat*}
\end{minipage}
\caption{Encoding a `goto' statement modelling a loop. Note that the invariant $I(i,S)$ is a proposition which holds before the execution of line 7 in the above code. }
\label{App: Goto}
\end{figure}

\subsection{Motivating Piecewise-Linear Integer Relations.} \label{App: MotivateT}
We finally motivate the need for defining our transition relation as piece-wise linear integer relations.  If our loop body is just a sequence of assignments whose right-hand side expressions are formulas from LIA theory, the transition relation $T$ can be written as a linear
transformation from $\mathbb{R}^{n+1}$ to $\mathbb{R}^{n+1}$, where $n$ is the number of program variables. For example, if the loop body is ``$x=2x+y; y=y+4$'', then the transition $T$ is $T(x,y,xp,yp) = (xp = 2x + y) \land (yp = y + 4)$. %% as in Fig \ref{Fig:LB1}
However, the presence of deterministic-guard conditionals in the
loop body requires that the transition relation be more general. In
particular, the transition relation can be a `piece-wise linear
function': given a partition of the loop guard $B$ into disjoint sets, it is a linear transformation on each set from
this partition. For example for the loop body in Fig \ref{Fig:LB2}, our
transition relation $T$ can be written as $T(x,y,xp,yp) \equiv ((x > 5)
\Rightarrow  (xp,yp) = (2x + y, y + 4)) \land ((x \le 5) \Rightarrow (xp,yp) = (x + 1, y - 1))$. If the loop body has a nondeterministic-guard
conditional, then our transition relation is a sequence of
linear transformations over $B$. For example in Fig \ref{Fig:LB3}, the
transition relation $T$ can be written as $T(x,y,xp,yp) \equiv ( 
\texttt{True} \Rightarrow (((xp,yp) = (2x + y, y + 4)) \lor ((xp,yp) = (x + 1, y - 1)))$. 

\begin{figure}[t]
\begin{minipage}{.47\textwidth}
\begin{verbatim}
    while (x + y <= 14) {
        if (x > 5) { x = 2x + y; y = y + 4; }
        else { x = x + 1; y = y - 1; }
    }
\end{verbatim}
    \caption{Loop with a deterministic-guard conditional.}
    \label{Fig:LB2}
\end{minipage}%
\hspace{1.0 em}
\begin{minipage}{.49\textwidth}
\begin{verbatim}
    while (x + y <= 14) {
        if (rand_bool()) { x = 2x + y; y = y + 4; }
        else { x = x + 1; y = y - 1; }
    }
\end{verbatim}
    \caption{Loop with a nondeterministic-guard conditional.}
    \label{Fig:LB3}
\end{minipage}%
\end{figure}

\section{Uniform Sampling from polytope lattice points.} \label{App: Uniform Sampling poly lIA}

We give an algorithm for constructing uniform samples from polytope lattice points. We first describe rejection sampling and constructing uniform samples from the solution space of linear diophantine equations.

 \subsection{Rejection Sampling} \label{App: Rejection Sampling}
Let $U$ be a finite set. Our goal is to produce a uniform sample from $U$---however suppose that sampling from $U$ uniformly isn't directly possible. Also suppose there exists a finite set $V$ s.t. $U \subseteq V$ and we can easily uniformly sample from $V$, then \texttt{REJECTION\_SAMPLING} (Algorithm \ref{Alg: Rej Sampling}) produces a uniform sample from $U$. \texttt{REJECTION\_SAMPLING} takes as input the set $U$ and a procedure $\texttt{V\_UNIFORM\_SAMPLE}$ to construct a uniform sample from some $V \supseteq U$.

\begin{algorithm}[H]
\caption{\texttt{REJECTION\_SAMPLING}}
\begin{algorithmic}[1]
\State \textbf{Input: } $U$, $\texttt{V\_UNIFORM\_SAMPLE}$
\State $x \leftarrow $ $\texttt{V\_UNIFORM\_SAMPLE}$
\While  {$x \not\in U$}
\State $x \leftarrow $ $\texttt{V\_UNIFORM\_SAMPLE}$
\EndWhile 
\State Return $x$
\end{algorithmic}
\label{Alg: Rej Sampling}
\end{algorithm}

\begin{lemma}
    Algorithm \ref{Alg: Rej Sampling} produces a uniform sample from $U$.
\end{lemma}
\begin{proof}
    Let $x$ be the sample produced by Algorithm \ref{Alg: Rej Sampling} and $x_0$ be any element of $U$. Then 
 $$   \probP (x = x_0) =  \sum_{t \ge 0} \Big( \frac{|V| - |U|}{|V|} \Big)^t \cdot \frac{1}{|V|} = \frac{1}{|U|} $$
\end{proof}

For example, if we wish to uniformly sample from the lattice points contained in a polytope $\mathcal{P} \subseteq \mathbb{R}^n$, then we can take $V$ to be any axis aligned orthotope containing $\mathcal{P}$.

Another key point to note is that Algorithm \ref{Alg: Rej Sampling} does not produce a uniform sample from $U$ in constant time. If $\tau$ is the time for Algorithm \ref{Alg: Rej Sampling} to produce a uniform sample from $U$, then
$$ \mathop{{}\mathbb{E}}(\tau) = \sum_{t \ge 0} \probP (\tau > t) = \sum_{t \ge 0} \Big( \frac{|V| - |U|}{|V|} \Big)^t = \frac{|V|}{|U|}
$$

This means that if $|V| \gg |U|$, then uniform sampling from $U$ via rejection sampling on $V$ isn't feasible. In practice, such a case usually occurs when we wish to uniformly sample from the lattice points contained in an affine space $\mathcal{A}$ of dimension $< n$ in $\mathbb{R}^n$ and we take $V$ to be any axis aligned orthotope containing $\mathcal{A}$. But for this case, we can make do of a specialized sampling procedure - sampling from the solution space of linear Diophantine Equations which we mention next.

\subsection{Uniform Sampling from Low Rank Polytopes: Uniform Sampling from the Solution Space of Linear Diophantine Equations.}\label{App: Sampling from Affine Spaces.}
Let $AX = B$ denote an affine space where $A \in \mathbb{Z}^{m \times n}$ and $B, X \in \mathbb{Z}^{n \times 1}$ where $X$ is the column vector of unknowns. This is a linear Diophantine equation, and the solution space of this equation is described in ~\cite{gilbert1990linear}. 
Uniformly sampling from this solution space is easy to do (described in ~\cite{gilbert1990linear}), and we present a concise version in Algorithm \ref{alg: ldus}.

\begin{algorithm}[H]
\caption{\texttt{LINEAR\_DIOPHANTINE\_UNIFORM\_SAMPLER}} \label{alg: ldus}
\begin{algorithmic}[1]
\State \textbf{Input: } $A$, $B$
\State $[R : T] \leftarrow $ \texttt{UNIMODULARLY-ROW-REDUCE}($[A^T : I]$)
\State $Y_0 \longleftarrow $ \texttt{UNIFORMLY\_SAMPLE\_ROW\_ECHELON}($R^TY = B$).
\State Return $T^T Y_0$.
\end{algorithmic}
\label{Alg: Affine Sampling}
\end{algorithm}

We briefly describe the sampling technique used in~\cite{gilbert1990linear}. In summary, given the matrix equation \( AX = B \), where we aim to construct an integer sample, we first reduce it to a matrix \( T \) and an equation \( RY = C \), where \( R \) is in row-echelon form and there is a bijection between the solution spaces of these two equations, given by \( Y_0 = T^T X_0 \), where \( Y_0 \) is a solution to the latter equation and \( X_0 \) is a solution to the former. \\
Since \( R \) is in row-echelon form, we can derive integer vectors \( \{ b_1, b_2, \dots, b_t \} \), which are basis vectors of the null space of \( R \), along with a particular integer solution \( p_0 \) for \( RY = C \) such that any integer solution of \( RY = C \) can be expressed as \( Y = p_0 + \sum_{i=1}^t \alpha_i b_i \), where the \( \alpha_i \)'s are integers. \\
Consequently, any solution of \( AX = B \) has the form 
\[
X = T^T p_0 + \sum_{i=1}^t \alpha_i T^T b_i,
\]
or equivalently 
\[
X = T^T p_0 + b' \alpha',
\]
where \( b' \) is a matrix with $T^Tb_i$'s as column vectors and \( \alpha' = (\alpha_1, \alpha_2, \dots)^T \). If we are interested in solutions within a bounded state space \( \mathcal{D}_{state}(n) \), we can compute upper and lower bounds for each \( \alpha_i \), say \( L_i \leq \alpha_i \leq U_i \), and then use rejection sampling on the vector \( \alpha' \) to construct a uniform sample \( X \). It is important to note that we can construct such a uniform sample effeciently.

\subsection{\texttt{LIA\_UNIFORM\_SAMPLE} procedure.}
We give the details of the \texttt{LIA\_UNIFORM\_SAMPLE} in Algorithm \ref{alg: LIAUS}. \texttt{LIA\_UNIFORM\_SAMPLE} takes as input a polytope-lattice point $\mathcal{P}$, checks if it is contained in a proper affine subspace of $\mathbb{Z}^n$ (where $n$ is the dimension of the state space) by a call to the $\texttt{IS\_AFFINE\_SPACE}$ procedure. If so, it calls \texttt{REJECTION\_SAMPLING} where the input procedure $\texttt{V\_UNIFORM\_SAMPLE}$ is set to \texttt{LINEAR\_DIOPHANTINE\_UNIFORM\_SAMPLER}. If not, it first computes a bounding hyperrectangle $\mathcal{R}$, and then calls \texttt{REJECTION\_SAMPLING} where the input procedure $\texttt{V\_UNIFORM\_SAMPLE}$ is set to a procedure which constructs a uniform sample from $\mathcal{R}$ (called \texttt{UNIFORM\_SAMPLE\_$\mathcal{R}$}).

\begin{algorithm}[H]
\caption{\texttt{LIA\_UNIFORM\_SAMPLE}} \label{alg: LIAUS}
\begin{algorithmic}[1]
\State \textbf{Input: } $\mathcal{P}$
\If{ \texttt{IS\_AFFINE\_CONTAINED}($\mathcal{P}$)} 
\State $(AX = B) \longleftarrow$ $\texttt{GET\_AFFINE\_SPACE}(\mathcal{P})$
\State \Return \texttt{REJECTION\_SAMPLING} ( $\mathcal{P}$,  \texttt{LINEAR\_DIOPHANTINE\_UNIFORM\_SAMPLER(A, B)})
\Else
\State $\mathcal{R} \longleftarrow \texttt{BOUNDING\_HYPERRECTANGLE}(\mathcal{P})$
\State \Return \texttt{REJECTION\_SAMPLING}($\mathcal{P}$, \texttt{UNIFORM\_SAMPLE\_$\mathcal{R}$})
\EndIf
\end{algorithmic}
\label{Alg: Affine Sampling}
\end{algorithm}

A key point of interest is how \texttt{IS\_AFFINE\_CONTAINED} checks if $\mathcal{P}$ is contained in a proper affine subspace. We implement \texttt{IS\_AFFINE\_CONTAINED} as a simple syntactic check procedure: it checks for equality predicates in the representation of $\mathcal{P}$. Although this implementation is sound, it is not complete, and we may falsely declare $\mathcal{P}$ to not be contained in an affine space when it actually is. However, completeness is not required, as sampling from the solution space of linear Diophantine equations is merely an optimization. Simply using rejection sampling from the bounding hyperrectangle will also produce uniform samples from $\mathcal{P}$, albeit more slowly. In practice, we do not encounter cases where our syntactic check implementation of \texttt{IS\_AFFINE\_CONTAINED} fails, hence our implementation suffices for our application.

Finally, for completeness, we describe how to construct uniform samples from \( U \) when the axis-aligned orthotope \( V \), containing \( U \), satisfies \( |V| \gg |U| \), but \( U \) is not contained in any proper affine subspace of \( \mathbb{R}^n \). Although this case does not typically arise in practice and hence isn't dealt with by Algorithm \ref{alg: LIAUS}, it is included here for completeness. In such instances, rather than finding a single axis-aligned orthotope containing \( U \), we can identify multiple orthotopes, say \( V_1, V_2, \dots, V_l \), for some hyperparameter \( l \), such that \( U \subseteq V_1 \sqcup V_2 \sqcup \dots \sqcup V_l \) and \( |V_1 \sqcup V_2 \sqcup \dots \sqcup V_l| \not\gg |U| \). Since \( U \) is not contained in any proper affine subspace of \( \mathbb{R}^n \), the required \( l \) will not be very large. And then we can use rejection sampling to generate uniform samples from $U$.

To generate a uniform sample from \( U \) using rejection sampling, we first need to uniformly sample from \( V_1 \sqcup V_2 \sqcup \dots \sqcup V_l \). This is done by selecting an index from \( [1, l] \), where index \( i \) is chosen with probability proportional to \( \Lambda(V_i) \) (i.e., the number of lattice points contained in \( V_i \)). Once an index, say \( k \), is selected, we then uniformly sample from \( V_k \), giving a uniform sample from \( V_1 \sqcup V_2 \sqcup \dots \sqcup V_l \).

\section{Supplementary Algorithms.} \label{App: suppAlg}
In this section, we give the algorithms for low-level procedures used in the paper.

\subsection{\texttt{ITERATED\_TRANSITION\_RELATION} Procedure.}
The algorithm for the \texttt{ITERATED\_TRANSITION\_RELATION} procedure is presented in Algorithm \ref{alg: ITR}. This procedure takes as input a set of states, \texttt{states}, and computes all $j_0$-iterated tails of these states for some hyperparameter $j_0$ by making repeated calls to the \texttt{TRANSITION\_RELATION} procedure. 

The \texttt{TRANSITION\_RELATION} procedure, detailed in Algorithm \ref{alg: TR}, takes as input a single state \texttt{hd} and returns all resulting tails for a single application of the transition relation $T$ on this state. To achieve this, it iterates over all piecewise linear integer relations $(\text{Linear\_Integer\_Relation}_i, B_i)$ in $T$. For each relation, it checks if \texttt{hd} is within the region defined by $B_i$. If \texttt{hd} belongs to $B_i$, the algorithm applies all possible linear functions defined in $\text{Linear\_Integer\_Relation}_i$ to \texttt{hd}, generating new states. These new states are collected and returned as the result.

\begin{algorithm}[H]
\caption{\texttt{ITERATED\_TRANSITION\_RELATION}} \label{alg: ITR}
\begin{algorithmic}[1]
\State \textbf{Input: } $T, states$
\State \textbf{Hyperparameters: } $j_0$
\State $iteratedTls \longleftarrow states$
\For{$j = 1$ to $j_0$}
\State $nextiteratedTls \longleftarrow \varnothing$
\For{$tl$ in $iteratedTls$}
\State $nextiteratedTls \longleftarrow nextiteratedTls \cup \texttt{TRANSITION\_RELATION}(tl)$
\EndFor
\State $iteratedTls \longleftarrow nextiteratedTls$
\EndFor
\State \Return $iteratedTls$
\end{algorithmic}
\label{Alg: Affine Sampling}
\end{algorithm}

\begin{algorithm}[H]
\caption{\texttt{TRANSITION\_RELATION}} \label{alg: TR}
\begin{algorithmic}[1]
\State \textbf{Input: } $T, hd$
\State $Tls \longleftarrow \varnothing$
\For{$( \text{Linear\_Integer\_Relation}_i , B_i)$ in $T$}
\If{$hd \in B_i$}
\For {\text{Linear\_Function} in $\text{Linear\_Integer\_Relation}_i$}
\State $Tls \longleftarrow Tls \cup \text{Linear\_Function}(hd)$
\EndFor
\EndIf
\EndFor
\State \Return $Tls$
\end{algorithmic}
\label{Alg: Affine Sampling}
\end{algorithm}

\subsection{\texttt{BOUNDED\_RANDOM\_WALK} Procedure.}
The \texttt{BOUNDED_RANDOM_WALK} algorithm, presented in Algorithm \ref{alg: BRW} takes as input a state space $\mathcal{X}$, a neighborhood function $\mathcal{N}$, a cost function $c$, an initial state $s_0$, and a number of iterations $t_{rw}$, and performs a random walk on the state space $\mathcal{X}$. Starting from the initial state $s_0$, it repeatedly samples a new state $s'$ uniformly from the neighborhood of the current state, updates the current state and its associated cost, and records the cost transitions. This process is repeated for $t_{rw}$ iterations, and the set of recorded transitions is returned.

\begin{algorithm} 
\caption{\texttt{BOUNDED\_RANDOM\_WALK}}\label{alg: BRW}
\begin{algorithmic}[1] 
\State \textbf{Input: } $(\mathcal{X}, \mathcal{N}, c , s_0, t_{rw})$
\State $s \longleftarrow s_0$
\State $c \longleftarrow c(s)$
\State $rv \longleftarrow \varnothing$
\For{$t = 1$ to $t_{rw}$}
\State $s' \longleftarrow \texttt{UNIFORM\_SAMPLE\_LIST}(\mathcal{N}(s))$
\State $c ' \longleftarrow c(s')$
\State $rv \longleftarrow rv \cup \{(c,c')\}$
\State $s \longleftarrow s'$
\State $c \longleftarrow c'$
\EndFor
\State \Return rv
\end{algorithmic}
\end{algorithm}

\subsection{\texttt{GET\_POSITIVE\_TRANSITIONS} Procedure.}
The \texttt{GET_POSITIVE_TRANSITIONS} algorithm, detailed in Algorithm \ref{alg: GPT}, filters a set of cost transitions, returning only those where the subsequent cost is greater than the current cost.
\begin{algorithm} 
\caption{\texttt{GET\_POSITIVE\_TRANSITIONS}}\label{alg: GPT}
\begin{algorithmic}[1] 
\State \textbf{Input: } $\delta$
\State $rv \longleftarrow \varnothing$
\For{ $(c,c')$ in $\delta$}
\If{ $c' > c$}
\State $rv \longleftarrow rv \cup \{ (c,c')\}$
\EndIf
\EndFor
\State \Return rv
\end{algorithmic}
\end{algorithm}

\subsection{\texttt{RANDOMIZED\_$\epsilon$\_NET} Procedure.}
Algorithm \ref{alg: enet} outlines the \texttt{RANDOMIZED\_$\epsilon$\_NET} procedure, which begins by determining the number of uniform samples $m$ required to construct such an $\epsilon$-net for $X \land \mathcal{D}_{state}$. These samples are generated using the \texttt{LIA\_UNIFORM\_SAMPLE} function (refer Appendix \ref{App: Uniform Sampling poly lIA}). Internally, \texttt{LIA\_UNIFORM\_SAMPLE} checks if $X$ is contained within an affine space by making simple syntactic checks on the representation of $X$. If so, it uses rejection sampling on the solution space of the corresponding linear Diophantine equation (refer to Section \ref{sec: enet}, Appendix \ref{App: Sampling from Affine Spaces.}); else, uniform samples are obtained via simple rejection sampling (refer to Section \ref{sec: enet}, Appendix \ref{App: Rejection Sampling}). \footnote{In practice, rather than directly constructing a randomized $\epsilon$-net for $X$, we construct randomized $\epsilon$-nets for each convex polytope $\mathcal{P}$ of $X$ (as $X$ is an LIA DNF, it is a union of convex polytopes $\mathcal{P}$ in $\mathbb{R}^n$). This approach is better because rejection sampling generates uniform samples from convex shapes much more efficiently. By Theorem \ref{Thm: HAW}, this approach ensures we obtain a randomized $\epsilon$-net for $X$ as a whole as well.}
\begin{algorithm}
\caption{\texttt{RANDOMIZED\_$\epsilon$\_NET}}\label{alg: enet}
\begin{algorithmic}[1]
\State \textbf{Input: } $(\epsilon, X, \delta, vc)$
\State $m \longleftarrow 0$
\While {$1 - 2\phi_{vc}(2m)2^{\frac{-\epsilon m}{2}} < \delta$}
\State $m \longleftarrow m + 1$
\EndWhile
\State $rv \longleftarrow \varnothing$
\For{$i = 1$ to $m$}
\State $rv \longleftarrow rv \bigcup \texttt{LIA\_UNIFORM\_SAMPLE}(X \land \mathcal{D}_{state})$ 
\EndFor
\State \Return{$rv$}
\end{algorithmic}
\end{algorithm}

\subsection{\texttt{INITIAL_DATASET} Procedure.}
The details of the \texttt{INITIAL\_DATASET} are given in Algorithm \ref{alg: initDS}. The \texttt{INITIAL\_DATASET} procedure first computes the dimension of the state space by computing the length of the coefficient vector in any atomic predicate of $P$ using the \texttt{GET\_DIM} procedure. It then sets the global variable $\epsilon$ accordingly. Next, it sets $+$, $HEAD(\rightarrow)$, and $-$ to be $\epsilon$-nets of $P$, $B$, and $\neg Q$ with respect to ellipsoids as ranges by invoking the \texttt{RANDOMIZED\_$\epsilon$\_NET} procedure. Finally, it computes $\rightarrow$ using $HEAD(\rightarrow)$ and the transition relation $T$, and outputs the constructed dataset.
\begin{algorithm}
\caption{\texttt{INITIAL\_DATASET}}\label{alg: initDS}
\begin{algorithmic}[1]
\State \textbf{Input: } $(P, B, T, Q)$
\State \textbf{Hyperparameters: } $(\epsilon_0, \delta_0)$
\State $n \longleftarrow \texttt{GET\_DIM}(P)$
\State $\epsilon \longleftarrow \epsilon_0$
\State $+ \longleftarrow \texttt{RANDOMIZED\_$\epsilon$\_NET} (\epsilon, P, \delta_0, \frac{n^2 + 3n}{2})$
\State $HEAD(\rightarrow) \longleftarrow \texttt{RANDOMIZED\_$\epsilon$\_NET} (\epsilon, B, \delta_0, \frac{n^2 + 3n}{2})$
\State $- \longleftarrow \texttt{RANDOMIZED\_$\epsilon$\_NET} (\epsilon, \neg Q, \delta_0, \frac{n^2 + 3n}{2})$
\State $\rightarrow \:\: \longleftarrow \{ (\vec{t}, \vec{t'}) \: : \: \vec{t} \in HEAD(\rightarrow) \land T(\vec{t}, \vec{t'}) \}$
\State \Return{$(+, \rightarrow , -)$}
\end{algorithmic}
\end{algorithm}

\subsection{\texttt{REFINED_DATASET} Procedure.}
The details of the \texttt{REFINED\_DATASET} are given in Algorithm \ref{alg: suppDS}. The \texttt{REFINED\_DATASET} procedure first computes the dimension of the state space from $P$ using the \texttt{GET\_DIM} procedure. It then sets a new, lower value for the global variable $\epsilon$ (by dividing the previous value by 2). Next, it sets $+$, $HEAD(\rightarrow)$, and $-$ to be $\epsilon$-nets of $P$, $B$, and $\neg Q$ with respect to ellipsoids as ranges by invoking the \texttt{RANDOMIZED\_$\epsilon$\_NET} procedure.\footnote{In practice, we use an optimization here: Using Theorem \ref{Thm: HAW}, to construct an $\epsilon$-net of $X$ with respect to some range space, we need to construct $m(\epsilon)$ uniform samples of $X$. If we already have an $\epsilon$-net of $X$ and want to construct an $\epsilon'$-net of $X$ for the same range space, where $\epsilon' < \epsilon$, we only need to add $m(\epsilon') - m(\epsilon)$ new uniform samples to the previously constructed $\epsilon$-net. Therefore, we do not need to discard the previously constructed $\epsilon$-net and can simply append to it to construct an $\epsilon'$-net. } Finally, it computes $\rightarrow$ using $HEAD(\rightarrow)$ and the transition relation $T$, and outputs the constructed dataset.

\begin{algorithm}
\caption{\texttt{REFINED\_DATASET}}\label{alg: suppDS}
\begin{algorithmic}[1]
\State \textbf{Input: } $(P, B, T, Q)$
\State \textbf{Hyperparameters: } $\delta_0$
\State $n \longleftarrow \texttt{GET\_DIM}(P)$
\State $\epsilon \longleftarrow \frac{\epsilon}{2}$
\State $+ \longleftarrow \texttt{RANDOMIZED\_$\epsilon$\_NET} (\epsilon, P, \delta_0, \frac{n^2 + 3n}{2})$
\State $HEAD(\rightarrow) \longleftarrow \texttt{RANDOMIZED\_$\epsilon$\_NET} (\epsilon, B, \delta_0, \frac{n^2 + 3n}{2})$
\State $- \longleftarrow \texttt{RANDOMIZED\_$\epsilon$\_NET} (\epsilon, \neg Q, \delta_0, \frac{n^2 + 3n}{2})$
\State $\rightarrow \:\: \longleftarrow \{ (\vec{t}, \vec{t'}) \: : \: \vec{t} \in HEAD(\rightarrow) \land T(\vec{t}, \vec{t'}) \}$
\State \Return{$(+, \rightarrow , -)$}
\end{algorithmic}
\end{algorithm}

\subsection{\texttt{REFINED_DATASET\_CRITERION} Procedure.}
The algorithm for this procedure is given in Algorithm \ref{alg: refineDS}.

\begin{algorithm}
\caption{\texttt{REFINE\_DATASET\_CRITERION}}\label{alg: refineDS}
\begin{algorithmic}[1]
\State \textbf{Input: } $t$
\State \textbf{Hyperparameters: } $t_{\text{refine}}$
\State \Return{$(t \bmod t_{\text{refine}} \equiv 0)$}
\end{algorithmic}
\end{algorithm}

\subsection{\texttt{CHC\_VERIFIER} Procedure.}
We detail the \texttt{CHC_VERIFIER} procedure in Algorithm \ref{alg: chcVERIFIER}. \texttt{CHC\_VERIFIER} begins by converting the CHC to a bounded CHC, adding constraints on the state space variables through the \texttt{STATE\_SPACE} procedure. The \texttt{STATE\_SPACE} procedure returns an LIA formula asserting that the input state must lie in $\mathcal{D}_{state}$. This step differs slightly for inductive clauses compared to fact or query clauses, as the inductive clause quantifies over a pair of states, whereas the other two quantify over a single state. The \texttt{INDUCTIVECLAUSE} procedure is used to check if the input CHC is an inductive clause. Next, the \texttt{CHC\_VERIFIER} procedure checks if the \texttt{BOUNDED\_CHC} is valid by calling an SMT solver to verify \texttt{BOUNDED\_CHC} using the \texttt{SMT\_VALID} procedure. If the check is valid, it returns \texttt{True} with no counterexamples. If not, it repeatedly probes the SMT solver for up to \texttt{\text{cex}_{\text{max}}} models of the negation of the \texttt{BOUNDED\_CHC} via the \texttt{SMT\_MODEL} procedure. To ensure we sample dispersed CEX, we add the condition $\bigwedge_{\vec{t} \in \texttt{CEX}} \texttt{L1_NORM}(\vec{s} - \vec{t}) \ge d_0$, where \texttt{CEX} stores the already sampled cex models, $d_0 > 0$ is a hyperparameter, and \texttt{L1\_NORM} is a procedure that returns the L1-norm of a vector in $\mathbb{R}^n$. Note that this process is slightly different for inductive clauses, which have pairs of states as cex, wherin we ensure that the heads of the ICE pairs are dispersed. Finally, \texttt{CHC\_VERIFIER} returns \texttt{False} and the set of sampled cex.
\begin{algorithm}
\caption{\texttt{CHC\_VERIFIER}}\label{alg: chcVERIFIER}
\begin{algorithmic}[1]
\State \textbf{Input: } $\texttt{CHC}$
\State \textbf{Hyperparameters: } $\text{cex}_{\text{max}}$, $d_0$
\If{$ \neg \texttt{INDUCTIVECLAUSE}(\texttt{CHC})$}
\State $\texttt{BOUNDED\_CHC} \longleftarrow \forall \vec{s} \:\; \texttt{CHC}(\vec{s}) \land \texttt{STATE\_SPACE}(\vec{s}) $
\Else
\State $\texttt{BOUNDED\_CHC} \longleftarrow \forall (\vec{s}, \vec{s'}) \:\; \texttt{CHC}(\vec{s}, \vec{s'}) \land \texttt{STATE\_SPACE}(\vec{s}) \land \texttt{STATE\_SPACE}(\vec{s'}) $
\EndIf
\State $\texttt{CORRECT} \longleftarrow \texttt{True}$
\State $\texttt{CEX} \longleftarrow \varnothing$
\If{ $\neg \texttt{SMT\_VALID}(  \texttt{BOUNDED\_CHC}  ) $} 
\State $\texttt{CORRECT} \longleftarrow \texttt{False}$
\For{$i = 1$ to $\text{cex}_{\text{max}}$}
\If{$ \neg \texttt{INDUCTIVECLAUSE}(\texttt{CHC})$}
\State { \small $\texttt{CEX} \longleftarrow \texttt{CEX} \bigcup \texttt{SMT\_MODEL}\Bigg( \forall \vec{s} \:\; \neg \Big( \texttt{BOUNDED\_CHC}(\vec{s}) \Big) \land \Big( \bigwedge_{\vec{t} \in \texttt{CEX}}  \texttt{L1\_NORM}(\vec{s} - \vec{t} ) \ge  d_0 \Big) \Bigg)$ }
\Else
\State { \small $\texttt{CEX} \longleftarrow \texttt{CEX} \bigcup \texttt{SMT\_MODEL}\Bigg( \forall (\vec{s}, \vec{s'}) \:\; \neg \Big( \texttt{BOUNDED\_CHC}(\vec{s}, \vec{s'}) \Big) \land \Big( \bigwedge_{(\vec{t}, \vec{t'}) \in \texttt{CEX}}  \texttt{L1\_NORM}(\vec{s} - \vec{t} ) \ge  d_0 \Big) \Bigg)$ }
\EndIf
\EndFor
\EndIf
\State \Return{$(\texttt{CORRECT}, \texttt{CEX} )$}
\end{algorithmic}
\end{algorithm}

\subsection{\texttt{ITERATED\_IMPLICATION\_PAIRS} Procedure.}

We give the algorithm of the \texttt{ITERATED\_IMPLICATION\_PAIRS} procedure in Algorithm \ref{alg:iteratedICE}.
\begin{algorithm}
\caption{\texttt{ITERATED\_IMPLICATION\_PAIRS}}\label{alg:iteratedICE}
\begin{algorithmic}[1]
\State \textbf{Input: } $T, ICEpairs$
\State \textbf{Hyperparameters: } $k_0$
\State $iteratedICEpairs \;\; \longleftarrow \varnothing$
\For{$(hd,tl)$ in $ICEpairs$}
\State $tail \longleftarrow tl$
\For{$i = 1$ to $k_0-1$}
\State $iteratedTails \longleftarrow \texttt{ITERATED\_TRANSITION\_RELATION}(T, tail)$
\State $newTail \longleftarrow \texttt{RANDOM\_CHOICE}(iteratedTails)$
\If{$\neg B(newTail)$}
\State \textbf{break}
\EndIf
\State $tail \longleftarrow newTail$
\EndFor
\State $iteratedICEpairs  \longleftarrow iteratedICEpairs \cup \{ (hd, tail) \}$
\EndFor
\State \Return{$iteratedICEpairs$}
\end{algorithmic}
\end{algorithm}

The \texttt{ITERATED\_IMPLICATION\_PAIRS} procedure takes as input the piecewise linear integer transition relation $T$ and a set of ICE pairs and returns a set of iterated ICE pairs. For each ICE pair in its input, it iteratively computes the set of tails for the current tail (note that our transition relation is nondeterministic, allowing for multiple tails for a single head) using the \texttt{ITERATED\_TRANSITION\_RELATION} procedure (refer Appendix \ref{App: suppAlg}) and randomly selects a new iterated tail using the \texttt{RANDOM\_CHOICE} procedure. It then checks if the iterated tail remains within the loop guard; if not, it terminates the iteration; otherwise, it continues iterating up to a maximum depth of \( k_0 \), a hyperparameter. Finally, it stores the newly computed iterated ICE pair. 

\subsection{\texttt{SIMULATED\_ANNEALING\_PERIODIC\_CHECK} Procedure.}
The \texttt{SIMULATED\_ANNEALING\_PERIODIC\_CHECK} algorithm, described in Algorithm \ref{alg:SimAPeriodicCheck}, modifies the \\ \texttt{SIMULATED\_ANNEALING} algorithm described in algorithm \ref{alg: SimA} by periodically checking if an invariant has been found by other threads. It takes as input the thread ID $tID$, the search space $\mathcal{X}$, the neighborhood function $\mathcal{N}$, the cost function $c$, and the initial invariant $I_0$, along with hyperparameters $t_{max}$, $t_{check}$ and $p$. The shared global memory $IList$ is used for communication between threads. The algorithm follows the simulated annealing procedure but includes periodic checks (every $t_{\text{check}}$ iterations) to see if any other thread has found a valid invariant, in which case it terminates early. And like before, if the cost $c$ reaches zero, indicating an approximate invariant, it returns the invariant; otherwise, it continues to search until the maximum number of iterations is reached.

\begin{algorithm} 
\caption{\texttt{SIMULATED\_ANNEALING\_PERIODIC\_CHECK}}\label{alg:SimAPeriodicCheck}
\begin{algorithmic}[1] 
\State \textbf{Input: } $(tID, \mathcal{X}, \mathcal{N}, c , I_0)$
\State \textbf{Hyperparameters: } $(t_{max}, p, t_{check})$
\State \textbf{Shared Global Memory: } $IList$
\State $I \longleftarrow I_0$
\State $c \longleftarrow c(I)$
\State $T_0 \longleftarrow \texttt{INITIAL\_TEMPERATURE}(\mathcal{X}, \mathcal{N}, c)$
\For{$t = 1$ to $t_{max}$}
\If {$c = 0$} 
\State \Return $I$
\EndIf
\If{$t \bmod t_{\text{check}} \equiv 0$}
\For{$i=1$ to $p$}
\If{$(i \neq tID) \land (IList[i] \neq \texttt{Fail})$}
\State \Return \texttt{Fail}  
\EndIf
\EndFor
\EndIf
\State $T \longleftarrow \frac{T_0}{\ln 1 + t}$
\State $s' \longleftarrow \texttt{UNIFORM\_SAMPLE\_LIST}(\mathcal{N}(I))$
\State $c ' \longleftarrow c(I')$
\State $a \longleftarrow e^{-\frac{(c' - c)^+}{T}}$
\If{$\texttt{UNIFORM\_SAMPLE\_INTERVAL}([0,1]) \le a$}
\State $I \longleftarrow I'$.
\State $c \longleftarrow c'$
\EndIf
\EndFor
\State \Return \texttt{FAIL}
\end{algorithmic}
\end{algorithm}

\subsection{\texttt{PARALLEL\_SIMULATED\_ANNEALING} Procedure.}
We give the algorithm for \texttt{PARALLEL\_SIMULATED\_ANNEALING} in Algorithm \ref{alg:Parallel SA}. This procedure takes as input a dataset $\mathcal{S}$, an initial invariant $I_0$, and hyperparameters $p$ and $kList$. The shared global memory, $IList$, is used to store the results from the parallel threads. For each value of $k$ in $kList$, the procedure initializes the search space $\mathcal{X}[k]$, the neighborhood function $\mathcal{N}$, and the cost function $c$ by calling \texttt{GET\_SA\_PARAMETERS}. Then, it forks $p$ threads, each running the \texttt{SIMULATED\_ANNEALING\_PERIODIC\_CHECK} procedure with the respective parameters. After joining all threads, it returns the first successful invariant from $IList$ or fails if none are found. The \texttt{SIMULATED\_ANNEALING\_PERIODIC\_CHECK} algorithm, as described in Algorithm \ref{alg:SimAPeriodicCheck}, is a variant of \texttt{SIMULATED\_ANNEALING} (refer Algorithm \ref{alg: SimA}) that includes checks across threads using shared global memory to monitor progress.

\begin{algorithm} 
\caption{\texttt{PARALLEL\_SIMULATED\_ANNEALING}}\label{alg:Parallel SA}
\begin{algorithmic}[1] 
\State \textbf{Input: } $(\mathcal{S}, I_0)$
\State \textbf{Hyperparameters: } $p$ , $kList$
\State \textbf{Shared Global Memory:} $IList$
\For{$i=1$ to $p$}
\State $k \longleftarrow kList[i]$
\State $(\mathcal{X}[k], \mathcal{N}, c) \longleftarrow \texttt{GET\_SA\_PARAMETERS}(\mathcal{S}, k)$
\State $IList[i] \longleftarrow  \texttt{FORK}(\texttt{SIMULATED_ANNEALING_PERIODIC_CHECK}(i, \mathcal{X}[k], \mathcal{N}, c, I_0))$
\EndFor
\For{$i = 1$ to $p$} \texttt{JOIN}(i)
\EndFor
\For{$i = 1$ to $p$}
\If{$IList[i] \neq \texttt{Fail}$} \Return $IList[i]$
\EndIf
\EndFor
\State \Return \texttt{Fail}
\end{algorithmic}
\end{algorithm}

\subsection{\texttt{INITIAL\_INVARIANT} Procedure.}
We provide the algorithm for \texttt{INITIAL\_INVARIANT} in Algorithm \ref{alg:initInv}. This procedure takes a dataset $\mathcal{S}$ as input along with hyperparameters $l_0$ and $k$ (maximum coefficient bound). It starts by computing the dimension of the state space using \texttt{GET\_DIM}. It then conducts $l_0$ trials, where each trial uniformly samples coefficients for each predicate in the invariant and sets the constant of each predicate to zero. The procedure concludes by selecting and returning the invariant with the lowest cost. We next give the full parallel SA search algorithm.

\begin{algorithm}
\caption{\texttt{INITIAL\_INVARIANT}}\label{alg:initInv}
\begin{algorithmic}[1]
\State \textbf{Input: } $\mathcal{S}$
\State \textbf{Hyperparameters: } $l_0$, $k$
\State $n \longleftarrow \texttt{GET\_DIM}(\mathcal{S})$
\State $I \longleftarrow \varnothing$
\State $c \longleftarrow \infty$
\For{$l = 1$ to $l_0$}
\State $I' \longleftarrow \texttt{False}$
\For{$j = 1$ to $d$}
\State $CC \longleftarrow \texttt{True}$
\For{$i = 1$ to $c$}
\State $\vec{w} \longleftarrow \texttt{UNIFORM\_SAMPLE\_LIST}([-k,k]^n)$
\State $CC \longleftarrow CC \land (\vec{w}\cdot\vec{x} \le 0)$
\EndFor
\State $I' \longleftarrow I' \lor CC$
\EndFor
\State $c' \longleftarrow c(I')$
\If{$c' < c$}
\State $I \longleftarrow I'$
\State $c \longleftarrow c'$
\EndIf
\EndFor
\State \Return I
\end{algorithmic}
\end{algorithm}

\section{Proofs for Lemmas and Theorems.} \label{App: Proofs for Lemmas in Methodology}

\begin{proof}[Proof of Lemma \ref{Lemma: CEX spaces}]
To prove the first half of the lemma we only need to observe that the negation of a CHC clause of the form $Body \rightarrow Head$ is given by $Body \land \neg Head$. The latter half of the lemma follows trivially from the following 5 results:
\begin{enumerate}
    \item If $p$ is a predicate derived from LIA theory, then $\neg p$ is also a single predicate derived from LIA theory (where we follow the convention of using only inequailty LIA predicates; and that all equality LIA predicates are originally represented as two inequality predicates).
    \item If $F$ is a $d$-$c$ DNF, then $\neg F$ is a $c^d$-$d$ DNF.
    \item If $F_1$ is a $d_1$-$c_1$ DNF and $F_2$ is a $d_2$-$c_2$ DNF, then $F_1 \lor F_2$ is a $(d_1 + d_2)$-$\max (c_1, c_2)$ DNF and $F_1 \land F_2$ is a $(d_1d_2)$-$(c_1 + c_2)$ DNF.
    \item Let $T$ be a $d_T$-piecewise linear integer $r_T$-relation over $B$ where each partition block is a $d_{B_T}$-$c_{B_T}$ DNF and $F$ is a $d$-$c$ DNF, then $T^{-1}(F) := \{ \vec{t} \: : \: T(\vec{t}, \vec{t'}) \land  F(\vec{t'}) \}$ can be represented as $(dd_tr_Td_{B_T})$-$(c + c_{B_T})$ DNF.
    \item Let $I$ be an approximate invariant. If $c_I^\rightarrow = \{ (\vec{t}, \vec{t'}) \: : \: I(\vec{t}) \land B(\vec{t}) \land T(\vec{t}, \vec{t'}) \land \neg I(\vec{t'}) \}$, then $HEAD(c_I^\rightarrow) = \{ \vec{t} \: : \: I(\vec{t}) \land B(\vec{t}) \land \vec{t} \in T^{-1}(\neg I)\}$ where $T^{-1}(\cdot)$ is defined in the previous point.
\end{enumerate}
Points (1)--(3) and (5) are straightforward. We now prove (4):

It is given that there exists a \(d_T\)-sized partition \(\{B_i\}\) of \(B\), where each \(B_i\) is a \(d_{B_T}\)-\(c_{B_T}\) DNF such that \(T(\vec{t}, \vec{t'}) \equiv \bigvee_{i=1}^{d_T} B_i(\vec{t}) \rightarrow \bigvee_{j=1}^{r} (\vec{t'} = t_{ij}(\vec{t}))\), where \(t_{ij}\) are linear functions with integer coefficients.

If \(F = \bigvee_{k=1}^{d} C^F_k\) where \(C^F_k\) are \(c\)-cubes (i.e., the conjunction of \(c\) predicates), then \(T^{-1}(F)\) can be represented by the formula \(\bigvee_{k=1}^{d} T^{-1}(C^F_k)\). Furthermore, \(T^{-1}(C^F_k)\) can be represented by the formula \(\bigvee_{i=1}^{d_T} \left( B_i \land \left( \bigvee_{j=1}^{r} t_{ij}^{-1}(C^F_k) \right) \right)\), where \(t_{ij}^{-1}(C^F_k)\) is defined as the set \(\{\vec{t} : \vec{t'} = t_{ij}(\vec{t}) \land C^F_k(\vec{t'})\}\).

If the dimension of the state space is \(n\), then any \(c\)-cube \(C^F_k\) can be represented by the matrix equation \(A(C^F_k)X \leq b(C^F_k)\), where \(A(C^F_k)\) is a \((n + 1) \times c\) matrix, and \(X\) and \(b(C^F_k)\) are \((n + 1)\)-dimensional column vectors (the additional dimension is used to model linear transformations in \(\mathbb{R}^n\) by homogeneous transformations in \(\mathbb{R}^{n+1}\)).

Since each \(t_{ij}\) is a linear function in \(\mathbb{R}^n\), it can be modeled as a \((n + 1) \times (n + 1)\) matrix denoted \(M(t_{ij})\). It is then straightforward to see that the set \(t_{ij}^{-1}(C^F_k)\) can be represented by the matrix equation \((M(t_{ij})A(C^F_k))X \leq b(C^F_k)\). This matrix equation defines a \(c\)-cube, which implies that \(t_{ij}(C^F_k)\) can be represented by a \(c\)-cube (specifically the one defined by the previous matrix equation).

Then by (3), we get that \(T^{-1}(C^F_k)\) can be represented by a \(d_T r_T d_{B_T}-(c_{B_T} + c)\) DNF, which implies the result.

\end{proof}

\begin{proof}[Proof of Theorem \ref{Thm: DS Loop Conv Guar}]
Let \( I(t) \) be any approximate \( d \)-\( c \) LIA invariant after \( k t_{\text{refine}} + k' \) iterations of the DS Loop for some \( k \in \mathbb{N} \) and \( 0 < k' < t_{\text{refine}} \). If the dataset sampled at this stage is \((+, \rightarrow, -)\), then \( + \), \( \text{HEAD}(\rightarrow) \), and \( - \) must contain subsets of \( P \), \( B \), and \( \neg Q \), respectively, such that these subsets are \( \frac{\epsilon_0}{2^k} \)-nets of \( P \), \( B \), and \( \neg Q \) with respect to ellipsoids as ranges with probability at least \( \delta_0 \) each.

By Lemma \ref{Lemma: CEX spaces}, for any counterexample dataset \( \mathcal{C}_I = (c_I^+, c_I^\rightarrow, c_I^-) \) for \( I \), \( c_I^+ \), \( c_I^\rightarrow \), and \( c_I^- \) are unions of some polytopes in \( \mathbb{R}^n \). Let \( \mathcal{P} \) be any such polytope of \( c_I^+ \) and let \( \mathcal{E}_{inn}(\mathcal{P}) \) and \( \mathcal{E}_{out}(\mathcal{P}) \) be the inner and outer Lowner-John ellipsoids for \( \mathcal{P} \), respectively. By the definition of an \( \epsilon \)-net, we must have \( \Lambda(\mathcal{E}_{inn}(\mathcal{P})) < \frac{\epsilon_0}{2^k} \cdot \Lambda(P) \).

Using Theorem \ref{Thm: JLEllipsoid} and Theorem \ref{Thm: Ratio of Ellipses}, we have:
\[
\Lambda(\mathcal{E}_{out}(\mathcal{P})) \leq \frac{\epsilon_0}{2^k} n^n \left( \frac{1 + \frac{C}{n^{2 - \frac{2}{n+1}}}}{1 - C} \right) \cdot \Lambda(P)
\]
where the constant \( C > 0 \) is defined in Theorem \ref{Thm: Ratio of Ellipses}. As \( \mathcal{P} \subseteq \mathcal{E}_{out}(\mathcal{P}) \), this implies that:
\[
\Lambda(\mathcal{P}) \leq \frac{\epsilon_0}{2^k} n^n \left( \frac{1 + \frac{C}{n^{2 - \frac{2}{n+1}}}}{1 - C} \right) \cdot \Lambda(P)
\]

By Lemma \ref{Lemma: CEX spaces}, there are at most \( c^d d_P \) such polytopes \( \mathcal{P} \), which implies that:
\[
|c_I^+| \leq c^d d_P \frac{\epsilon_0}{2^k} n^n \left( \frac{1 + \frac{C}{n^{2 - \frac{2}{n+1}}}}{1 - C} \right) \cdot \Lambda(P)
\]

In a similar manner, we can bound \( |c_I^\rightarrow| \) and \( |c_I^-| \).

Additionally, there are no more positive counterexamples for \( I \) if \( \Lambda(\mathcal{P}) < 1 \), which happens with probability at least \( \delta_0 \) when:
\[
\frac{\epsilon_0}{2^k} n^n \left( \frac{1 + \frac{C}{n^{2 - \frac{2}{n+1}}}}{1 - C} \right) \cdot \Lambda(P) < 1 \Leftrightarrow k > \log_2 \left( \epsilon_0 n^n \left( \frac{1 + \frac{C}{n^{2 - \frac{2}{n+1}}}}{1 - C} \right) \cdot \Lambda(P) \right)
\]

Similarly, with probability at least \( \delta_0 \), when \( k >  \log_2 \left( \epsilon_0 n^n \left( \frac{1 + \frac{C}{n^{2 - \frac{2}{n+1}}}}{1 - C} \right) \cdot \Lambda(B) \right) \), there are no more ICE pairs and, with probability at least \( \delta_0 \), when \( k >  \log_2 \left( \epsilon_0 n^n \left( \frac{1 + \frac{C}{n^{2 - \frac{2}{n+1}}}}{1 - C} \right) \cdot \Lambda(\neg Q) \right) \), there are no more negative cex. This implies that with probability at least \( \delta_0^3 \), there are no cex for an approximate invariant when \( k >  \log_2 \left( \epsilon_0 n^n \left( \frac{1 + \frac{C}{n^{2 - \frac{2}{n+1}}}}{1 - C} \right) \cdot \max( \Lambda(P), \Lambda(B), \Lambda(\neg Q)) \right) \). Consequently, we halt and the approximate invariant is the correct loop invariant.
\end{proof}

\begin{lemma} \label{Lemma: SA constant bound}
    Assume \(\mathcal{D}_{\text{state}}\) has the origin as its center of symmetry. Then if \(\mathcal{X}^*(k)\) denotes an invariant space with no constraint on the constants of its candidate invariants, i.e.,
\[ 
\mathcal{X}^*(k) = \left\{ \Big( \forall \vec{x} \in \mathcal{D}_{\text{state}} \: \bigvee_{i=1}^d \bigwedge_{j=1}^c \vec{w}_{ij} \cdot \vec{x} \le b_{ij} \Big) \: : \:
    \forall i \forall j \:  0 < \|\vec{w}_{ij}\|_{\infty} \le k  \right\} 
\]
then we have \(\mathcal{X}(k) = \mathcal{X}^*(k)\), where any two invariants \(I \in \mathcal{X}(k)\) and \(I' \in \mathcal{X}^*(k)\) are equal if \(I \leftrightarrow I'\).
\end{lemma}

\begin{proof}[Proof of Lemma \ref{Lemma: SA constant bound}]
    To show the lemma, it suffices to show that for any LIA predicate $p \equiv \vec{w}^T \vec{x} \le b$ s.t. $|| \vec{w} ||_{\infty} \le k$, there exists another LIA predicate $p' \equiv \vec{w}^T \vec{x} \le b'$ s.t. $b' \le k'$ and $p \leftrightarrow p'$. WLOG let $b > k'$. Then set $b' = k'$, we can clearly check that $\forall \vec{x} p'(\vec{x}) \Rightarrow p(\vec{x})$ as $b' < b$. Now suppose $\exists \vec{x} \in \mathcal{D}_{state} \: p(\vec{x}) \land \neg p'(\vec{x})$, which implies $|| \vec{x} ||_2 > \frac{b'}{||\vec{w}||_2} > \rho(\mathcal{D}_{state})$ which further implies that $\vec{x} \not\in \mathcal{D}_{state}$, which is a contradiction. This proves the lemma.
\end{proof}

\begin{proof}[Proof of Lemma \ref{Lemma: Properties of delta}]
(1) and (3) trivially hold. For (2), we show the following property for $\delta_{approx}$:
\[
(\forall I)  (\forall \vec{t}_1, \vec{t}_2 \in \mathbb{Z}^n) \quad \delta_{approx}(I, \vec{t}_1 + \vec{t}_2) \le \delta_{approx}(I, \vec{t}_1) +  || \vec{t}_2|| 
\]
To show this property, note that from $(x+y)^+ \le (x)^+ + (y)^+$ and $(\alpha x)^+ = |\alpha| \cdot (x)^+$ for $\alpha \in \mathbb{R}$, we have
\begin{align*}
\delta_{approx}(\vec{w}^T\vec{x} \le b, \vec{t}_1 + \vec{t}_2) &=  \frac{(\vec{w}^T (\vec{t}_1 + || \vec{t}_2 || \cdot \hat{\vec{t}}_2) - b_i)^+}{||\vec{w}||} \\
&\le  \frac{(\vec{w}^T \vec{t}_1 - b_i)^+}{||\vec{w}||} + || \vec{t}_2 || \cdot \frac{(\vec{w}^T \hat{\vec{t}}_2)^+}{||\vec{w}||}  \\
&\le \delta_{approx}(\vec{w}^T\vec{x} \le b, \vec{t}_1) +  || \vec{t}_2 || 
\end{align*}

Now from this, we can show that, 
\begin{align*}
\delta_{approx}(\bigwedge_{i = 1}^c \vec{w}_i^T \vec{x} \le b_i, \vec{t}_1 + \vec{t}_2) &= \frac{1}{c} \cdot \sum_{i=1}^c \delta_{approx}(\vec{w}_i^T\vec{x} \le b_i, \vec{t}_1 + \vec{t}_2) \\
&\le \frac{1}{c} \cdot \sum_{i=1}^c \Big(\delta_{approx}(\vec{w}_i^T\vec{x} \le b_i, \vec{t}_1) +  || \vec{t}_2 || \Big) \\
&= \delta_{approx}(\bigwedge_{i=1}^c \vec{w}_i^T\vec{x} \le b_i, \vec{t}_1) + || \vec{t}_2 ||.
\end{align*}

Now suppose that $\delta_{approx}( \bigvee_{j=1}^d \bigwedge_{i=1}^c \vec{w}_{ij}^T \vec{x} \le b_{ij} , \vec{t}_1 ) = \delta_{approx}( \bigwedge_{i=1}^c \vec{w}_{ij_0}^T \vec{x} \le b_{ij_0} , \vec{t}_1 )$ for some $j_0 \in \mathbb{Z}_{[1,d]}$. Then,
\begin{align*}
\delta_{approx}( \bigvee_{j=1}^d \bigwedge_{i=1}^c \vec{w}_{ij}^T \vec{x} \le b_{ij} , \vec{t}_1 + \vec{t}_2 ) &= \min_{1 \le j \le d} \delta_{approx}( \bigwedge_{i=1}^c \vec{w}_{ij}^T \vec{x} \le b_{ij} , \vec{t}_1 + \vec{t}_2) \\
&\le \delta_{approx}( \bigwedge_{i=1}^c \vec{w}_{ij_0}^T \vec{x} \le b_{ij_0} , \vec{t}_1 + \vec{t}_2) \\
&\le \delta_{approx}( \bigwedge_{i=1}^c \vec{w}_{ij_0}^T \vec{x} \le b_{ij_0} , \vec{t}_1 ) + ||\vec{t}_2 || \\
&= \delta_{approx}( \bigvee_{j=1}^d \bigwedge_{i=1}^c \vec{w}_{ij}^T \vec{x} \le b_{ij} , \vec{t}_1 ) +  ||\vec{t}_2 ||.
\end{align*}
which proves the property stated above.

Let $\vec{t} \in \mathcal{D}_{state}$. Then there exists $\vec{t}^* \in \{ \vec{t} \in  \mathcal{D}_{state} \: : \: I(\vec{t})\}$ such that $d_I(\vec{t}) = || \vec{t} - \vec{t}^*||$. Then by the previous result, we have 
\[
\delta_{approx}(I, \vec{t}^* + (\vec{t} - \vec{t}^*)) \le \delta_{approx}(I, \vec{t}^*)  +  || \vec{t} - \vec{t}^*|| \le  d_I(\vec{t}).
\]
\end{proof}

\begin{lemma} \label{Lemma: normProp}
Let \( F_{\alpha, \beta} \) be the normalization function for \(\alpha > 1\) and \(\beta \ge 1\). Then the following properties hold:
\begin{itemize}
    \item \( F_{\alpha, \beta}(0) = 0 \)
    \item \( \lim_{x \to \infty} F_{\alpha, \beta}(x) = \frac{\alpha}{\beta} + 1 \)
    \item \( F_{\alpha, \beta} \) is strictly increasing
    \item \( F_{\alpha, \beta} \) is differentiable for all \( x \in \mathbb{R} \) and \( F_{\alpha, \beta}'(x) \le \frac{1}{\beta} \) for all \( x \)
    \item For any \( x < y \), \( F_{\alpha, \beta}(y) - F_{\alpha, \beta}(x) \le \frac{y - x}{\beta} \)
\end{itemize}
\end{lemma}

\begin{proof}
    It is straightforward to check that \( F_{\alpha, \beta}(0) = 0 \) and \( \lim_{x \to \infty} F_{\alpha, \beta}(x) = \frac{\alpha}{\beta} + 1 \). Additionally, it is trivial to verify that \( F_{\alpha, \beta} \) is strictly increasing. To show differentiability, we need only demonstrate that \( F_{\alpha, \beta} \) is differentiable at \( x = \alpha \). We have:
    \[
    \lim_{x \to \alpha^-} F_{\alpha, \beta}'(x) = \frac{1}{\beta} \quad \text{and} \quad \lim_{x \to \alpha^+} F_{\alpha, \beta}'(x) = \left\lvert \frac{4}{\beta} \frac{e^{\frac{-2(x - \alpha)}{\beta}}}{\left( 1 + e^{\frac{-2(x - \alpha)}{\beta}} \right)^2 } \right\rvert_{x = \alpha} = \frac{1}{\beta}
    \]
    which shows \( F_{\alpha, \beta} \) is differentiable at \( x = \alpha \). Additionally, if \( x \le \alpha \), \( F_{\alpha, \beta}'(x) = \frac{1}{\beta} \), and if \( x > \alpha \), \( F_{\alpha, \beta}'(x) < F_{\alpha, \beta}'(\alpha) = \frac{1}{\beta} \). Finally, \( x < y \Rightarrow F_{\alpha, \beta}(y) - F_{\alpha, \beta}(x) \le \frac{y - x}{\beta} \) follows from the Mean Value Theorem.
\end{proof}

\begin{proof}[Proof of Theorem \ref{thm: OurSAguarantee}.]
Using the terminology from Theorem \ref{thm: SAguarantee}, we first give bounds on the parameters of our invariant space. \\
\textit{Upper and Lower Bounds on \(r\).} \\
We first give bounds on $r$ which is defined by: $r = \min_{I \in \mathcal{X}(k) \setminus \mathcal{X}_m(k)} \max_{I' \in \mathcal{X}(k)} d(I, I')$ where $\mathcal{X}_m(k)$ denotes the set of all local maximas in $\mathcal{X}(k)$ and $d(I,I')$ denotes the length of the shortest path between $I$ and $I'$. \\
Our invariant space \(\mathcal{X}(k)\) doesn't allow predicates with coefficients \(\vec{w} = \vec{0}\). To analyze the lower bounds on \(r\) for \(\mathcal{X}(k)\), let us define a new state space \(\mathcal{Y}(k)\), which allows predicates with coefficients \(\vec{w} = \vec{0}\), i.e.,
$$
\mathcal{Y}(k) = \left\{ \Big( \forall \vec{x} \in \mathcal{D}_{\text{state}} \: \bigvee_{i=1}^d \bigwedge_{j=1}^c \vec{w}_{ij} \cdot \vec{x} \le b_{ij} \Big) \: : \:
    \forall i \forall j \: \|\vec{w}_{ij}\|_{\infty} \le k \land |b_{ij}| \le k'  \right\}
$$
It is easy to see that the center of the invariant space \(\mathcal{Y}(k)\) is the candidate invariant \(\lor_{j=1}^d \land_{i=1}^c \vec{0} \cdot \vec{x} \le 0\), and hence its radius \(r'\) is given by \(r' = cd (nk + k') = k cd \sqrt{n}(\sqrt{n} + \rho(\mathcal{D}_{state}))\). Since \(\mathcal{X}(k) \subseteq \mathcal{Y}(k)\), we have \(r \ge r' = k cd \sqrt{n}(\rho(\mathcal{D}_{state}) + \sqrt{n}) \).

For an upper bound on \(r\), consider the candidate invariant \(\lor_{j=1}^d \land_{i=1}^c \vec{e}_s \cdot \vec{x} \le 0\) where $\vec{e}_s$ is any standard basis vector of $\mathbb{R}^n$. Then the distance of the predicate $\vec{e}_s \cdot \vec{x} \le 0$ from any other predicate is atmost $((n-1)k + (k+1) + k') = (nk + k' + 1) = k\sqrt{n}( \sqrt{n} + \rho ( \mathcal{D}_{state} ) + \frac{1}{k\sqrt{n}} )$ Then we have, 
$$r \le kcd\sqrt{n}( \sqrt{n} + \rho ( \mathcal{D}_{state} ) + \frac{1}{k\sqrt{n}} ) \le kcd\sqrt{n}( \rho ( \mathcal{D}_{state} ) + \sqrt{n}) + cd$$
\textit{Value of \(w\).}\\
It is trivial to see that \(w := \min_{I \in \mathcal{X}(k)} \frac{1}{|\mathcal{N}(I)|} = \frac{1}{2cd(n+1)}\). \\[1\baselineskip]
\textit{Upper Bound on L.} \\
We give an upper bound on $L$ defined as: 
$$L = \max_{I \in \mathcal{X}^*(k)} \max_{I' \in \mathcal{N}(I)} |c(I) - c(I')| $$
Let $I := \bigvee_{j=1}^d \bigwedge_{i=1}^c (\vec{w}_{ij} \cdot \vec{x} \le b_{ij})$ and let $I'$ be any neighbor of $I$. WLOG let $I'$ differ with $I$ on all predicates except the predicate with index $(j_0,i_0)$ which is $\vec{w}'_{i_0j_0} \cdot \vec{x} \le b'_{i_0j_0}$ for $I'$.
Then we have:
\begin{align*}
|c(I) - c(I')| \le \frac{1}{3} \Bigg( &\frac{\sum_{\vec{p} \in +} | \delta (I, \vec{p}) - \delta (I', \vec{p})| }{\sum_{\vec{p} \in +} 1} + \\
&\frac{\sum_{(\vec{h},\vec{t}) \in \rightarrow} |\min (\delta (\neg I, \vec{h}), \delta (I, \vec{t})) - \min (\delta (\neg I', \vec{h}), \delta (I', \vec{t}))| }{\sum_{(\vec{h},\vec{t}) \in \rightarrow} 1} + \\
&\frac{\sum_{\vec{n} \in -} |\delta (\neg I, \vec{n}) - \delta (\neg I', \vec{n})| }{\sum_{\vec{n} \in -} 1} \Bigg) 
\end{align*}
Then by Lemma \ref{Lemma: normProp}, we have,
\begin{align*}
|c(I) - c(I')|  \le \frac{1}{3\beta} \Bigg( &\frac{\sum_{\vec{p} \in +} | \delta_{\text{approx}} (I, \vec{p}) - \delta_{\text{approx}} (I', \vec{p})| }{\sum_{\vec{p} \in +} 1} + \\
&\frac{\sum_{(\vec{h},\vec{t}) \in \rightarrow} |\min (\delta_{\text{approx}} (\neg I, \vec{h}), \delta_{\text{approx}} (I, \vec{t})) - \min (\delta_{\text{approx}} (\neg I', \vec{h}), \delta_{\text{approx}} (I', \vec{t}))| }{\sum_{(\vec{h},\vec{t}) \in \rightarrow} 1} + \\
&\frac{\sum_{\vec{n} \in -} |\delta_{\text{approx}} (\neg I, \vec{n}) - \delta_{\text{approx}} (\neg I', \vec{n})| }{\sum_{\vec{n} \in -} 1} \Bigg)
\end{align*}
Now we deal with the terms of the 3 separate summations case by case. \\[1\baselineskip]
\textit{Case 1.} Let $\vec{p} \in +$. WLOG let $ \delta_{\text{approx}} (I, \vec{p}) \ge \delta_{\text{approx}} (I', \vec{p})$. Let $j^*$ be the index of the cube in $I'$ s.t. $\delta_{approx}(I', \vec{p}) = \frac{1}{c} \sum_{i = 1}^c \frac{ ( \vec{w}_{ij^*} \cdot \vec{p} - b_{ij^*}  )^+  }{|| \vec{w}_{ij^*} ||}$. Then
\begin{align*}
    | \delta_{\text{approx}} (I, \vec{p}) - \delta_{\text{approx}} (I', \vec{p})| &= \delta_{\text{approx}} (I, \vec{p}) - \delta_{\text{approx}} (I', \vec{p}) \\
    &\le \frac{1}{c} \sum_{i = 1}^c \frac{ ( \vec{w}_{ij^*} \cdot \vec{p} - b_{ij^*}  )^+ - ( \vec{w}'_{ij^*} \cdot \vec{p} - b'_{ij^*}  )^+  }{|| \vec{w}_{ij^*} ||} \\
    &\le \frac{1}{c} \frac{ | ( \vec{w}_{i_0j_0} \cdot \vec{p} - b_{i_0j_0}  )^+ - ( \vec{w}'_{i_0j_0} \cdot \vec{p} - b'_{i_0j_0}  )^+ | }{|| \vec{w}_{i_0j_0} ||} \\
    &\le \frac{1}{c} \frac{ |  (\vec{w}_{i_0j_0} - \vec{w}'_{i_0j_0}) \cdot \vec{p} - (b_{i_0j_0} - b'_{i_0j_0})   | }{|| \vec{w}_{i_0j_0} ||} 
\end{align*}
where the final inequality follows from $| (x)^+ - (y)^+| \le |x - y|$.
Now we have two cases. \\
\textit{Case 1a.} The transition from $I$ to $I'$ is a constant transition.
Then $|  (\vec{w}_{i_0j_0} - \vec{w}'_{i_0j_0}) \cdot \vec{p} - (b_{i_0j_0} - b'_{i_0j_0})   | = |(b_{i_0j_0} - b'_{i_0j_0}) | \le 1$. \\
\textit{Case 1b.} The transition from $I$ to $I'$ is a coefficient transition. Then $|  (\vec{w}_{i_0j_0} - \vec{w}'_{i_0j_0}) \cdot \vec{p} - (b_{i_0j_0} - b'_{i_0j_0})   | = |(\vec{w}_{i_0j_0} - (\vec{w}_{i_0j_0} + \vec{e}_s) \cdot \vec{p} | \le | \vec{p}_s | \le || \vec{p} ||_\infty$, where $\vec{e_s}$ is any standard basis vector of $\mathbb{R}^n$.
Thus from both cases, we get 
\begin{align*}
    | \delta_{\text{approx}} (I, \vec{p}) - \delta_{\text{approx}} (I', \vec{p})| &\le \frac{\max (1, || \vec{p} ||_\infty )}{c ||\vec{w}_{i_0j_0}|| } \le \frac{\max (1, || \vec{p} ||_\infty )}{ck\sqrt{n} }
\end{align*}
\textit{Case 2.} Let $\vec{n} \in -$. WLOG let $ \delta_{\text{approx}} (\neg I, \vec{n}) \ge \delta_{\text{approx}} (\neg I', \vec{n})$.
Now if $I$ and $I'$ differ in exactly one predicate, then (a) the number of cubes and $I$ and $I'$ are equal and (ii) For each cube in $I$, there is a cube in $I'$ with all but exactly one predicate equal. Let $j^*$ be the index of the cube in $I'$ s.t. $\delta_{approx}(\neg I', \vec{n}) = \frac{1}{c} \sum_{i = 1}^c \frac{ ( b_{ij^*} + 1 - \vec{w}_{ij^*} \cdot \vec{p}   )^+  }{|| \vec{w}_{ij^*} ||}$. Then
\begin{align*}
    | \delta_{\text{approx}} (\neg I, \vec{n}) - \delta_{\text{approx}} (\neg I', \vec{n})| &= \delta_{\text{approx}} (\neg I, \vec{n}) - \delta_{\text{approx}} (\neg I', \vec{n}) \\
    &\le \frac{1}{c} \sum_{i = 1}^c \frac{ ( b_{ij^*} + 1 - \vec{w}_{ij^*} \cdot \vec{n}  )^+ - ( b'_{ij^*} + 1 -\vec{w}'_{ij^*} \cdot \vec{n}  )^+  }{|| \vec{w}_{ij^*} ||} \\
    &\le \frac{1}{c} \frac{ | ( b_{i_0j_0} + 1 - \vec{w}_{i_0j_0} \cdot \vec{n}   )^+ - ( b'_{i_0j_0} + 1 - \vec{w}'_{i_0j_0} \cdot \vec{n}   )^+ | }{|| \vec{w}_{i_0j_0} ||} \\
    &\le \frac{1}{c} \frac{ | (b_{i_0j_0} - b'_{i_0j_0}) - (\vec{w}_{i_0j_0} - \vec{w}'_{i_0j_0}) \cdot \vec{n}   | }{|| \vec{w}_{i_0j_0} ||} 
\end{align*}
By using the same argument as in \textit{Case 1}, we get, 
\begin{align*}
    | \delta_{\text{approx}} (I, \vec{n}) - \delta_{\text{approx}} (I', \vec{n})| &\le \frac{\max (1, || \vec{n} ||_\infty )}{ck\sqrt{n}}
\end{align*}
\textit{Case 3.} Let $(\vec{h}, \vec{t}) \in \rightarrow$. 
If we have $(\vec{h},\vec{t})$ satisfies both $I$ and $I'$, then $|\min (\delta_{\text{approx}} (\neg I, \vec{h}), \delta_{\text{approx}} (I, \vec{t})) = 0$. Now, consider the case where $(\vec{h}, \vec{t})$
doesn't satisfy at least one of $I$ or $I'$. Then using the inequality $|\min(a,b) - \min(c,d)| \le \max(|a - c|, |b - d|) +  \min(|a-b|, |c-d|)$, we have:
\begin{align*}
|\min (\delta_{\text{approx}} (\neg I, \vec{h}), \delta_{\text{approx}} (I, \vec{t}))| &- \min (\delta_{\text{approx}} (\neg I', \vec{h}), \delta_{\text{approx}} (I', \vec{t}))| \\
&\le \max( |\delta_{\text{approx}} (\neg I, \vec{h}) - \delta_{\text{approx}} (\neg I', \vec{h})| , |\delta_{\text{approx}} (I, \vec{t}) - \delta_{\text{approx}} (I', \vec{t})| ) \\
&\quad + \min(|\delta_{\text{approx}} (\neg I, \vec{h}) - \delta_{\text{approx}} (I, \vec{t})|, |\delta_{\text{approx}} (\neg I', \vec{h}) - \delta_{\text{approx}} (I', \vec{t})| )
\end{align*}
Using the same arguments as in \textit{Case 1 and 2}, we get $|\delta_{\text{approx}} (\neg I, \vec{h}) - \delta_{\text{approx}} (\neg I', \vec{h})| \le \frac{\max (1, || \vec{h} ||_\infty )}{ck\sqrt{n}}$ and $|\delta_{\text{approx}} (I, \vec{t}) - \delta_{\text{approx}} (I', \vec{t})| \le \frac{\max (1, || \vec{t} ||_\infty )}{ck\sqrt{n}}$. Now we give an upper bound for $\min(|\delta_{\text{approx}} (\neg I, \vec{h}) - \delta_{\text{approx}} (I, \vec{t})|, |\delta_{\text{approx}} (\neg I', \vec{h}) - \delta_{\text{approx}} (I', \vec{t})|)$. WLOG suppose $(\vec{h}, \vec{t})$ doesn't satisfy $I$.
\begin{align*}
    \min(|\delta_{\text{approx}} (\neg I, \vec{h}) - \delta_{\text{approx}} (I, \vec{t})|, |\delta_{\text{approx}} (\neg I', \vec{h}) - \delta_{\text{approx}} (I', \vec{t})|) &\le |\delta_{\text{approx}} (\neg I, \vec{h}) - \delta_{\text{approx}} (I, \vec{t})| \\
    &\le \max ( \delta_{\text{approx}} (\neg I, \vec{h}) , \delta_{\text{approx}} (I, \vec{t})) \\
    &\le \max ( d (\neg I, \vec{h}) , d (I, \vec{t}))
\end{align*}
where $d(I, \vec{h})$ denotes the Euclidean set distance of $\vec{h}$ from $I$, and the final inequality follows from Lemma \ref{Lemma: Properties of delta}. And as $(\vec{h}, \vec{t})$ don't satisfy $I$, we have $\vec{h} \in I$, $\vec{t} \in \neg I$, and hence we will have $\max ( d (\neg I, \vec{h}) , d (I, \vec{t})) \le || \vec{h} - \vec{t}||$. Thus we finally get:
\begin{align*}
|\min (\delta_{\text{approx}} (\neg I, \vec{h}), \delta_{\text{approx}} (I, \vec{t}))| &- \min (\delta_{\text{approx}} (\neg I', \vec{h}), \delta_{\text{approx}} (I', \vec{t}))| \\
&\le \frac{\max (1, || \vec{h} ||_\infty ) + \max (1, || \vec{t} ||_\infty )}{ck\sqrt{n}} + || \vec{h} - \vec{t} ||
\end{align*}
Substituting back into the original summation, we get
\begin{align*}
|c(I) - c(I')| &\le \frac{1}{3\beta} \Bigg( \frac{\sum_{\vec{p} \in +} \max (1, || \vec{p} ||_\infty )}{ck\sqrt{n}|+|} \\
&\quad + \frac{\sum_{\vec{n} \in -} \max (1, || \vec{n} ||_\infty )}{ck\sqrt{n}|-|} \\
&\quad + \frac{\sum_{(\vec{h}, \vec{t}) \in \rightarrow} \max( \max (1, || \vec{h} ||_\infty ) , \max (1, || \vec{t} ||_\infty ) )}{ck\sqrt{n}|\rightarrow|} \\
&\quad + \frac{\sum_{(\vec{h}, \vec{t}) \in \rightarrow} ||\vec{h} - \vec{t}|| }{|\rightarrow|} \Bigg)
\end{align*}
Now note that as $\vec{p} \in \mathbb{Z}^n$, we have $|| \vec{p} ||_{\infty} \in \mathbb{Z}_{\ge 0}$ and furthermore $|| \vec{p} ||_{\infty} = 0 \Leftrightarrow \vec{p} = \vec{0}$. As there can be atmost one plus point which is the origin vector, we have $\sum_{\vec{p} \in +} \max (1, || \vec{p} ||_\infty ) \le 1 + \sum_{\vec{p} \in +} || \vec{p} ||_\infty $. Similarly, we get $\sum_{\vec{n} \in +} \max (1, || \vec{n} ||_\infty ) \le 1 + \sum_{\vec{n} \in +} || \vec{n} ||_\infty $ and $\sum_{(\vec{h}, \vec{t}) \in \rightarrow} \max( \max (1, || \vec{h} ||_\infty ) , \max (1, || \vec{t} ||_\infty ) ) \le 1 + \sum_{(\vec{h}, \vec{t}) \in \rightarrow} \max(|| \vec{h} ||_\infty, || \vec{t} ||_\infty) $. Then we get, 
\begin{align*}
|c(I) - c(I')| \le \frac{1}{3\beta} \Bigg( \frac{1}{ck\sqrt{n}} \Bigg( \frac{\sum_{\vec{p} \in +} ||\vec{p}||_\infty}{|+|}  &+ \frac{\sum_{(\vec{h}, \vec{t}) \in \rightarrow} \max(||\vec{h}||_\infty, ||\vec{t}||_\infty)}{|\rightarrow|} + \frac{\sum_{\vec{n} \in -} ||\vec{n}||_\infty}{|-|}\Bigg) \\
&+ \frac{3}{ck\sqrt{n}} +  \frac{\sum_{(\vec{h}, \vec{t}) \in \rightarrow} ||\vec{h} - \vec{t}|| }{|\rightarrow|} \Bigg) \\
&=  \frac{1}{\beta ck\sqrt{n}} \cdot \kappa_\infty(\mathcal{S}) + \frac{1}{3\beta} \cdot \lambda(\rightarrow) + \frac{1}{\beta ck\sqrt{n}}
\end{align*}
where $\kappa_\infty(\mathcal{S}) = \frac{1}{3} \Bigg( \frac{\sum_{\vec{p} \in +} ||\vec{p}||_\infty}{|+|}  + \frac{\sum_{(\vec{h}, \vec{t}) \in \rightarrow} \max( ||\vec{h}||_\infty , ||\vec{t}||_\infty) }{|\rightarrow|} + \frac{\sum_{\vec{n} \in -} ||\vec{n}||_\infty}{|-|} \Bigg)$ and $\lambda(\rightarrow) = \frac{\sum_{(\vec{h}, \vec{t}) \in \rightarrow} ||\vec{h} - \vec{t}|| }{|\rightarrow|}$. Finally, we get $L \le \frac{1}{\beta ck\sqrt{n}} \cdot \kappa_\infty(\mathcal{S}) + \frac{1}{3\beta} \cdot \lambda(\rightarrow) + \frac{1}{\beta ck\sqrt{n}}$. \\[1\baselineskip]
\textit{Lower Bound on \(\delta\).} \\
To establish a lower bound on \(\delta\), we first define \(\mathcal{X}^*(k)\) as the set of all \(\mathcal{S}\)-approximate invariants in \(\mathcal{X}(k)\). Then, \(\delta\) is defined as:
\[
\delta = \min_{I \notin \mathcal{X}^*(k)} \min_{I^* \in \mathcal{X}^*(k)} (c(I) - c(I^*))
\]
Given that \(I^* \in \mathcal{X}^*(k)\) implies \(c(I^*) = 0\), \(\delta\) simplifies to:
\[
\delta = \min_{I \notin \mathcal{X}^*(k)} c(I)
\]
Thus, finding a lower bound on the cost of any invariant which is not an \(\mathcal{S}\)-approximate invariant will provide a lower bound on \(\delta\). Let \(I'\) be any invariant that is not an \(\mathcal{S}\)-approximate invariant, meaning there exists at least one datapoint in $\mathcal{S}$ that does not satisfy \(I'\). For our lower bound cost analysis, we assume that there is only one datapoint that does not satisfy \( I' \), as considering multiple datapoints not satisfying \( I' \) would lead to higher costs.
Then we have $3$ cases based on which type the datapoint is.

\textit{Case 1:} Suppose this point is a plus point, call this point $\vec{p}$, then we have: $c(I') = \frac{ \delta( I, \vec{p} ) }{3 |+|}$ where $\delta(I, \vec{p}) = F_{\alpha, \beta} \circ \delta_{approx}(I, \vec{p})$. If $I \equiv \lor_j \land_i \vec{w}_{ij} \cdot \vec{x} \le b_{ij}$, then $\delta_{approx}(I, \vec{p}) := \min_j \sum_i \frac{ (\vec{w}_{ij} \cdot \vec{x} - b_{ij})^+ }{||\vec{w}_{ij} ||} $.  Let the index of the cube of $I$ which is not satisfied by $\vec{p}$ be $j_0$ - hence there is at least one predicate in the cube indexed by $j_0$ which doesn't satisfy $\vec{p}$. As we are looking for lower bounds, let us say there is only one predicate which doesn't satisfy $\vec{p}$ - let the indices of this predicate be $(j_0, i_0)$. Since \( \vec{w}_{i_0j_0} \cdot \vec{x} - b_{i_0j_0} \in \mathbb{Z} \) (as \( \vec{w}_{i_0j_0}, \vec{x} \in \mathbb{Z}^n \) ,\( b_{i_0j_0} \in \mathbb{Z} \) ) the minimum positive value for \( \vec{w}_{i_0j_0} \cdot \vec{x} - b_{i_0j_0} \) is 1. Therefore,
\[
\delta_{\text{approx}}(I, \vec{p}) = \frac{1}{||\vec{w}_{i_0j_0}||} \ge \frac{1}{k\sqrt{n}}
\]
As \( \alpha > 1 \), as $\frac{1}{k\sqrt{n}} < 1$ it follows that
\[
\delta(I, \vec{p}) = F_{\alpha, \beta} \circ \delta_{\text{approx}}(I, \vec{p}) \ge \frac{1}{k\beta \sqrt{n}}
\]
Thus, \(c(I') \ge \frac{1}{3k\beta \sqrt{n} |+|}\). \\
\textit{Case 2:} Suppose the datapoint is an ICE pair $(\vec{h}, \vec{t})$, then we have: $c(I') = \frac{ \min( \delta( \neg I, \vec{h} ) , \delta( I, \vec{t} ) ) }{3 |\rightarrow|}$. Similar to the previous case, we can get $\delta(\neg I, \vec{h}) \ge \frac{1}{k\beta \sqrt{n}}$ and $\delta(I, \vec{t}) \ge \frac{1}{k\beta \sqrt{n}}$, which gives that in this case we have \(c(I') \ge \frac{1}{3k\beta \sqrt{n} |\rightarrow|}\). \\
\textit{Case 3:} Suppose the datapoint is a minus point $\vec{n}$, then we have: $c(I') = \frac{ \delta( \neg I, \vec{n} ) }{3 |-|}$. Then like before we have $\delta(\neg I, \vec{n}) \ge \frac{1}{k\beta \sqrt{n}}$, which implies that $c(I') \ge \frac{1}{3k\beta \sqrt{n} |-|}$. \\
Hence, in the general case, $c(I') \ge \frac{1}{3k\beta \sqrt{n} \max(|+|, |\rightarrow|, |-|)} \ge \frac{1}{3k\beta \sqrt{n} |\mathcal{S}|}$, which gives us $\delta \ge \frac{1}{3k\beta \sqrt{n} |\mathcal{S}|}$. \\[1\baselineskip]
Then by theorem \ref{thm: SAguarantee}, if we set $T_0 \ge  
\Big( \frac{d\rho(\mathcal{D}_{state}) + d\sqrt{n}}{\beta} + \frac{d}{\beta k \sqrt{n}} \Big)\cdot (1 + \kappa_\infty(\mathcal{S})) + \frac{cdk\sqrt{n}\rho ( \mathcal{D}_{state} ) + cdkn + cd }{3\beta} \cdot \lambda(\rightarrow)$, then we get that $T_0 \ge rL$ by the previous results. Furthermore setting $t_{max} > kcd\sqrt{n}\rho ( \mathcal{D}_{state} ) + cdkn + cd$ ensures that $t_{max} > r$, which implies that our SA search satisfies the conditions of theorem \ref{thm: SAguarantee}, and we converge to a global minima with probability at least $1 - \frac{A}{(\lfloor \frac{t_{max}}{r} \rfloor)^{exp}} \ge 1 - \frac{A}{\lfloor \frac{t_{max}}{k\sqrt{n}cd\rho(\mathcal{D}_{state}) + kncd + cd} \rfloor^{exp}}$ where $exp = \min ( a , b)$.\footnote{Note that the constant $A$ depends on the parameters of the search space. However, no general form for $A$ is provided in \cite{mitra1986convergence}, as the focus is on the rates of convergence, i.e., the value of the exponent.}

\begin{lemma}
    We have $exp \ge \min \Big(\frac{1}{(cd(2n+2))^{kcd\sqrt{n}(\rho(\mathcal{D}_{state} + \sqrt{n}) )}}, \frac{1}{|\mathcal{S}| \cdot \Big(\frac{3}{c} \kappa_\infty(\mathcal{S}) + k\sqrt{n} \lambda(\rightarrow) + \frac{3}{c} \Big) }  \Big)$
\end{lemma}
\begin{proof}
We have $a:= \frac{w^r}{r^{\frac{rL}{T_0}}} > w^r \ge \frac{1}{(cd(2n+2))^{kcd\sqrt{n}(\rho(\mathcal{D}_{state} + \sqrt{n}) )}}$ and $b := \frac{\delta}{L} \ge \frac{1}{|\mathcal{S}| \cdot \Big(\frac{3}{c} \cdot \kappa_\infty(\mathcal{S}) + k\sqrt{n} \cdot \lambda(\rightarrow) + \frac{3}{c} \Big) }$, which proves the lemma.
\end{proof}

For moderately sized datasets, we have $\frac{1}{(cd(2n+2))^{kcd\sqrt{n}(\rho(\mathcal{D}_{state} + \sqrt{n}) )}}$ to be the lower term in $exp$, wheras for very large sized datasets, eventually $\frac{1}{|\mathcal{S}| \cdot \Big(\frac{3}{c} \kappa_\infty(\mathcal{S}) + k\sqrt{n} \lambda(\rightarrow) + \frac{3}{c} \Big) }$ becomes the lower term.

And finally the global minima is an $\mathcal{S}$-approximate invariant iff a $\mathcal{S}$-approximate invariant exists in the parameterized state space (by definition of the cost function). This proves the theorem.
\end{proof}

\begin{proof}[Proof for Theorem \ref{Thm:FullSuccessGuarantees}]
    By Theorem \ref{Thm: DS Loop Conv Guar}, \toolname{} necessarily terminates within $$T = t_{\text{refine}} \log_2 \left( \epsilon_0 n^n \left( \frac{1 + \frac{C}{n^{2 - \frac{2}{n+1}}}}{1 - C} \right) \max (\Lambda(P), \Lambda(B), \Lambda(\neg Q)) \right)$$ with probability at least $\delta_0^3$ and any approximate invariant found at the end of the $T$th iteration of the DS Algorithm must be the loop invariant. Additionally Theorem \ref{thm: OurSAguarantee} guarantees that each SA search finds an approximate invariant with probability at least \( p_0(\mathcal{S}_i) := 1 - \frac{A}{\left\lfloor \frac{t_{max}}{k\sqrt{n}cd\rho(\mathcal{D}_{state}) + kncd + cd} \right\rfloor^{exp}}\) for \( exp = \min \left( \frac{1}{(cd(2n+2))^{kcd\sqrt{n}(\rho(\mathcal{D}_{state}) + \sqrt{n})}}, \frac{1}{|\mathcal{S}_i| \left(\frac{3}{c} \kappa_\infty(\mathcal{S}_i) + k\sqrt{n} \lambda(\rightarrow_i) + \frac{3}{c}\right)} \right)\). In order to reach the $T$th iteration of the DS Algorithm, all SA searches upto that point must find the approximate invariant, and hence \toolname{} finds the loop invariant with probability at least $\delta_0^3 \prod_{i=1}^T p_0(\mathcal{S}_i)$. The final result follows by observing that $\kappa_\infty (\mathcal{S}_i) \le \rho(\mathcal{D}_{state})$, $\lambda (\rightarrow_i) \le \rho(\mathcal{D}_{state})$ and $|\mathcal{S}_i| \le \tau\left( \frac{\epsilon_0}{\left\lfloor \frac{T}{t_{\text{refine}}} \right\rfloor}, \delta \right) + 3T \cdot \text{cex}_{\text{max}}$, where  \(\tau(\epsilon, \delta)\) represents the minimum number of points to sample to obtain randomized \(\epsilon\)-nets of polytopes with respect to ellipsoids with a probability of at least \(\delta\).
\end{proof}

\section{Supplementary Figures.} \label{App: SuppFigures}

\subsection{Efficiently computable Approximate Polytope Set Distance.}
We illustrate the difference between the true polytope set distance and our approximate polytope set distance in Figure \ref{fig:costFuncPlot}.

\begin{figure}
    \centering
    \includegraphics[scale = 0.75]{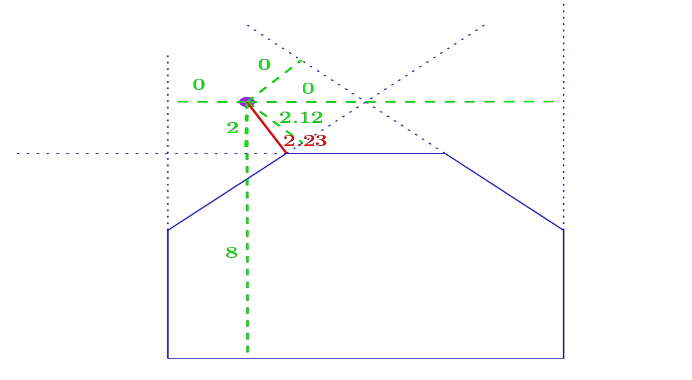}
    \caption{In the figure above, the boundary of the polytope is depicted in blue. Our objective is to compare the polytope set distance and the approximate polytope set distance for the point indicated in purple. The dotted green lines illustrate the distances from the point to the six boundaries of the polytope, listed clockwise as 0, 0, 0, 2, 8, and 2.12. Consequently, the approximate polytope set distance for the point in question is calculated as the mean of these distances, yielding $\frac{0 + 0 + 0 + 2 + 8 + 2.12}{6} = 2.02$. The true polytope set distance for this point is represented in red and measures $2.23$. This comparison demonstrates that our approximate set distance is quite accurate in this instance. }
    \label{fig:costFuncPlot}
\end{figure}

\subsection{Coefficient Transitions cause higher cost jumps in SA search.}
Figure \ref{fig:costFunctionMotivate2} shows why coefficient transitions cause larger cost jumps than constant transitions.

\begin{figure}
    \centering
    \includegraphics[scale=0.45]{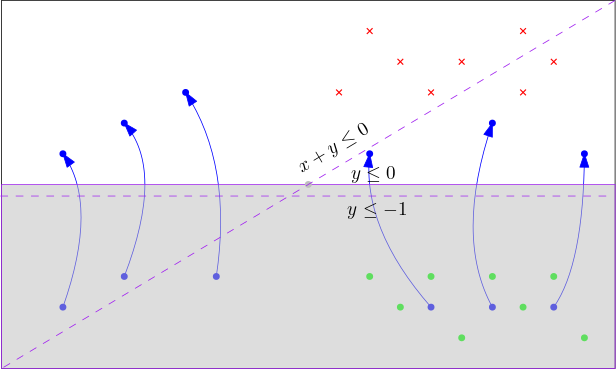}
    \caption{The plot above (not to scale) illustrates the issue of coefficient transitions significantly altering the average set distance of the dataset to the candidate invariant. The black box marks the state space \(\mathcal{D}_{state}\). The green dots represent positive points, the red crosses indicate negative points, and the blue arrows show the implication pairs. The current candidate invariant is indicated by the gray region, defined by a single predicate: \(y \leq 0\). A few possible transitions, including a constant transition to \(y \leq -1\) and a coefficient transition to \(x + y \leq 0\), are shown with dashed lines. The plot demonstrates how coefficient transitions unlike constant transitions drastically change the boundary line of the invariant, thereby significantly altering the average set distance of the dataset to the candidate invariant. }
    \label{fig:costFunctionMotivate2}
\end{figure}

\subsection{Normalization Function.}
Fig \ref{fig:normalizer} plots $F_{\alpha,\beta}$ for specific values of $\alpha$ and $\beta$.

\begin{figure}
    \centering
    \includegraphics[scale = 0.15]{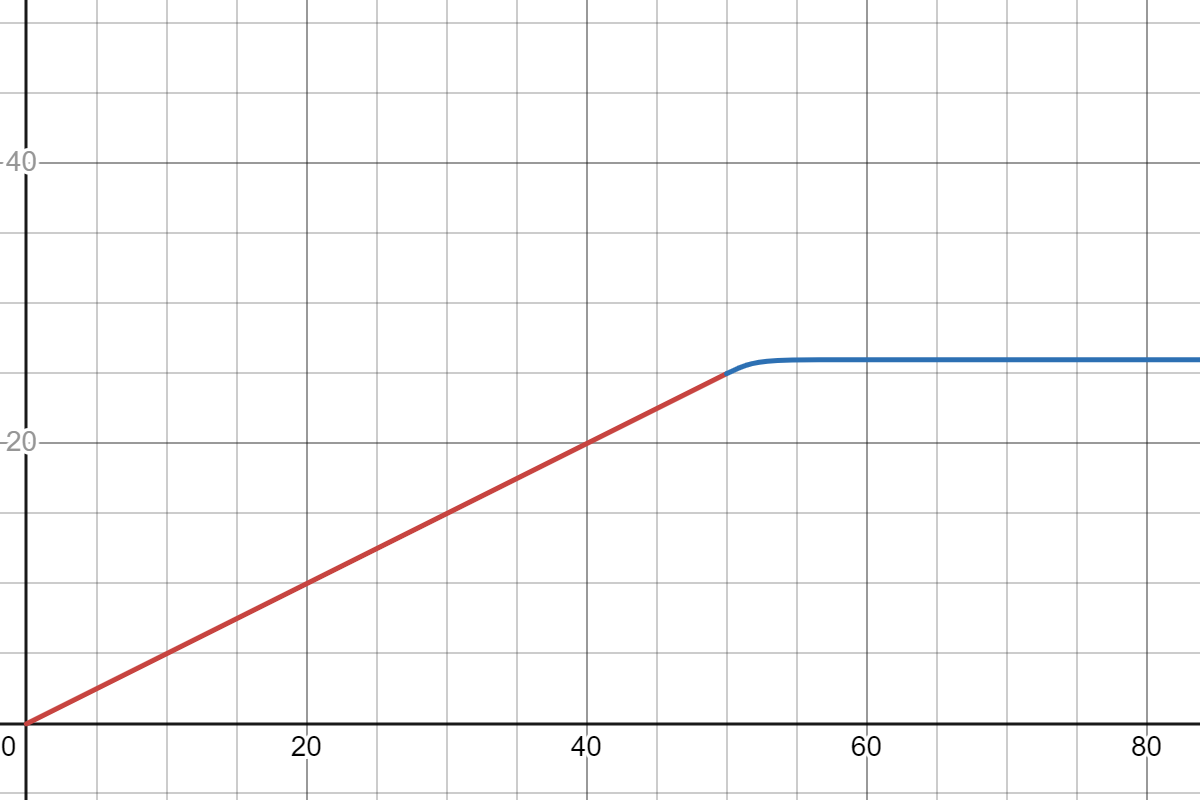}
    \caption{The plot shows the normalizer function \(F_{\alpha, \beta}\) for \(\alpha = 50\) and \(\beta = 2\). The red portion of the curve highlights the linear part of the normalizer, while the blue portion of the curve highlights the sigmoid-like approximation of the normalization function.}
    \label{fig:normalizer}
\end{figure}

\subsection{\toolname{} framework.}
Figure \ref{fig:Implementation} gives the framework for our tool.
\begin{figure}
    \centering
    \includegraphics[scale=0.35]{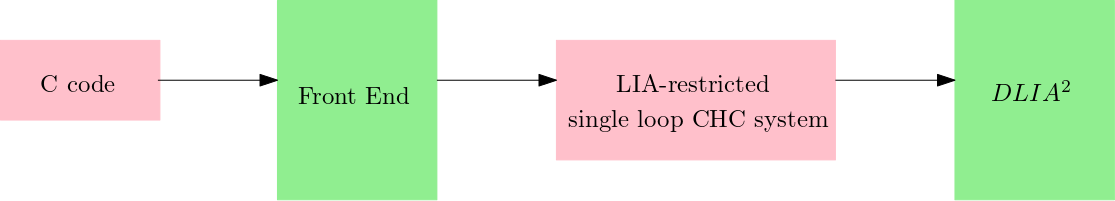}
    \caption{Diagram showing end-to-end framework for \toolname{}. }
    \label{fig:Implementation}
\end{figure}

\subsection{Time taken for DS Loop vs SA Loop.} \label{App: timeDSSA}
We present the computation times for the DS Algorithm and SA Algorithm over the converged benchmarks (best runs), depicted in Figure \ref{fig:SAvsDS}. Each blue dot in the figure represents a converged benchmark in Table~\ref{tab:results} (left), plotting the total DS Loop time against the total SA time required for \toolname{}. The red dotted line corresponds to $y = x$. The plot illustrates that the blue dots predominantly lie below this line, indicating that SA runs constitute the majority of \toolname{}'s computational workload.
\begin{figure}
    \centering
    \includegraphics[width=0.5\linewidth]{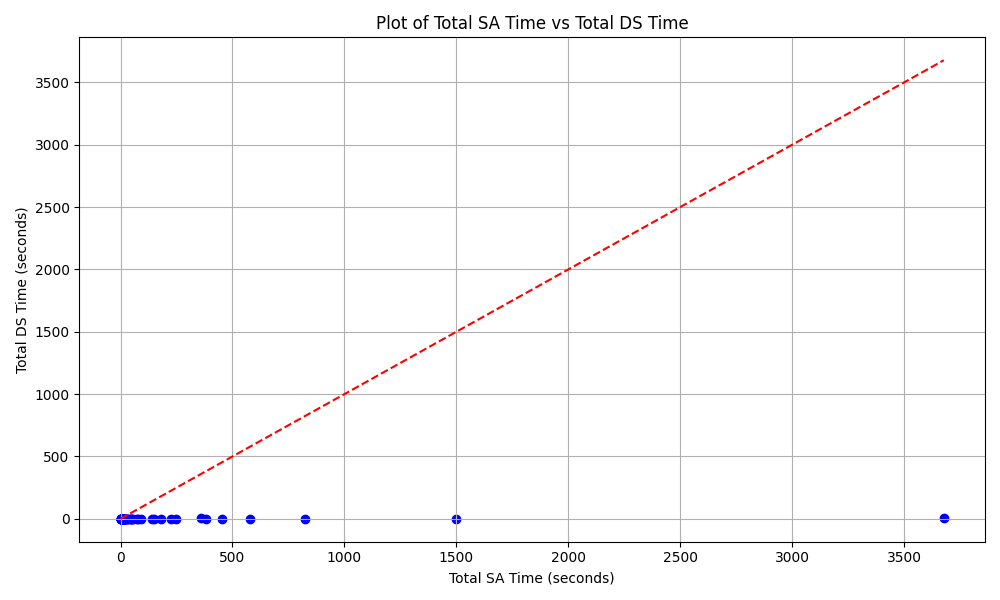}
    \caption{The plot above shows the time taken by the simulated annealing loop vs. the time taken by the datapoint sampling loop for successful searches of \toolname{} on the benchmarks (best runs).}
    \label{fig:SAvsDS}
    \vspace{-1.0 em}
\end{figure}

\end{document}